\documentclass[12pt]{article}
\usepackage{amsmath,amsthm,amssymb,amsfonts,bm,mathtools}
\usepackage{geometry,setspace,natbib,graphicx,booktabs,enumitem,microtype}
\usepackage{xcolor,hyperref}
\usepackage{xr-hyper}
\usepackage{appendix,longtable,array}
\usepackage{tikz}
\usepackage{float}
\usetikzlibrary{arrows.meta,positioning,shapes.geometric}
\hypersetup{colorlinks=true,citecolor=blue,linkcolor=blue,urlcolor=blue,hypertexnames=false,pdfauthor={Ulrich Hounyo},pdftitle={Identification and Information after Nuisance Projection}}

\newtheorem{assumption}{Assumption}
\newtheorem{theorem}{Theorem}
\newtheorem{proposition}{Proposition}
\newtheorem{corollary}{Corollary}
\newtheorem{lemma}{Lemma}
\newtheorem{remark}{Remark}
\newtheorem{example}{Example}
\newtheorem{definition}{Definition}
\newtheorem*{principle}{Projected Information Principle}
\newcommand{\E}{\mathbb E}

\newcommand{\Cov}{\operatorname{Cov}}

\newcommand{\rank}{\operatorname{rank}}

\newcommand{\Var}{\operatorname{Var}}

\title{Identification and Information after Nuisance Projection}
\author{
Ulrich Hounyo\thanks{%
Department of Economics, University at Albany -- State University of New
York, Albany, NY 12222, USA. E-mail: \texttt{khounyo@albany.edu}. }\\
Department of Economics\\
University at Albany, SUNY
}
\date{August 2, 2026}

\begin{document}
\maketitle

\begin{abstract}
Empirical work often removes fixed effects, latent factors, or high-dimensional controls
before estimating structural relationships. These transformations reduce confounding but
may also remove identifying variation. We study linear panel IV after one
equation-compatible nuisance projection under two-way dependence. The projected Jacobian
determines which structural directions remain visible; the projected-score law determines
their precision; and, on Gaussian fixed-rank strata, they combine in a Projected
Information Matrix. We derive weak-identification limits with dimension-specific
information accumulation, feasible factor-transfer conditions, identification-robust
tests, bootstrap procedures for non-Gaussian interaction limits, and inference for the
projected spectrum, rank, subspaces, and information matrix. Simulations show that a raw
first-stage statistic above 500 can support the wrong sign while projected diagnostics
reveal weak valid information. In an international monetary application, common
projection substantially attenuates apparent foreign-output persistence, while
Gaussian-reference Anderson--Rubin sets remain unbounded. Identification should therefore
be assessed after nuisance removal.
\end{abstract}


\section{Introduction}\label{sec:intro}

Empirical credibility often requires removing variation. Fixed effects eliminate persistent
heterogeneity, latent-factor methods absorb common shocks, and high-dimensional or
orthogonalized procedures partial out rich nuisance components
\citep{bai2009,pesaran2006,cattaneojanssonma2019,chernozhukov2018,chernozhukovetal2022}.
The same transformations can also remove the variation that distinguishes competing
economic mechanisms. A specification may therefore become less confounded but also less
informative. Standard practice is well equipped to estimate after nuisance removal, yet
offers much less guidance on whether enough identifying information remains.

A two-equation supply--demand example gives the basic intuition. Suppose price is
instrumented by a cost shifter, but both variables contain a common aggregate shock.
Removing that shock can restore the exclusion restriction; it can also remove most of
the cost variation that moves price. The raw first stage may therefore look strong even
though the final, credible specification contains little identifying variation. The same
problem arises when global factors absorb monetary-policy exposure, aggregate credit
conditions absorb regulatory variation, or rich industry--time controls nearly absorb a
firm-level instrument. A large nominal sample can then contain only a few economically
distinguishable directions.

This paper studies the resulting credibility-information trade-off through one central
question: \emph{what information remains after nuisance projection?} The relevant issue is
not whether an instrument was strong in the raw data, but whether the moments used in the
final structural analysis remain sensitive to the parameter after the same nuisance
transformation and under the same dependence structure. This distinction can reverse an
estimated sign, eliminate an apparently persistent policy response, or turn a seemingly
precise conclusion into an unbounded identification-robust confidence set. The answer determines which
policy or behavioral mechanisms the final specification can distinguish and, therefore,
which counterfactual conclusions the evidence can support. A robust covariance estimator
can measure uncertainty in surviving variation, but it cannot recreate variation that the
specification has removed.

\begin{center}
\fbox{\begin{minipage}{0.90\textwidth}
\begin{principle}
Identification and statistical information must be evaluated after the final nuisance
transformation, using the moments and dependence structure that remain.
\end{principle}
\end{minipage}}
\end{center}

Projection changes the moment map, score law, identified directions, and information
rates. The workflow is to apply one equation-compatible projection, compute the projected
Jacobian and score law, and report the PIM, effective rank, and robust inference; see
Online Appendix Figure~\ref{fig:conceptual-workflow}.

Existing methods address different pieces of this problem. Interactive-effects methods
remove latent common components \citep{bai2003,bai2009,pesaran2006,baili2014,moonweidner2017};
weak-IV and weak-GMM methods characterize local-to-zero asymptotics and develop
identification-robust inference
\citep{dufour1997,staigerstock1997,stockwright2000,kleibergen2002,moreira2003,andrewsmikusheva2022};
and multiway covariance and bootstrap methods approximate uncertainty under common-shock
dependence \citep{cameron2011,davezies2021,menzel2021,chiang2024,hounyolin2026}. The joint
question remains: after the nuisance transformation used in the final specification,
which economic directions are still distinguishable, how quickly does their information
accumulate, and how uncertain is that surviving information?

The failure mechanism is exact. The nuisance transformation changes both the derivative
of the moment map and the law of its score. A conventional first-stage statistic may be
computed before that transformation, while a robust covariance estimator changes only
the estimated uncertainty. Neither object measures the sensitivity and precision of the
moments actually used in the final specification, and neither can regularize an IV
denominator that remains random at first order. Identification diagnostics must therefore
be constructed from the same projected moments and dependence structure as the
structural analysis.

This paper develops such a framework for linear panel IV. Let $i=1,\ldots,N$ index units
and $t=1,\ldots,T$ index periods. After a common oracle nuisance projection, let
$y_{it}^{o}$ denote the scalar outcome, $x_{it}^{o}\in\mathbb R^p$ the endogenous
regressors, and $z_{it}^{o}\in\mathbb R^q$ the instruments. The structural coefficient is
$\beta_0\in\mathbb R^p$. The projected population moment map and its Jacobian are
\[
 \mathcal Q_{NT}(\beta)
 =\frac{1}{NT}\sum_{i=1}^{N}\sum_{t=1}^{T}
 \E\!\left[z_{it}^{o}
 \bigl(y_{it}^{o}-x_{it}^{o\prime}\beta\bigr)\right],
 \qquad
 J_{NT}=\frac{1}{NT}\sum_{i,t}\E[z_{it}^{o}x_{it}^{o\prime}].
\]
For the linear model,
\[
 \mathcal Q_{NT}(\beta)=-J_{NT}(\beta-\beta_0).
\]
Thus the rank and right singular subspaces of $J_{NT}$ determine which combinations of
structural coefficients remain visible after nuisance removal. The law of the projected
score determines how precisely those combinations are observed. Projection changes the
source of identification; dependence changes the precision with which the surviving
source is measured.

The first structural contribution is to make the nuisance transformation
\emph{equation compatible}. Applying unrelated final factor projections to the outcome,
regressors, and instruments generally destroys the common transformed equation. We give a
transparent condition under which one joint two-sided projection preserves both the
structural coefficient and the projected reduced form. This distinguishes the paper from
factor-adjusted IV procedures that establish consistency but do not study whether the
same nuisance removal changes the strength or dimensionality of identification.

The second contribution is a joint projected-score experiment under asymmetric
information accumulation. For a local displacement $h\in\mathbb R^p$, a normalized
projected moment statistic converges to
\[
 Y(h)=Z+\mathcal Hh,
\]
where $Z\in\mathbb R^q$ is the limiting projected score and
$\mathcal H\in\mathbb R^{q\times p}$ is the limiting normalized projected Jacobian. This is a finite-dimensional
shift representation for the projected moments; it is not a claim of local asymptotic
normality for the full-data likelihood. When $Z$ is Gaussian with covariance $\Omega$, the covariance-adjusted local
sensitivity
\[
 \mathcal I_P=\mathcal H'\Omega^\dagger\mathcal H
\]
is the Projected Information Matrix, where $\Omega^\dagger$ is the Moore--Penrose inverse.
It plays for the projected moment experiment the role that Fisher information plays in a
regular likelihood experiment and that $G'\Omega^{-1}G$ plays in regular GMM
\citep{hansen1982,chamberlain1987,neweymcfadden1994}. This is an analogy at the level of
the finite-dimensional projected experiment, not a claim of full-likelihood LAN. The
novelty is the integrated theory for this object after an information-reducing
projection, with heterogeneous normalizations, estimated latent factors, two-way
dependence, weak identification, and potentially non-Gaussian score limits.

Economically, $\mathcal I_P$ describes the information that remains after the controls
and common components used to make the specification credible have been removed. Its
null space contains structural directions that the projected moments cannot distinguish.
Its eigenvalues measure local information along economically meaningful combinations of
parameters, and $h'\mathcal I_Ph$ is the noncentrality parameter governing Gaussian local
power in direction $h$. A small eigenvalue can therefore reflect little surviving policy
variation even in a large panel, while a large condition number can reveal that some
economic margins are much more informative than others.

If $h_{NT}'\mathcal I_{P,NT}h_{NT}\to0$, Gaussian local power converges to size
because that direction has been projected out; see Appendix
Proposition~\ref{prop:near-orthogonal-power}.

\paragraph{Reader's guide.}
$J_{NT}$ and $\mathcal H$ measure surviving moment sensitivity; $\Omega$ measures score
noise; and, on Gaussian fixed-rank strata,
$\mathcal I_P=\mathcal H'\Omega^\dagger\mathcal H$ combines the two. Their ranks and
spectra determine which structural directions remain distinguishable. The paper develops
four corresponding components: projected identification, information, inference, and
feasible implementation. A schematic is in Online Appendix
Figure~\ref{fig:economic-motivation}.

The third contribution concerns how information accumulates across panel dimensions.
Classical weak-IV asymptotics balance deterministic first-stage drift against score noise
\citep{staigerstock1997,stockwright2000}. The former supplies the canonical local-to-zero
IV normalization, while the latter extends the logic to weak GMM. After projection, that
balance need not involve
$\sqrt{NT}$. To see why, consider the scalar projected first stage
\[
 x_{it}^{o}=\pi_{NT}z_{it}^{o}+v_{it}^{o},
\]
where $\pi_{NT}$ is the projected first-stage coefficient and $v_{it}^{o}$ is its
projected residual. Define
\[
 \zeta_{i,NT}=T^{-1}\sum_{t=1}^{T}z_{it}^{o}v_{it}^{o},
 \qquad
 \Omega_{v,NT}=\Var\!\left(N^{-1/2}\sum_{i=1}^{N}\zeta_{i,NT}\right).
\]
If $T^\delta\Omega_{v,NT}\to\Omega_v\in(0,\infty)$ for an
information-accumulation index $\delta\in[0,1]$, the local balance gives
\[
 \pi_{NT}=O\!\left((NT^\delta)^{-1/2}\right).
\]
When $\delta=1$, both dimensions contribute in the familiar way. When $\delta=0$, more
time periods can improve estimation of common components without adding first-order
independent first-stage information. The relevant rate is determined by where the
instrument varies and by which movements survive projection, not mechanically by the
number of cells.

This dimension-specific logic extends to multiple endogenous regressors. Identification
is directional: one combination of policy coefficients can be well supported while
another is weakly or only partially identified. Singular values quantify surviving
projected relevance, invariant singular subspaces describe the combinations of
coefficients supported by the data, and spectral gaps determine whether those
combinations are stable enough to interpret. We develop confidence intervals for
separated singular values, rank-consistent thresholding on separated strata, subspace
stability measures, and inference for $\mathcal I_P$ itself. These procedures quantify
uncertainty about identification rather than conditioning silently on identification
being strong.

Feasibility creates an additional requirement. Factor estimates that are accurate for
prediction need not be accurate enough for weak-identification inference. We prove a
sample-split strong-factor benchmark and state general transfer conditions under which
estimated nuisance spaces preserve the structural score, first-stage score, Jacobian, and
relevant spectral gaps at their own normalizations. The operative comparison is therefore
between factor-estimation error and the scale of the surviving information.

Inference separates identification robustness from dependence robustness. Under the
structural null, the Anderson--Rubin score removes the weak first-stage denominator and
remains centered regardless of projected identification strength. Its critical values
must nevertheless approximate the post-projection score law. Gaussian fixed-rank limits
justify chi-square critical values, but unit--time interaction components can generate
non-Gaussian limits. We provide a self-contained bootstrap for a primitive finite-range
benchmark and establish feasible transfer of the broader PWB-H procedure of
\citet{hounyolin2026} under its exact oracle conditions. The bootstrap therefore targets
the uncertainty generated by the empirical design after projection, rather than only a
covariance matrix.

The paper makes three contributions. First, it develops an equation-compatible projected
moment framework and a primitive joint structural- and first-stage-score limit under
asymmetric two-way dependence. Second, it derives scalar and multivariate
weak-identification limits with dimension-specific information accumulation and unifies
them through a local projected-moment experiment. Third, it provides feasible factor
transfer, identification- and dependence-robust inference, and uncertainty measures for
the projected spectrum, rank, subspaces, and PIM. Standard rank algebra, matrix
perturbation, regular GMM, and Gaussian quadratic-form arguments serve as supporting
tools.

The evidence changes substantive conclusions. In the weak-information simulation, a confounding common factor produces a median raw
first-stage statistic above 500 and a median IV estimate with the wrong sign;
equation-compatible projection moves the estimate toward the truth while the projected
Jacobian and PIM reveal how the surviving information differs from the raw association. In the international monetary application, conventional fixed-effects estimates
imply persistent foreign output contractions, whereas the projected estimates move toward
zero and become positive at five years. The projected Jacobian falls even as the
first-stage statistic and PIM rise because score variability falls more sharply, and
Gaussian-reference Anderson--Rubin sets remain unbounded. Coefficient paths, relevance,
moment sensitivity, and covariance-adjusted information therefore answer different
questions.

\paragraph{Implications for empirical practice.}
Researchers should evaluate relevance after the final nuisance transformation and under
the dependence structure used for inference. Applications should report surviving
instrument variation, the projected Jacobian spectrum, score uncertainty, the PIM or
full score law where appropriate, effective rank, and identification-robust confidence
sets. These diagnostics complement rather than establish instrument validity.

\subsection*{Literature positioning}

The paper connects factor methods
\citep{bai2003,bai2009,pesaran2006,baili2014,moonweidner2017}, dynamic-panel IV
\citep{arellanobond1991,blundellbond1998}, weak-identification and robust-inference theory
\citep{andersonrubin1949,dufour1997,staigerstock1997,stockwright2000,kleibergen2002,moreira2003,andrewsmikusheva2022},
uniform-size analysis \citep{andrewschengguggenberger2020}, orthogonal-score methods
\citep{cattaneojanssonma2019,chernozhukov2018,chernozhukovetal2022}, and multiway
dependence \citep{cameron2011,davezies2021,menzel2021,chiang2024,juodis2025,hounyolin2026}.
Its distinct question is how identification and score information change after the final
common nuisance projection. Regular GMM geometry
\citep{hansen1982,chamberlain1987,neweymcfadden1994} is used within that experiment.

\section{Primitive framework and the projected identification map}\label{sec:model}

The purpose of this section is to define the final empirical specification on which
identification must be judged. We first separate observed-control residualization from
latent-factor removal and then state the compatibility condition that allows one common
projection to preserve the structural coefficient. The econometric requirement is:
removing confounding variation is useful only if the resulting moments still refer to
the same structural equation.

We observe a balanced triangular array
\[
 \mathcal W_{NT}=\{(y_{it},x_{it},z_{it},w_{it}):1\le i\le N,\ 1\le t\le T\},
\]
where $y_{it}$ is scalar, $x_{it}\in\mathbb R^p$ is endogenous, $z_{it}\in\mathbb R^q$ is excluded from the structural equation, $q\ge p$, and $w_{it}$ contains included exogenous controls. Both panel dimensions diverge jointly, but no proportionality restriction is imposed. The amount of identifying information is not assumed to be proportional to $NT$; it is derived from the projected reduced-form moments generated by the primitive array.

\paragraph{Notation, topology, and projection terminology.}
The dimensions $p$, $q$, the observed-control dimension $d_w$, and all maintained
factor ranks are fixed as $N,T\to\infty$. Structural and local parameters are finite
dimensional and use the Euclidean topology; matrices use the spectral norm unless
$\|\cdot\|_{\mathsf F}$ is displayed. For a matrix
$A$, $A'$ is transpose, $A^\dagger$ the Moore--Penrose inverse, and
$\mathcal R(A)$ and $\mathcal N(A)$ its range and null space. The operators
$\operatorname{vec}$ and $\operatorname{vech}$ stack all entries and the lower-triangular
entries of a symmetric matrix. We write $\mathbb S^m$ and $\mathbb S_+^m$ for symmetric
and positive-semidefinite $m\times m$ matrices.

Three meanings of ``projection'' are kept distinct. The nuisance residual operator
$\mathcal M_0:\mathbb R^{N\times T}\to\mathbb R^{N\times T}$ is the
Frobenius-orthogonal projection onto
\[
 \mathcal V_{NT}
 =\{B:\mathcal R(B)\subseteq\mathcal L_0^\perp,\
       \mathcal R(B')\subseteq\mathcal F_0^\perp\},
\]
so
$\mathcal M_0(A)=\arg\min_{B\in\mathcal V_{NT}}\|A-B\|_{\mathsf F}^2
=M_{\mathcal L_0}AM_{\mathcal F_0}$.
Spectral projectors are orthogonal projectors onto singular or eigenspaces; projection of
a confidence region onto a subvector is a set-theoretic image. The maintained projection is a linear residual operator, not a profiled criterion.

All asymptotics concern the triangular array $N,T\to\infty$ jointly. Local paths are
$\beta_{NT}(h)=\beta_0+D_{NT}^{-1}h$, where $D_{NT}$ is deterministic and nonsingular and
$h$ ranges over a fixed compact subset of $\mathbb R^p$. Expectations and probabilities
are under the triangular-array law; $\Rightarrow$, $\to_p$, $O_p$, and $o_p$ have their
usual meanings, and $O_{p^*}$ and $o_{p^*}$ are conditional bootstrap orders in outer
probability. A superscript $E$ denotes a population-average moment.

\subsection{A single primitive triangular-array representation}\label{subsec:unifieddgp}

For each $(N,T)$, all random elements are defined on a probability space $(\Omega_{NT},\mathcal F_{NT},\mathbb P_{NT})$. Let $A_i$ denote a unit-indexed latent state, $D_t$ a time-indexed latent state, and $E_{it}$ a cell innovation, and let
\[
 \mathcal G_{it,NT}=\sigma(A_i,D_t,E_{it})\subseteq\mathcal F_{NT}.
\]
This is a local information $\sigma$-field, not a time filtration; no martingale ordering of the two panel indices is imposed. Each observable at cell $(i,t)$ is $\mathcal G_{it,NT}$-measurable. For measurable
state spaces $\mathsf A_{NT}$, $\mathsf D_{NT}$, and $\mathsf E_{NT}$,
\begin{equation}
 (y_{it},x_{it}',z_{it}',w_{it}')'=h_{NT}(A_i,D_t,E_{it}),\qquad
 h_{NT}:\mathsf A_{NT}\times\mathsf D_{NT}\times\mathsf E_{NT}
 \to\mathbb R^{1+p+q+d_w},
 \label{eq:primitive-map}
\end{equation}
where $h_{NT}$ is measurable. The sigma-fields generated by the three collections $\{A_i\}_i$, $\{D_t\}_t$, and $\{E_{it}\}_{i,t}$ are mutually independent. Dependence within $\{A_i\}_i$ and within $\{D_t\}_t$ is permitted; conditional on the complete unit- and time-state arrays, the cell innovations are independent across $(i,t)$. This representation simultaneously accommodates spatially dependent unit heterogeneity, serially dependent aggregate shocks, and idiosyncratic cell variation.

Let $Y=(y_{it})$ be the $N\times T$ outcome matrix, and let $X_j=(x_{it,j})$
and $Z_k=(z_{it,k})$ denote the corresponding matrices for regressor $j\le p$ and
instrument $k\le q$. The structural parameter $\beta_0\in\mathbb R^p$ is defined by
the projected moment restriction developed below. Observed controls are first partialled
out from $Y$, every $X_j$, and every $Z_k$ by the same orthogonal residual operator onto
the complement of the vectorized control span; a Moore--Penrose inverse is used if the
control design is not of full column rank. This step is conceptually separate from
low-rank factor removal. To avoid an unnecessary layer of notation, the
control-residualized variables are again denoted by $(y_{it},x_{it},z_{it})$. The
maintained linear-IV specialization is
\begin{align}
 y_{it}&=x_{it}'\beta_0+\alpha_i+\tau_t+\lambda_i'f_t+u_{it}, \label{eq:y}\\
 x_{it}&=\Pi_{NT}'z_{it}+c_i+s_t+\Gamma_i'g_t+v_{it}. \label{eq:x}
\end{align}
The coefficient matrix $\Pi_{NT}\in\mathbb R^{q\times p}$ may drift with
$(N,T)$. Here $u_{it}$ is scalar; $c_i$, $s_t$, and $v_{it}$ are $p$-vectors;
$\lambda_i,f_t\in\mathbb R^{r_y}$; and, for a fixed first-stage factor rank $r_x$,
$\Gamma_i\in\mathbb R^{r_x\times p}$ and $g_t\in\mathbb R^{r_x}$, so that
$\Gamma_i'g_t$ is a $p$-vector. The ranks $r_y$ and $r_x$ are fixed in the oracle
formulation. Equations \eqref{eq:y}--\eqref{eq:x} are maintained restrictions; they do
not follow from the measurable representation \eqref{eq:primitive-map}. The primitive
representation organizes dependence, while the two equations impose the structural and
reduced-form content used for identification.
For each instrument component $k\le q$, the corresponding $N\times T$ matrix is also assumed to admit
\[
 Z_k=N_{Z_k}+Z_k^{\sharp},
\]
where $N_{Z_k}$ is an additive or interactive nuisance matrix and $Z_k^{\sharp}$ is the excluded component retained after oracle projection. This decomposition does not impose exogeneity by itself; exclusion is stated separately below. Its role is to make explicit which nuisance components of the instruments must belong to the common loading and factor spans.

Let $\mathcal L_0\subseteq\mathbb R^N$ and $\mathcal F_0\subseteq\mathbb R^T$ be the smallest linear subspaces containing, respectively, every loading vector and every factor vector appearing in the additive or interactive nuisance matrices of the maintained outcome, first-stage, and instrument equations. In particular, $\mathbf 1_N\in\mathcal L_0$ and $\mathbf 1_T\in\mathcal F_0$. Let $P_{\mathcal L_0}$ and $P_{\mathcal F_0}$ be their orthogonal projectors and define the common two-sided oracle residual operator
\[
 \mathcal M_0(A)=M_{\mathcal L_0}AM_{\mathcal F_0},\qquad
 M_{\mathcal L_0}=I_N-P_{\mathcal L_0},\quad
 M_{\mathcal F_0}=I_T-P_{\mathcal F_0}.
\]
The maintained nuisance specification requires each nuisance matrix to admit a finite decomposition
$N=\sum_r L_rF_r'$ such that, term by term, either
$\operatorname{col}(L_r)\subseteq\mathcal L_0$ or
$\operatorname{col}(F_r)\subseteq\mathcal F_0$. Requiring both inclusions is sufficient but stronger than necessary. This one-sided-or-two-sided span condition, rather than an informal ``union'' of factor spaces, is what makes the common annihilation exact.

\paragraph{Purpose.}
The first result establishes the algebra needed for the entire framework. It gives a
transparent sufficient condition under which one common two-sided projection removes the
maintained nuisance components without changing the structural coefficient or the
projected reduced form.

\begin{lemma}[Projection compatibility]\label{lem:projectioncompat}
If $N=\sum_{r=1}^{R}L_rF_r'$ and, for every $r$, either $\operatorname{col}(L_r)\subseteq\mathcal L_0$ or $\operatorname{col}(F_r)\subseteq\mathcal F_0$, then $\mathcal M_0(N)=0$. Consequently, if every nuisance term in the outcome, first-stage, and instrument equations satisfies this condition, one common two-sided projection preserves the structural and reduced-form equations exactly.
\end{lemma}

\noindent\emph{Proof.} See Online Appendix~\ref{app:proof:lem-projectioncompat}.

\begin{remark}[Projection stability and structural rank]\label{rem:projection-stability}
The oracle projector depends on the nuisance spaces, not on $J_{NT}$; structural rank may
therefore vanish while $\mathcal M_0$ remains well defined. Uniform feasible projection
requires fixed nuisance ranks and separating eigengaps. Complete qualifications are in
Online Appendix~\ref{app:projection-stability}.
\end{remark}

\paragraph{Interpretation.}
A projection is not innocuous merely because it removes fitted nuisance components.
Using one final operator across the outcome, regressors, and instruments is what preserves
the common structural coefficient. Variable-specific residualization can change the
equation that the moments are intended to identify.

\paragraph{Why a common projection is essential.}
The structural equation is a relation among several matrices with the same coefficient $\beta_0$. Applying different final residual operators to $Y$ and to the columns of $X$ generally changes that relation. The common operator avoids this problem because linearity gives
$\mathcal M_0(X\beta_0)=\{\mathcal M_0(X)\}\beta_0$.

For the outcome matrix, each regressor matrix, and each instrument matrix, write
\[
 y_{it}^{o}=(\mathcal M_0Y)_{it},\qquad
 x_{it,j}^{o}=(\mathcal M_0X_j)_{it},\qquad
 z_{it,k}^{o}=(\mathcal M_0Z_k)_{it}.
\]
Linearity of $\mathcal M_0$, its common use across all terms, and the stated span condition imply the exact oracle equations
\begin{equation}
 y_{it}^{o}=x_{it}^{o\prime}\beta_0+u_{it}^{o},\qquad
 x_{it}^{o}=\Pi_{NT}'z_{it}^{o}+v_{it}^{o}.                     \label{eq:oracle}
\end{equation}
This algebra would generally fail under variable-specific final projections. Such projections may be used to estimate the union spaces, but the moments analyzed below use one common final operator. If observed controls are not partialled out first, they must instead be included in an explicit common vector-space residual maker; a two-sided factor projection alone does not annihilate unrestricted controls. The projected moment restriction is therefore the maintained identifying restriction of the transformed model. It should not be inferred mechanically from an unprojected exclusion condition when the projection is estimated from the same outcomes or disturbances; cross-fitting, conditional exclusion, or a direct projected exclusion assumption is needed.

\subsection{Maintained conditions and probabilistic scope}

\begin{assumption}[Projected population orthogonality]\label{ass:score}
At the true parameter,
\[
 \frac1{NT}\sum_{i,t}\E[z_{it}^{o}u_{it}^{o}]=0,
 \qquad
 \frac1{NT}\sum_{i,t}\E[z_{it}^{o}v_{it}^{o\prime}]=0
\]
for every $(N,T)$. The first restriction is projected structural exclusion. When $Q_{zz,NT}^{E}$ is nonsingular, the second restriction makes
$\Pi_{NT}$ the coefficient of the population-average linear projection of
$x_{it}^{o}$ on $z_{it}^{o}$. Cellwise orthogonality is a stronger sufficient condition, not a maintained requirement. Neither restriction excludes correlation between $u_{it}^{o}$ and $v_{it}^{o}$. Because the final projection may be data dependent, primitive pre-projection exclusion implies these restrictions only under additional conditions, such as deterministic nuisance spaces, suitable conditional exclusion, or cross-fitting.
\end{assumption}

\begin{remark}[Population-average versus cellwise exclusion]
Assumption~\ref{ass:score} is sufficient for the population-average moment map used in
the paper. It does not assert cellwise conditional exogeneity. Consequently, primitive
CLTs and bootstrap results must be verified for the aggregated projected score rather
than inferred from a stronger unmaintained cellwise restriction.
\end{remark}

\begin{assumption}[Projected instrument second-moment stability]\label{ass:moments}
With
$Q_{zz,NT}^{E}=(NT)^{-1}\sum_{i,t}\E[z_{it}^{o}z_{it}^{o\prime}]$,
\[
 \widehat Q_{zz,NT}=\frac1{NT}\sum_{i,t}z_{it}^{o}z_{it}^{o\prime},\qquad
 \|\widehat Q_{zz,NT}-Q_{zz,NT}^{E}\|\to_p0,\qquad
 Q_{zz,NT}^{E}\to Q_{zz},
\]
where the eigenvalues of $Q_{zz}$ are bounded away from zero and infinity. Moment,
mixing, and tail conditions sufficient for a particular central limit theorem, factor
estimator, or bootstrap procedure are imposed only in the corresponding result.
\end{assumption}

\paragraph{Probabilistic scope.}
The measurable representation \eqref{eq:primitive-map} is an organizing device for the
primitive dependence benchmark, not a global requirement for the algebraic or high-level
results. The algebraic identification results below require no mixing, independence, or
central limit theorem. Distributional results may use either the high-level joint convergence in Assumption~\ref{ass:primitiveclt} or the primitive finite-range benchmark in Theorem~\ref{thm:primitive-benchmark}. The benchmark covers fixed-range dependence in the unit and time states, a finite-rank interaction projection, and conditionally independent cell innovations. More general spatial mixing, temporal mixing, and infinite-rank interactions require additional summability exponents, variance orders, and truncation arguments; the Online Appendix records those ingredients.

\paragraph{Asymptotic conventions.}
Both $N$ and $T$ diverge jointly. Each theorem states the deterministic scalar or
matrix normalization it uses. Moore--Penrose inverses enter inferential results only on
fixed-rank strata separated from rank boundaries.
On such strata the pseudoinverse is continuous. Ridge regularization changes the target
quadratic form and is not used for inference.

\paragraph{How the assumptions are used.}
Assumption~\ref{ass:score} defines projected population validity, while
Assumption~\ref{ass:moments} supplies only the instrument second-moment law needed by
results that use $\widehat Q_{zz,NT}$.  Theorem~\ref{thm:identification-map} is algebraic and does not use a central limit theorem. Assumption~\ref{ass:primitiveclt}, Theorem~\ref{thm:primitive-benchmark}, or the more specialized Assumption~\ref{ass:joint} supplies the distributional score experiment used in the weak-IV limits. Theorem~\ref{thm:pcbenchmark} verifies factor transfer for the sample-split strong-factor benchmark; Assumption~\ref{ass:factororth} retains a more general high-level route. Assumption~\ref{ass:regular-eff} is confined to the regular-GMM section, and Assumption~\ref{ass:boot} is confined to the bootstrap transfer. This separation prevents a high-level probabilistic condition from being mistaken for a primitive theorem.

\paragraph{A useful diagnostic distinction.}
Assumption~\ref{ass:score} concerns population orthogonality,
Assumption~\ref{ass:moments} concerns the instrument second-moment law, and
Assumption~\ref{ass:primitiveclt} concerns the shape of the limiting score experiment. None of these conditions implies the others. In particular, a valid instrument can generate a highly non-Gaussian score limit, and a well-behaved score covariance does not imply that the projected Jacobian is well conditioned.

\subsection{Dimension-specific endogeneity and score dependence}

Endogeneity concerns the dependence between the reduced-form disturbance and the structural disturbance, whereas sampling uncertainty concerns the dependence of the IV score. These are different objects.

\begin{example}[Motivating asymmetric-endogeneity specialization]\label{ex:endog}
Define the unit-state and time-state covariance projections
\[
 \Delta^{A}_{i;ts}=\mathbb E(v_{it}^{o}u_{is}^{o}\mid A_i)-\mathbb E(v_{it}^{o}u_{is}^{o}),
 \qquad
 \Delta^{D}_{t;ij}=\mathbb E(v_{it}^{o}u_{jt}^{o}\mid D_t)-\mathbb E(v_{it}^{o}u_{jt}^{o}).
\]
For infinitely many array sizes, there is at least one pair $(t,s)$ in the maintained time-index set for which
$\Pr(\Delta^{A}_{i;ts}\neq0)>0$ for some $i$, whereas
$\Delta^{D}_{t;ij}=0$ almost surely for every $(i,j,t)$ in the corresponding array. Thus the conditional regressor--error covariance is transmitted through the unit state, although both unit and time states may contribute to the covariance of the IV score. This specialization does not rule out spatial dependence across the unit states or serial dependence across the time states. The time-asymmetric case follows by interchanging $N$ and $T$. It motivates the asymmetric examples and simulations; no theorem below uses it as a maintained condition.
\end{example}

For the null score
\[
 g_{it}(\beta)=z_{it}^{o}(y_{it}^{o}-x_{it}^{o\prime}\beta),
\]
let $\mu_{it,NT}^g=\mathbb E[g_{it}(\beta_0)]$ and define the exact projection components
\begin{align*}
 a_{i,t,NT}^g&=\mathbb E[g_{it}(\beta_0)\mid A_i]-\mu_{it,NT}^g,\\
 d_{i,t,NT}^g&=\mathbb E[g_{it}(\beta_0)\mid D_t]-\mu_{it,NT}^g,\\
 w_{it,NT}^g&=\mathbb E[g_{it}(\beta_0)\mid A_i,D_t]
              -\mathbb E[g_{it}(\beta_0)\mid A_i]
              -\mathbb E[g_{it}(\beta_0)\mid D_t]+\mu_{it,NT}^g,\\
 e_{it,NT}^g&=g_{it}(\beta_0)-\mathbb E[g_{it}(\beta_0)\mid A_i,D_t].
\end{align*}
Under Assumption~\ref{ass:score}, the average centering term satisfies $(NT)^{-1}\sum_{i,t}\mu_{it,NT}^g=0$; individual cell means need not vanish. With
$a_{i,NT}^g=T^{-1}\sum_ta_{i,t,NT}^g$ and
$d_{t,NT}^g=N^{-1}\sum_id_{i,t,NT}^g$, the sample average has the exact decomposition
\begin{equation}
 \bar g_{NT}(\beta_0)
 =N^{-1}\sum_i a_{i,NT}^g+T^{-1}\sum_t d_{t,NT}^g
 +(NT)^{-1}\sum_{i,t}(w_{it,NT}^g+e_{it,NT}^g).             \label{eq:decomp}
\end{equation}
The four terms are the unit, time, interaction, and conditionally centered cell projections. The identity does not require independence within the unit- or time-state sequences; those dependence restrictions enter only when deriving probability limits.
Because the oracle projection mixes cells, $g_{it}(\beta_0)$ need not be measurable with respect to the local field $\mathcal G_{it,NT}$. The decomposition above is nevertheless a valid Hoeffding-type conditional projection with respect to $(A_i,D_t)$ in the full array probability space. It is an algebraic decomposition, not a claim that the final remainder is conditionally independent across cells.

\begin{assumption}[High-level joint normalized score limit]\label{ass:primitiveclt}
There exist deterministic nonsingular matrices
$A_{NT}\in\mathbb R^{q\times q}$ and
$B_{NT}\in\mathbb R^{pq\times pq}$ such that
\[
 \left(A_{NT}^{-1}\sum_{i,t}z_{it}^{o}u_{it}^{o},\;
 B_{NT}^{-1}\sum_{i,t}\operatorname{vec}(z_{it}^{o}v_{it}^{o\prime})\right)
 \Rightarrow (Z_u,Z_v),
\]
where the joint limit is nondegenerate on a maintained, fixed covariance-rank stratum.
The assumption does not impose Gaussianity. A primitive application must derive the
normalizers, the complete joint law, all cross-covariances, the relevant Lindeberg or
mixing conditions, and the negligibility of any interaction-truncation remainder.
\end{assumption}

\subsection{A primitive finite-range benchmark}\label{subsec:primitivebenchmark}

The high-level score limit permits non-Gaussian laws. The Online Appendix verifies this
possibility for an exact decomposition into unit, time, unit--time-interaction, and cell
components. Because these components accumulate at different rates, they must occupy
orthogonal score blocks:
\[
 \mathbb R^d=V_A\oplus V_D\oplus V_0,\qquad
 K_{NT}=T\sqrt N\,P_A+N\sqrt T\,P_D+\sqrt{NT}\,P_0.
\]
Unit effects lie in $V_A$, time effects in $V_D$, and interaction and cell effects in
$V_0$. This separation is necessary: one score coordinate cannot carry both a
nondegenerate main effect and a nondegenerate interaction component under the three
distinct rates.

\begin{assumption}[Block-separated primitive benchmark]\label{ass:primitivebenchmark}
The centered stacked score has an exact unit, time, bilinear-interaction, and cell
decomposition. The unit and time state arrays satisfy fixed-range triangular-array CLTs;
the cell array satisfies a Lindeberg CLT and is independent of the two state arrays; the
bilinear coefficient matrices converge; and the four components obey the orthogonal
block restrictions described above. The Online Appendix states the complete moment and
long-run covariance conditions.
\end{assumption}

\begin{theorem}[Primitive joint projected-score limit]\label{thm:primitive-benchmark}
Under Assumption~\ref{ass:primitivebenchmark} and the complete covariance-limit
conditions stated in the Online Appendix,
\[
 K_{NT}^{-1}\sum_{i=1}^N\sum_{t=1}^T\xi_{it,NT}
 \Rightarrow
 P_A Z_a+P_D Z_d+P_0\{\mathfrak h(Z_\phi,Z_\psi)+Z_\varepsilon\}.
\]
The Gaussian blocks $(Z_a,Z_\phi)$, $(Z_d,Z_\psi)$, and $Z_\varepsilon$ are mutually
independent. The limit is Gaussian on $V_A\oplus V_D$. A nonzero bilinear interaction
contributes a second Gaussian-chaos component on $V_0$, so covariance information alone
generally does not determine the limiting law.
\end{theorem}

\noindent\emph{Complete assumptions, proof, and non-vacuity example.} See the Online
Appendix.

\section{Identification through the projected moment map}\label{sec:identification}

This section asks which combinations of structural effects remain distinguishable after
nuisance removal. The projected population Jacobian provides the answer. Its rank
determines whether the transformed moments identify the full parameter, while its
singular values and singular subspaces describe the strength and direction of the
surviving economic variation.

Define the population projected moment map
\begin{equation}
 \mathcal Q_{NT}(\beta)
 =\frac1{NT}\sum_{i=1}^{N}\sum_{t=1}^{T}
 \mathbb E\!\left[z_{it}^{o}(y_{it}^{o}-x_{it}^{o\prime}\beta)\right].
 \label{eq:idmap}
\end{equation}
By instrument validity, $\mathcal Q_{NT}(\beta_0)=0$. Its Jacobian is
\begin{equation}
 J_{NT}
 =-\frac{\partial\mathcal Q_{NT}(\beta)}{\partial\beta'}
 =\frac1{NT}\sum_{i,t}\mathbb E[z_{it}^{o}x_{it}^{o\prime}]
 =Q_{zz,NT}^{E}\Pi_{NT}+Q_{zv,NT}^{E},                  \label{eq:jacobian}
\end{equation}
where
\[
 Q_{zv,NT}^{E}=\frac1{NT}\sum_{i,t}\E[z_{it}^{o}v_{it}^{o\prime}]
\]
and superscript $E$ denotes a population average. By the reduced-form orthogonality restriction in Assumption~\ref{ass:score}, $Q_{zv,NT}^{E}=0$, so identification is governed by $Q_{zz,NT}^{E}\Pi_{NT}$. This restriction is logically distinct from the population-average structural exclusion $(NT)^{-1}\sum_{i,t}\E[z_{it}^{o}u_{it}^{o}]=0$. The spectrum of $J_{NT}$ therefore provides a coordinate-free measure of identification in the endogenous-regressor space.

\begin{definition}[Identification regimes]\label{def:kappa}
Let $s_{1,NT}\ge\cdots\ge s_{p,NT}\ge0$ be the singular values of $J_{NT}$.
\begin{enumerate}
 \item Identification is strong when $\liminf s_{p,NT}>0$.
 \item Identification is locally weak on a separated singular-value cluster when the singular values in that cluster converge to zero at specified rates while remaining positive. Individual directions are interpreted only for simple separated singular values; repeated values are interpreted through their invariant cluster subspaces.
 \item Identification is partial when the population Jacobian is rank deficient on the maintained model, so that the moment equation has a non-singleton identified set. A full-rank sequence converging to a rank-deficient limit is instead a weak-identification sequence, not exact partial identification at each finite $(N,T)$.
\end{enumerate}
For an unrestricted local parameter space, the identified linear functionals have coefficient vectors in the row space of the relevant population Jacobian.
\end{definition}

\begin{definition}[Spectral identification sequences]\label{def:spectral-sequence}
Let
\[
 J_{NT}=L_{NT}\operatorname{diag}(s_{1,NT},\ldots,s_{p,NT})R_{NT}',
 \qquad s_{1,NT}\ge\cdots\ge s_{p,NT}\ge0,
\]
be a thin singular-value decomposition, where $L_{NT}\in\mathbb R^{q\times p}$ and
$R_{NT}\in\mathbb R^{p\times p}$ have orthonormal columns. For each singular-value
cluster whose interpretation is used below, its spectral projector converges and the
cluster is separated from the remaining singular values by a stated deterministic gap.
Within a simple cluster, there is a deterministic sequence $d_{j,NT}\ge1$ and a constant
$\kappa_j\ge0$ such that $d_{j,NT}s_{j,NT}\to\kappa_j$. Strong clusters have bounded
$d_{j,NT}$ and positive limits; locally weak clusters have $d_{j,NT}\to\infty$ and positive
scaled limits; limiting unidentified clusters have zero unscaled limits. Individual
singular vectors are not interpreted inside a repeated singular-value cluster.
\end{definition}

This convention is a local expansion of the primitive moment Jacobian, not an independent weak-instrument model. In the linear reduced form it is implied by a corresponding expansion of $\Pi_{NT}$, but the formulation in terms of $J_{NT}$ also applies when the projected first-stage relationship is nonlinear.

\paragraph{Purpose.}
The next theorem characterizes identification before any central limit theorem or
bootstrap argument is introduced. It identifies the exact population object that
determines which structural directions remain visible after nuisance projection.

\begin{theorem}[Projected identification map]\label{thm:identification-map}
Under the projected structural equation \eqref{eq:oracle} and the population-average instrument validity in Assumption~\ref{ass:score}, the map $\mathcal Q_{NT}(\beta)$ is affine in $\beta$ with Jacobian $-J_{NT}$. Hence:
\begin{enumerate}
 \item on an unrestricted affine parameter space, the population moment equation has a unique solution if and only if $\rank(J_{NT})=p$; on a restricted parameter set $\mathcal B$, the root $\beta_0$ is unique if and only if
 $\mathcal N(J_{NT})\cap(\mathcal B-\beta_0)=\{0\}$, while
 $\mathcal N(J_{NT})\cap(\mathcal B-\mathcal B)=\{0\}$ is the stronger condition that the moment map is injective on all of $\mathcal B$;
 \item the right-singular subspaces of $J_{NT}$ determine the structural subspaces identified by the moments, while the singular values measure deterministic relevance; individual singular vectors are identified only for simple, separated singular values;
 \item if $J_{NT}\to J_0$ with $\rank(J_0)=r<p$, the limiting affine equation identifies the set $\beta_0+\mathcal N(J_0)$ (intersected with the parameter space). When the parameter space is locally unrestricted, the row-space component $P_{\mathcal R(J_0')}\beta_0$ is identified, while perturbations in $\mathcal N(J_0)$ leave the limiting moment unchanged;
 \item the theorem is algebraic: the stochastic classification of a direction as strong, local, or unidentified additionally requires comparing its population singular value with the sampling and factor-estimation error of an estimated Jacobian.
\end{enumerate}
\end{theorem}

\noindent\emph{Proof.} See Online Appendix~\ref{app:proof:thm-identification-map}.

\paragraph{Strong-identification set recovery.}
When the smallest singular value of $J_{NT}$ is bounded away from zero, a feasible affine
sample moment map with $O_p(r_{NT})$ score error has a unique root
$\widehat\beta_{NT}$ satisfying
$d_H(\{\widehat\beta_{NT}\},\{\beta_0\})=O_p(r_{NT})$; see Online Appendix
Proposition~\ref{prop:strong-set-recovery}. The bound deteriorates as the projected
Jacobian approaches singularity, which is precisely why weak-identification inference is
needed.

\paragraph{Interpretation.}
The theorem separates the existence of economic identifying variation from the ability
to detect it in a finite sample. Population rank asks whether the projected moments can
in principle distinguish structural effects. Statistical diagnosis asks whether the
surviving variation is strong enough to separate those effects empirically after factor
and sampling uncertainty are taken into account.

The next subsection moves from this population question to the recoverability of the
same geometry from feasible sample moments.

\subsection{Spectral stability of feasible identification diagnostics}\label{subsec:stability}

Population geometry is useful only if it can be recovered from the sample. This
subsection quantifies when sampling and factor-estimation errors are small relative to
the signal and spectral gaps that define economically interpretable directions.

Define the oracle sample Jacobian and its feasible counterpart by
\[
 \widetilde J_{NT}=\frac1{NT}\sum_{i,t}z_{it}^{o}x_{it}^{o\prime},
 \qquad
 \widehat J_{NT}=\frac1{NT}\sum_{i,t}\widehat z_{it}^{o}\widehat x_{it}^{o\prime}.
\]
Write
\[
 \tau_{NT}=\|\widetilde J_{NT}-J_{NT}\|,
 \qquad
 \rho_{NT}=\|\widehat J_{NT}-\widetilde J_{NT}\|,
 \qquad
 e_{NT}=\tau_{NT}+\rho_{NT},
\]
where $\|\cdot\|$ is the spectral norm. The term $\tau_{NT}$ is sampling error and $\rho_{NT}$ is the additional error due to feasible nuisance removal. Their separation is important: consistency of factor spaces alone does not determine the sampling precision of an identification diagnostic.

\paragraph{Purpose.}
Population rank is not observed. We use standard singular-value perturbation tools \citep{wedin1972,stewartsun1990}. The next proposition quantifies how sampling error and
feasible nuisance removal perturb the singular values and singular subspaces used to
diagnose projected identification.

\begin{proposition}[Singular-value and subspace stability]\label{prop:stability}
For every $j\le p$,
\[
 |\sigma_j(\widehat J_{NT})-\sigma_j(J_{NT})|\le e_{NT}.
\]
Suppose a singular-value cluster of $J_{NT}$ is separated from the remainder of its
spectrum by a deterministic gap $\gamma_{NT}>0$. On the event
$\{e_{NT}<\gamma_{NT}/2\}$, the corresponding feasible right-singular projector satisfies
\[
 \|\widehat P_R-P_R\|\le \frac{2e_{NT}}{\gamma_{NT}}.
\]
In particular, if $e_{NT}=O_p(r_{NT})$ for a deterministic sequence
$r_{NT}=o(\gamma_{NT})$, then
$\|\widehat P_R-P_R\|=O_p(r_{NT}/\gamma_{NT})$.
Consequently:
\begin{enumerate}[leftmargin=2em]
 \item if $e_{NT}=o_p(c_{NT})$ and
 $c_{NT}=o\{\sigma_p(J_{NT})\}$, thresholding at $c_{NT}$ detects full-column-rank
 identification with probability approaching one;
 \item if $\sigma_j(J_{NT})\le C d_{NT}$ and
 $e_{NT}=O_p(d_{NT})$ for a deterministic $d_{NT}\downarrow0$, the perturbation bound
 alone does not place $\sigma_j(J_{NT})$ outside an $O_p(d_{NT})$ neighborhood of zero;
 \item numerical rank and singular subspaces are consistently recovered only for
 clusters whose nonzero singular values and separating gaps dominate $e_{NT}$.
\end{enumerate}
\end{proposition}

\noindent\emph{Proof.} See Online Appendix~\ref{app:proof:prop-stability}.

\paragraph{Interpretation.}
A reported singular value or direction is credible only relative to its estimation error.
The proposition therefore distinguishes the existence of population identification from
the empirical ability to diagnose it. Spectral gaps are essential for interpreting
directions, not merely technical devices used in the proof.

\begin{remark}[Practical interpretation]\label{rem:stability-practice}
A reported singular value is economically interpretable only together with an uncertainty scale. Empirical work should therefore accompany $\sigma_j(\widehat J_{NT})$ with a bootstrap or subsampling distribution for the Jacobian, sensitivity to the factor dimension, and several numerical-rank tolerances. Proposition~\ref{prop:stability} also explains why identification diagnostics should not be used as a discontinuous pretest for choosing conventional rather than robust inference.
\end{remark}

\subsection{Statistical inference for the projected Jacobian spectrum}

The perturbation bounds distinguish stable from unstable spectral objects. We now attach
sampling uncertainty to the stable ones. The resulting intervals ask whether a projected
identifying direction is measurably separated from zero; they are not uniform rank tests
at a boundary.

\label{subsec:diagnostic-inference}

Perturbation bounds determine whether a diagnostic is stable; statistical inference
requires a first-order law for the estimated Jacobian. Throughout the bootstrap results,
$\mathcal L^*(\cdot)$ denotes conditional law given the data and $d_{BL}$ denotes
bounded-Lipschitz distance. Let $s_{NT}\downarrow0$ be a deterministic rate and write
\[
 \Delta_{NT}=\widehat J_{NT}-J_{NT}.
\]

\begin{assumption}[Asymptotic linearity of the feasible Jacobian]
\label{ass:jacobian-clt}
There is a mean-zero random matrix $\mathbb Z_J\in\mathbb R^{q\times p}$ such that
\[
 s_{NT}^{-1}\Delta_{NT}\Rightarrow\mathbb Z_J.
 \label{eq:jacobian-clt}
\]
For bootstrap inference, there is a feasible replicate $\widehat J_{NT}^*$ satisfying
\[
 d_{BL}\!\left[
 \mathcal L^*\{s_{NT}^{-1}(\widehat J_{NT}^*-\widehat J_{NT})\},
 \mathcal L(\mathbb Z_J)
 \right]\to_p0.
 \label{eq:jacobian-bootstrap}
\]
The assumption may be verified by combining a sample-Jacobian central limit theorem with
the factor-transfer expansion. It is imposed only for the diagnostics in this subsection.
\end{assumption}

For a simple singular value $\sigma_j(J_{NT})$, let $u_{j,NT}$ and $v_{j,NT}$ denote
unit left and right singular vectors, with signs chosen by any deterministic local
orientation rule. Define the separation of the positive dilation eigenvalue $\sigma_j(J_{NT})$ by
\[
 \delta_{j,NT}
 =\min\!\left\{
 \sigma_j(J_{NT}),
 \min_{k\ne j}|\sigma_j(J_{NT})-\sigma_k(J_{NT})|
 \right\},
\]
with the inner minimum interpreted as $+\infty$ when there is no $k\ne j$. The first
term is the distance to the zero eigenspace of the symmetric dilation; distances to its
negative eigenvalues are no smaller. This is the gap used in the perturbation expansion
below.

\paragraph{Purpose.}
The perturbation bounds above establish stability but not statistical uncertainty. The
next theorem provides the first-order distribution and bootstrap approximation needed to
construct confidence intervals for simple, separated singular values.

\begin{theorem}[First-order inference for separated singular values]
\label{thm:singular-inference}
Suppose Assumption~\ref{ass:jacobian-clt} holds, the singular value is simple,
$\delta_{j,NT}>0$, and
\[
 \frac{s_{NT}}{\delta_{j,NT}}\to0.
\]
If $u_{j,NT}\to u_j$ and $v_{j,NT}\to v_j$, then
\[
 s_{NT}^{-1}\{\sigma_j(\widehat J_{NT})-\sigma_j(J_{NT})\}
 \Rightarrow u_j'\mathbb Z_Jv_j.
 \label{eq:singular-delta}
\]
Moreover,
\[
 \sigma_j(\widehat J_{NT})-\sigma_j(J_{NT})
 =u_{j,NT}'\Delta_{NT}v_{j,NT}
 +O_p\!\left(\frac{s_{NT}^2}{\delta_{j,NT}}\right).
 \label{eq:singular-expansion}
\]
The bootstrap consistently estimates the first-order law:

\[
 d_{BL}\!\left[
 \mathcal L^*\!\left\{
 s_{NT}^{-1}\bigl(
 \sigma_j(\widehat J_{NT}^*)-\sigma_j(\widehat J_{NT})
 \bigr)\right\},
 \mathcal L(u_j'\mathbb Z_Jv_j)
 \right]\to_p0.
 \label{eq:singular-bootstrap}
\]
\end{theorem}

Continuity of the scalar limit is required only for quantile-based confidence intervals,
not for bounded-Lipschitz bootstrap consistency.

\noindent\emph{Proof.} See Online Appendix~\ref{app:proof:thm-singular-inference}.

\paragraph{Interpretation.}
A singular value measures how strongly the data distinguish one combination of
structural effects after nuisance removal. The theorem permits confidence intervals for
that surviving relevance. A small lower confidence bound means that the corresponding
economic mechanism may not be distinguishable even when the point estimate appears
reassuring.

\begin{corollary}[Confidence intervals for a simple singular value]
\label{cor:singular-ci}
Under Theorem~\ref{thm:singular-inference}, let
$\widehat c_{j,1-\alpha/2}^*$ and $\widehat c_{j,\alpha/2}^*$ be the conditional bootstrap
quantiles of
$s_{NT}^{-1}\{\sigma_j(\widehat J_{NT}^*)-\sigma_j(\widehat J_{NT})\}$.
Then
\[
 \left[
 \sigma_j(\widehat J_{NT})-s_{NT}\widehat c_{j,1-\alpha/2}^*,
 \;
 \sigma_j(\widehat J_{NT})-s_{NT}\widehat c_{j,\alpha/2}^*
 \right]
 \label{eq:singular-ci}
\]
has asymptotic coverage $1-\alpha$, provided the limiting distribution is continuous at
the relevant quantiles. The deterministic scale $s_{NT}$ may be replaced by a feasible
$\widehat s_{NT}$ satisfying $\widehat s_{NT}/s_{NT}\to_p1$. A one-sided lower bound is
a pointwise procedure for a fixed positive, simple, separated singular value; it is not a
uniform rank test at a local boundary. Intersecting the interval with $[0,\infty)$ is
conservative and respects the parameter space.
\end{corollary}

\noindent\emph{Proof.} See Online Appendix~\ref{app:proof:cor-singular-ci}.

\paragraph{Rank recovery.}
The next result states when thresholding consistently recovers population rank and why
local rank boundaries must be excluded.

\begin{theorem}[Rank consistency on separated strata]\label{thm:rank-consistency}
Suppose Assumption~\ref{ass:jacobian-clt} holds,
$\rank(J_{NT})=r$ eventually, and
$c_{NT}\to0$, $s_{NT}/c_{NT}\to0$.
If $r\ge1$, also suppose
\[
 \frac{\sigma_r(J_{NT})}{c_{NT}}\to\infty.
\]
For $r=0$, this last condition is omitted.
Because the maintained IV system has $q\ge p$, define
\[
 \widehat r_{NT}
 =\sum_{j=1}^{p}\mathbf 1\{\sigma_j(\widehat J_{NT})>c_{NT}\}.
 \label{eq:rank-estimator}
\]
Then $\Pr(\widehat r_{NT}=r)\to1$. This conclusion is pointwise on separated rank strata
and, when $r\ge1$, does not apply to local sequences with $\sigma_r(J_{NT})=O(s_{NT})$.
\end{theorem}

\noindent\emph{Proof.} See Online Appendix~\ref{app:proof:thm-rank-consistency}.

\paragraph{Interpretation.}
Rank recovery is possible only when the smallest retained singular value dominates both
the threshold and the sampling error. Near a local rank boundary, rank should be reported
as uncertain rather than converted into a definitive pretest decision.

\paragraph{Subspace recovery.}
Singular values can be well estimated even when their directions are unstable. The next
result gives the uncertainty-to-gap condition for interpreting an estimated subspace.

\begin{theorem}[Subspace stability and inferential reliability]
\label{thm:subspace-reliability}
Let a right-singular cluster of $J_{NT}$ be separated from the rest of the spectrum by
$\gamma_{NT}>0$, and suppose $\|\Delta_{NT}\|=O_p(s_{NT})$. Then
\[
 \|\widehat P_R-P_R\|
 =O_p\!\left(\min\left\{1,\frac{s_{NT}}{\gamma_{NT}}\right\}\right).
 \label{eq:subspace-rate}
\]
If $s_{NT}/\gamma_{NT}\to0$, the cluster subspace is consistently estimable. If the ratio
does not vanish, the theorem does not establish stability of individual directions inside
or adjacent to the cluster, even when the associated singular values are precisely
estimated.
\end{theorem}

\noindent\emph{Proof.} See Online Appendix~\ref{app:proof:thm-subspace-reliability}.

\paragraph{Interpretation.}
The ratio $s_{NT}/\gamma_{NT}$ measures whether the combinations of policy or behavioral
parameters supported by the data are stable. A small ratio permits those combinations to
be interpreted. A nonvanishing ratio means that the available
bound cannot rule out economically important rotations among mechanisms, even when the
associated singular values are estimated precisely.

\paragraph{Practical reporting.}
For every retained singular value, report its estimate, a one- or two-sided confidence
interval, the nearest spectral gap, and the ratio of the estimated uncertainty scale to
that gap. Report rank over a range of theoretically admissible thresholds rather than as
an unconditional fact. Bootstrap re-estimation must repeat the same projection,
factor-dimension choice, and Jacobian construction used in the original sample unless a
fixed-projection bootstrap is justified by the relevant oracle-transfer theorem.

\subsection{Scalar concentration and asymmetric normalization}

For one endogenous regressor and one instrument, let
\[
 \Omega_{v,NT}
 =\Var\!\left(N^{-1/2}\sum_i T^{-1}\sum_t z_{it}^{o}v_{it}^{o}\right),
 \qquad T^\delta\Omega_{v,NT}\to\Omega_v\in(0,\infty).
\]
The concentration sequence and the first-stage score scale are
\begin{equation}
 \kappa_{N\mid T}
 =N\pi_{NT}^{2}\frac{(Q_{zz,NT}^{E})^{2}}{\Omega_{v,NT}},
 \qquad
 b_{NT}=T\sqrt{N\Omega_{v,NT}}.
 \label{eq:kappa}
\end{equation}
Thus a nondegenerate local experiment uses
$\pi_{NT}=c/\sqrt{NT^\delta}$, while the stochastic score normalization is
$b_{NT}\asymp\sqrt N\,T^{1-\delta/2}$; the rates coincide only when $\delta=1$.

\begin{proposition}[Weak-first-stage balance]\label{prop:balance}
Suppose $T^\delta\Omega_{v,NT}\to\Omega_v\in(0,\infty)$,
$\pi_{NT}=c/\sqrt{NT^\delta}$,
$b_{NT}^{-1}\sum_{i,t}z_{it}^{o}v_{it}^{o}\Rightarrow Z_v$, and
$\widehat Q_{zz,NT}=(NT)^{-1}\sum_{i,t}(z_{it}^{o})^2\to_pQ_{zz}>0$.
Define
\[
 \mu_c=\frac{cQ_{zz}}{\sqrt{\Omega_v}}.
 \label{eq:mu-c}
\]
Then
\[
 b_{NT}^{-1}\sum_{i,t}z_{it}^{o}x_{it}^{o}
 \Rightarrow \mu_c+Z_v.
\]
\end{proposition}

\noindent\emph{Proof and additional discussion.} See Online Appendix~\ref{app:proof:prop-balance}. The distinction is
that coefficient localization, first-stage-score normalization, and structural-score
normalization are distinct objects under asymmetric panel dependence.

\section{The local projected-moment experiment}\label{sec:localexperiment}

Population rank alone does not reveal whether the surviving variation is statistically
informative. The purpose of this section is to combine moment sensitivity with the
sampling law of the projected score. This local representation becomes the common
language for power, weak identification, covariance-optimal weighting, and the Projected
Information Matrix.

This section provides the common asymptotic experiment behind regular projected GMM,
weak-IV ratio limits, and Anderson--Rubin local power. The object is the experiment
generated by the projected moments, not necessarily the likelihood experiment generated
by the full panel.

Let $A_{NT}\in\mathbb R^{q\times q}$ and $D_{NT}\in\mathbb R^{p\times p}$ be deterministic
nonsingular normalizing matrices. For $h\in\mathbb R^p$, define
\[
 \beta_{NT}(h)=\beta_0+D_{NT}^{-1}h.
\]
The local family changes the structural location parameter while holding the joint law
of $(x_{it}^{o},z_{it}^{o},u_{it}^{o})_{i,t}$ fixed. Under the law indexed by $h$,
\[
 y_{it}^{o}=x_{it}^{o\prime}\beta_{NT}(h)+u_{it}^{o}.
\]
Define the normalized score evaluated at the null,
\[
 \mathcal S_{NT}^{0}(h)
 =A_{NT}^{-1}\sum_{i,t}z_{it}^{o}
 \bigl(y_{it}^{o}-x_{it}^{o\prime}\beta_0\bigr)
\]
and, for a candidate local parameter $k$, define
\[
 \mathcal S_{NT}(k;h)
 =A_{NT}^{-1}\sum_{i,t}z_{it}^{o}
 \bigl(y_{it}^{o}-x_{it}^{o\prime}\beta_{NT}(k)\bigr).
\]

\begin{assumption}[Local projected-moment regularity]\label{ass:localmoment}
The following conditions hold. The local indices $h$ and $k$ are restricted to fixed
compact subsets of $\mathbb R^p$ when the statistics above are evaluated.
\begin{enumerate}[label=(\roman*),leftmargin=2.6em]
\item
\[
 A_{NT}^{-1}\sum_{i,t}z_{it}^{o}u_{it}^{o}\Rightarrow Z,
\]
for a tight $q$-vector $Z$;
\item with
\[
 \widehat H_{NT}
 =A_{NT}^{-1}\left(\sum_{i,t}z_{it}^{o}x_{it}^{o\prime}\right)D_{NT}^{-1},
 \qquad
 H_{NT}
 =A_{NT}^{-1}(NTJ_{NT})D_{NT}^{-1},
\]
we have $\|\widehat H_{NT}-H_{NT}\|=o_p(1)$;
\item $H_{NT}\to\mathcal H\in\mathbb R^{q\times p}$.
\end{enumerate}
\end{assumption}

\paragraph{Roadmap.}
The projected score changes under a local alternative only through the projected
Jacobian. The next theorem formalizes that shift and then, on Gaussian fixed-rank strata,
combines sensitivity and score precision in the Projected Information Matrix. The power,
weighting, and weak-identification results that follow all build on this representation.

\begin{theorem}[Local projected-moment shift representation]\label{thm:localmoment}
Under Assumption~\ref{ass:localmoment},
\[
 \mathcal S_{NT}^{0}(h)\Rightarrow Z+\mathcal Hh,
 \qquad
 \mathcal S_{NT}(k;h)\Rightarrow Z+\mathcal H(h-k),
 \label{eq:local-moment-experiment}
\]
jointly over every finite collection of $(h,k)$ values contained in the maintained
compact local-index sets. Taking the union over arbitrary fixed compact local-index sets gives the moment-shift
family
\[
 \mathcal E_P=\{\mathcal L(Z+\mathcal Hh):h\in\mathbb R^p\}.
\]

If $Z\sim N(0,\Omega)$, $\rank(\Omega)=r_\Omega$ is fixed,
$\mathcal R(\mathcal H)\subseteq\mathcal R(\Omega)$, and the positive eigenvalues of
$\Omega$ are bounded away from zero, define the \emph{projected information matrix}
\begin{equation}
 \mathcal I_P=\mathcal H'\Omega^\dagger\mathcal H. \label{eq:projected-information}
\end{equation}
Then $\mathcal I_P$ is positive semidefinite and
\[
 \rank(\mathcal I_P)=\rank(\mathcal H).
\]
It is positive definite if and only if $\mathcal H$ has full column rank. Within the
Gaussian projected-moment experiment, the log density ratio relative to $h=0$ is
\[
 \ell_h(Y)
 =h'\mathcal H'\Omega^\dagger Y-\frac12h'\mathcal I_Ph,
 \qquad Y\in\mathcal R(\Omega).
 \label{eq:gaussian-shift-lr}
\]
\end{theorem}

\noindent\emph{Proof.} See Online Appendix~\ref{app:proof:thm-localmoment}.

\paragraph{Interpretation.}
The local experiment measures how much policy-relevant variation survives after common
shocks and latent heterogeneity have been removed. Different coefficients in the same
specification can therefore be learned with very different precision. Small eigenvalues
of $\mathcal I_P$ identify economic directions that the projected data can barely
distinguish; a zero eigenvalue corresponds to a mechanism absent from the local
experiment. When the score law is non-Gaussian, the full experiment rather than a
covariance-only summary is needed for inference.

\begin{remark}[Parameter geometry and scope]\label{rem:geometry-scope}
The structural parameter enters the projected moment condition linearly, and the joint
null is evaluated directly. For regular Wald and GMM results, $\beta_0$ is an interior
point of the maintained Euclidean parameter set. Identification-robust score inversion
can instead be formed in $\mathbb R^p$ and intersected with any fixed closed economic
restriction set $\mathcal B$, preserving joint null coverage. Set projection onto a
subvector remains valid but can be conservative or disconnected. Extensions to curved
parameter manifolds or inequality-constrained nuisance problems require local charts,
tangent cones, and separate uniform approximations and are outside the maintained model.
\end{remark}

\begin{corollary}[Gaussian local power of the quadratic AR statistic]\label{cor:ar-local-power}
Under the Gaussian conditions of Theorem~\ref{thm:localmoment}, suppose that, under
each fixed local law indexed by $h$,
$\widehat\Omega_{NT}\to_p\Omega$ and
$\rank(\widehat\Omega_{NT})=\rank(\Omega)$ with probability approaching one. Then,
under $\beta_{NT}(h)$,
\[
 \mathcal S_{NT}^{0}(h)'\widehat\Omega_{NT}^{\dagger}\mathcal S_{NT}^{0}(h)
 \Rightarrow\chi^2_{r_\Omega}\!\left(h'\mathcal I_Ph\right).
\]
Hence the noncentrality parameter of the quadratic Anderson--Rubin statistic equals the
projected information in direction $h$.
\end{corollary}

\noindent\emph{Proof.} See Online Appendix~\ref{app:proof:cor-ar-local-power}.

\begin{corollary}[Optimal linear unbiased rule in the Gaussian moment shift family]\label{cor:gaussian-blue}
If $\mathcal I_P$ is positive definite, the estimator
\[
 \widehat h_P
 =\mathcal I_P^{-1}\mathcal H'\Omega^\dagger Y
\]
satisfies $\widehat h_P\sim N(h,\mathcal I_P^{-1})$ in the Gaussian projected-moment
shift family. Among linear unbiased rules $KY$ satisfying $K\mathcal H=I_p$, where two rules are
identified when they agree on $\mathcal R(\Omega)$, it has the smallest covariance matrix
in the Loewner order.
\end{corollary}

\noindent\emph{Proof.} See Online Appendix~\ref{app:proof:cor-gaussian-blue}.

\begin{remark}[Scope of the local shift result]
Equation~\eqref{eq:gaussian-shift-lr} is the likelihood ratio inside the limiting Gaussian
moment experiment. It does not, by itself, prove local asymptotic normality of the full
panel-data likelihood. Such a result would additionally require a dominated statistical
model, quadratic-mean differentiability, and asymptotic sufficiency or equivalence of the
projected moments.
\end{remark}

\subsection{Inference for the projected information matrix}

The Jacobian records sensitivity, but economic information also depends on score
precision. This subsection develops inference for their covariance-adjusted combination
and separates uncertainty from estimating the derivative from uncertainty about the
score covariance.

\label{subsec:pim-inference}

Let
\[
 \mathcal I_P=\mathcal H'\Omega^\dagger\mathcal H,
 \qquad
 \widehat{\mathcal I}_{P,NT}
 =\widehat H_{NT}'
 \widehat\Omega_{NT}^{\dagger}
 \widehat H_{NT}.
\]
Inference for $\mathcal I_P$ combines uncertainty in the normalized Jacobian
$\mathcal H$ and in the score covariance $\Omega$. Because the generalized inverse is
not differentiable across rank changes, all results in this subsection are stated on a
fixed covariance-rank stratum.

\begin{assumption}[Joint first-order law for projected information inputs]
\label{ass:pim-clt}
There is a deterministic rate $\ell_{NT}\downarrow0$ and random limits
$\mathbb Z_H\in\mathbb R^{q\times p}$ and
$\mathbb Z_\Omega\in\mathbb S^q$, where $\mathbb S^q$ denotes the space of
symmetric $q\times q$ matrices, such that
\[
 \ell_{NT}^{-1}
 \begin{pmatrix}
 \operatorname{vec}(\widehat H_{NT}-\mathcal H)\\
 \operatorname{vech}(\widehat\Omega_{NT}-\Omega)
 \end{pmatrix}
 \Rightarrow
 \begin{pmatrix}
 \operatorname{vec}(\mathbb Z_H)\\
 \operatorname{vech}(\mathbb Z_\Omega)
 \end{pmatrix}.
 \label{eq:pim-joint-clt}
\]
This high-level condition includes the requirement that any deterministic drift
$H_{NT}-\mathcal H$ is $O(\ell_{NT})$ and is represented in the stated limit.
Moreover, $\widehat\Omega_{NT}$ denotes a rank-stabilized covariance estimator
(for example, a spectral truncation at a consistently selected rank) satisfying
$\rank(\widehat\Omega_{NT})=\rank(\Omega)=r_\Omega$ with probability approaching one.
The positive eigenvalues of $\Omega$ are bounded away from zero, and
$\mathcal R(\mathcal H)\subseteq\mathcal R(\Omega)$, as in the Gaussian local shift
representation. The range condition ensures that every local mean shift lies on the
support of the singular Gaussian score law; without it, some alternatives are mutually
singular and the quadratic PIM representation is not the appropriate finite local
information object. Let $P_0=I-\Omega\Omega^\dagger$ denote the projector onto the null
space of $\Omega$. The covariance perturbation limit is tangent to the symmetric
rank-$r_\Omega$ manifold at $\Omega$:
\[
 P_0\mathbb Z_\Omega P_0=0\quad\text{almost surely}.
\]
This tangency is implied by a first-order limit of estimators that remain on the same
rank stratum; it is stated explicitly because the Moore--Penrose derivative below is a
derivative along that manifold. Let
$(\widehat H_{NT}^*,\widehat\Omega_{NT}^*)$ be bootstrap replicates of
$(\widehat H_{NT},\widehat\Omega_{NT})$. The bootstrap consistently estimates the
joint law in \eqref{eq:pim-joint-clt} in the sense that
\[
 d_{BL}\!\left[
 \mathcal L^*\!\left\{
 \ell_{NT}^{-1}
 \begin{pmatrix}
 \operatorname{vec}(\widehat H_{NT}^*-\widehat H_{NT})\\
 \operatorname{vech}(\widehat\Omega_{NT}^*-\widehat\Omega_{NT})
 \end{pmatrix}
 \right\},
 \mathcal L\!\left\{
 \begin{pmatrix}
 \operatorname{vec}(\mathbb Z_H)\\
 \operatorname{vech}(\mathbb Z_\Omega)
 \end{pmatrix}
 \right\}
 \right]\to_p0.
\]
In addition,
\[
 \Pr^*\{\rank(\widehat\Omega_{NT}^*)=r_\Omega\}\to_p1
\]
in outer probability, and the original and bootstrap covariance estimators use the same
rank-stabilization rule. With
$\widehat P_{0,NT}=I-\widehat\Omega_{NT}\widehat\Omega_{NT}^{\dagger}$, assume also
\[
 \ell_{NT}^{-1}
 \widehat P_{0,NT}
 (\widehat\Omega_{NT}^*-\widehat\Omega_{NT})
 \widehat P_{0,NT}
 =o_{p^*}(1)
\]
in outer probability. This first-order tangency condition is automatic for a sufficiently
smooth bootstrap path on the same fixed-rank manifold, but equality of ranks alone does
not control a scaled secant remainder. An untruncated sample covariance that is
generically full rank does not satisfy the rank-stability condition when $\Omega$ is
singular.
\end{assumption}

Define the fixed-rank derivative of the generalized inverse by
\[
 \dot\Omega^\dagger[\mathbb Z_\Omega]
 =
 -\Omega^\dagger\mathbb Z_\Omega\Omega^\dagger
 +\Omega^\dagger\Omega^\dagger
  \mathbb Z_\Omega(I-\Omega\Omega^\dagger)
 +(I-\Omega^\dagger\Omega)
  \mathbb Z_\Omega\Omega^\dagger\Omega^\dagger.
 \label{eq:pinv-derivative}
\]
For symmetric $\Omega$, this is the Fr\'echet derivative of the Moore--Penrose inverse restricted to the fixed-rank manifold; see, for example, \citet{benisraelgreville2003}.

\paragraph{Roadmap.}
Estimating projected information introduces two distinct errors: uncertainty in the
moment derivative and uncertainty in the covariance geometry used to weight it. The next
theorem linearizes both components on a fixed-rank stratum. Its corollaries translate the
matrix limit into inference for economic directions, eigenvalues, and conditioning.

\begin{theorem}[Delta method for the projected information matrix]
\label{thm:pim-delta}
Under Assumption~\ref{ass:pim-clt},
\[
 \ell_{NT}^{-1}
 \{\widehat{\mathcal I}_{P,NT}-\mathcal I_P\}
 \Rightarrow
 \mathbb Z_{\mathcal I},
 \label{eq:pim-delta}
\]
where
\[
 \mathbb Z_{\mathcal I}
 =
 \mathbb Z_H'\Omega^\dagger\mathcal H
 +\mathcal H'\Omega^\dagger\mathbb Z_H
 +\mathcal H'\dot\Omega^\dagger[\mathbb Z_\Omega]\mathcal H.
 \label{eq:pim-linear-map}
\]
Equivalently,
\[
 \widehat{\mathcal I}_{P,NT}-\mathcal I_P
 =
 (\widehat H_{NT}-\mathcal H)'\Omega^\dagger\mathcal H
 +\mathcal H'\Omega^\dagger
  (\widehat H_{NT}-\mathcal H)
 +\mathcal H'
  \{\widehat\Omega_{NT}^{\dagger}-\Omega^\dagger\}
  \mathcal H
 +o_p(\ell_{NT}).
 \label{eq:pim-expansion}
\]
Moreover,
\[
 d_{BL}\!\left[
 \mathcal L^*\!\left\{
 \ell_{NT}^{-1}
 (\widehat{\mathcal I}_{P,NT}^*-\widehat{\mathcal I}_{P,NT})
 \right\},
 \mathcal L(\mathbb Z_{\mathcal I})
 \right]\to_p0.
 \label{eq:pim-bootstrap-law}
\]
\end{theorem}

\noindent\emph{Proof.} See Online Appendix~\ref{app:proof:thm-pim-delta}.

\paragraph{Interpretation.}
Uncertainty about projected information has two sources. The data may be uncertain about
how strongly economic outcomes respond to the surviving excluded variation, and they may
also be uncertain about how noisy that variation is under the dependence structure. The
expansion shows how these two forms of uncertainty jointly determine confidence about the
weakest economically informative direction.

\begin{corollary}[Directional projected-information inference]
\label{cor:pim-directional}
For any fixed $h\in\mathbb R^p$,
\[
 \ell_{NT}^{-1}
 \left\{
 h'\widehat{\mathcal I}_{P,NT}h-h'\mathcal I_Ph
 \right\}
 \Rightarrow h'\mathbb Z_{\mathcal I}h.
 \label{eq:pim-directional}
\]
If the limiting scalar law is continuous, bootstrap quantiles yield asymptotically valid
confidence intervals for the local noncentrality parameter $h'\mathcal I_Ph$.
\end{corollary}

\noindent\emph{Proof.} See Online Appendix~\ref{app:proof:cor-pim-directional}.

\begin{corollary}[Inference for a simple information eigenvalue]
\label{cor:pim-eigen}
Suppose $\lambda_j(\mathcal I_P)$ is simple and separated from the remaining eigenvalues.
Let $e_j$ be its unit eigenvector. Then
\[
 \ell_{NT}^{-1}
 \{\lambda_j(\widehat{\mathcal I}_{P,NT})
 -\lambda_j(\mathcal I_P)\}
 \Rightarrow
 e_j'\mathbb Z_{\mathcal I}e_j.
 \label{eq:pim-eigen}
\]
The bootstrap consistently estimates this law. When
$\lambda_{\min}(\mathcal I_P)$ is simple and fixed away from zero, a one-sided lower
confidence bound quantifies uncertainty about the weakest local direction. This pointwise
statement is not a uniform test at an information-rank boundary.
\end{corollary}

\noindent\emph{Proof.} See Online Appendix~\ref{app:proof:cor-pim-eigen}.

\begin{corollary}[Condition-number inference on positive-definite strata]
\label{cor:pim-condition}
Suppose $\mathcal I_P$ is positive definite and its largest and smallest eigenvalues are
simple and separated. Define
\[
 \kappa_P=\frac{\lambda_{\max}(\mathcal I_P)}
                 {\lambda_{\min}(\mathcal I_P)},
 \qquad
 \widehat\kappa_P
 =\frac{\lambda_{\max}(\widehat{\mathcal I}_{P,NT})}
        {\lambda_{\min}(\widehat{\mathcal I}_{P,NT})}.
\]
Then
\[
 \ell_{NT}^{-1}(\widehat\kappa_P-\kappa_P)
 \Rightarrow
 \frac{e_{\max}'\mathbb Z_{\mathcal I}e_{\max}}
      {\lambda_{\min}(\mathcal I_P)}
 -
 \frac{\lambda_{\max}(\mathcal I_P)}
      {\lambda_{\min}(\mathcal I_P)^2}
 e_{\min}'\mathbb Z_{\mathcal I}e_{\min}.
 \label{eq:pim-condition}
\]
The result is not used when $\lambda_{\min}(\mathcal I_P)$ approaches zero, because the
condition number is then nonregular.
\end{corollary}

\noindent\emph{Proof.} See Online Appendix~\ref{app:proof:cor-pim-condition}.

\paragraph{Purpose.}
Pointwise delta-method results do not by themselves provide simultaneous uncertainty for
the full information matrix. The next theorem constructs a bootstrap confidence region
for the entire positive-semidefinite matrix on a fixed-rank stratum.

\begin{theorem}[Bootstrap confidence regions for projected information]
\label{thm:pim-bootstrap-region}
Define the bootstrap information matrix
$\widehat{\mathcal I}_{P,NT}^*
=\widehat H_{NT}^{*\prime}\widehat\Omega_{NT}^{*\dagger}\widehat H_{NT}^*$.
Let $\mathbb S_+^p$ denote the cone of positive-semidefinite $p\times p$ matrices,
let $\|\cdot\|_{\mathsf F}$ denote the Frobenius norm. If the distribution of $\|\mathbb Z_{\mathcal I}\|_{\mathsf F}$ is continuous at its
$(1-\alpha)$ quantile, then the bootstrap region
\[
 \mathcal C_{1-\alpha}^{\mathcal I}
 =
 \left\{
 M\in\mathbb S_+^p:
 \|\widehat{\mathcal I}_{P,NT}-M\|_{\mathsf F}
 \le \ell_{NT}\widehat c_{1-\alpha}^{*,\mathcal I}
 \right\}
 \label{eq:pim-region}
\]
has asymptotic coverage $1-\alpha$, where
$\widehat c_{1-\alpha}^{*,\mathcal I}$ is the conditional bootstrap quantile of
$\ell_{NT}^{-1}
\|\widehat{\mathcal I}_{P,NT}^*
 -\widehat{\mathcal I}_{P,NT}\|_{\mathsf F}$.
\end{theorem}

\noindent\emph{Proof.} See Online Appendix~\ref{app:proof:thm-pim-bootstrap-region}.

\paragraph{Reporting.}
The confidence region quantifies uncertainty about the complete local-information
geometry, not only one eigenvalue. It is most useful when substantive questions involve
several directions or when eigenvalues are close enough that individual eigenvectors are
not stable. Report $\widehat{\mathcal I}_{P,NT}$ together with confidence intervals for its smallest
and largest simple eigenvalues, confidence intervals for economically relevant
directional quantities $h'\mathcal I_Ph$, and a confidence region for the full matrix.
A large point estimate of $\lambda_{\min}$ is not persuasive if its lower confidence bound
is near zero. Conversely, a large condition number may reflect precisely estimated
anisotropy rather than weak overall information. Repeated eigenvalues should be handled
through eigenspace or cluster-level functionals, not individual eigenvectors.

\section{Weak-IV limits under asymmetric score dependence}\label{sec:limits}

The local experiment now permits a direct study of estimator behavior. The central
question is whether the projected first stage stabilizes faster than the structural
score. When it does not, the estimator divides one first-order random quantity by
another, and conventional Gaussian approximations cease to describe the economic
uncertainty.

For exposition, maintain one endogenous regressor and one instrument. Define
\[
 S_{u,NT}=\sum_{i,t}z_{it}^{o}u_{it}^{o},\qquad
 S_{v,NT}=\sum_{i,t}z_{it}^{o}v_{it}^{o}.
\]
Let $a_{NT}$ normalize the structural score. From \eqref{eq:decomp}, $a_{NT}$ may be of order $T\sqrt N$, $N\sqrt T$, or $\sqrt{NT}$, with intermediate orders under drifting projection variances.

\begin{assumption}[Joint local experiment]\label{ass:joint}
For deterministic $a_{NT},b_{NT}\to\infty$,
\[
 \left(a_{NT}^{-1}S_{u,NT},\ b_{NT}^{-1}S_{v,NT}\right)\Rightarrow(Z_u,Z_v),
\]
where the joint limit is nondegenerate and may be mixed Gaussian or contain Gaussian-product components. Also $\widehat Q_{zz,NT}\to_p Q_{zz}>0$.
\end{assumption}

The oracle IV estimator obeys
\[
 \widehat\beta_{IV}^{o}-\beta_0
 =\frac{S_{u,NT}}{NT\pi_{NT}\widehat Q_{zz,NT}+S_{v,NT}}.
\]

\paragraph{Purpose.}
The local experiment becomes nonregular when the random first-stage component remains in
the denominator. The next theorem makes this mechanism explicit in the scalar model and
shows how asymmetric score rates determine the estimator's limit.

\begin{theorem}[Weak-IV ratio limit]\label{thm:ratio}
Let $\mu_c$ be defined by \eqref{eq:mu-c}. Under
Proposition~\ref{prop:balance} and Assumption~\ref{ass:joint}, and if
$\Pr(\mu_c+Z_v=0)=0$,
\begin{equation}
 \frac{b_{NT}}{a_{NT}}\bigl(\widehat\beta_{IV}^{o}-\beta_0\bigr)
 \Rightarrow \frac{Z_u}{\mu_c+Z_v}. \label{eq:ratio}
\end{equation}
If $a_{NT}/b_{NT}\to\rho\in(0,\infty)$, the unscaled estimator has the ratio limit
$\rho Z_u/(\mu_c+Z_v)$. If $a_{NT}/b_{NT}\to0$, the estimator is consistent, although
its nondegenerate rate is $b_{NT}/a_{NT}$ rather than a conventional sample-size rate.
If $a_{NT}/b_{NT}\to\infty$ and the ratio limit is nonzero with positive probability,
the unscaled estimator is not tight.
\end{theorem}

\noindent\emph{Proof.} See Online Appendix~\ref{app:proof:thm-ratio}.

\paragraph{Interpretation.}
Weak identification after projection means that the variation separating the structural
effect from the disturbance remains random at first order. The estimator therefore
compares two noisy pieces of post-projection variation rather than converging around a
stable denominator. Correcting only the standard error cannot recover information that
the design no longer supplies.

\begin{remark}[Interpretation]
The asymmetry matters through two separate channels. $b_{NT}$ is controlled by first-stage information accumulation. $a_{NT}$ is controlled by the unit, time, and interaction projections of the structural score. Two designs with the same nominal $N$ and $T$ can therefore have the same concentration parameter but different estimator rates and limit shapes.
\end{remark}

\subsection{Feasible interactive-effect removal and oracle equivalence}\label{subsec:factor}

The preceding limits use oracle projected variables. The practical question is whether
estimating the nuisance spaces changes the first-order experiment. The results below make
clear that projector consistency alone is insufficient: score and Jacobian remainders
must be negligible at the normalizations relevant for inference.

Throughout this subsection, let $\mathcal B\subset\mathbb R^p$ be a compact parameter set containing $\beta_0$ in its interior.

After partialling out observed controls, stack the outcome, endogenous regressors, and instruments into a finite collection of $N\times T$ matrices. Let $\mathcal L_0$ and $\mathcal F_0$ be the joint linear loading and factor spaces defined in Section~\ref{subsec:unifieddgp}; additive effects are included through $\mathbf 1_N$ and $\mathbf 1_T$. The oracle residualization is
\[
 \mathcal M_0(A)=M_{\mathcal L_0}AM_{\mathcal F_0},\qquad
 M_{\mathcal L_0}=I_N-P_{\mathcal L_0},\quad M_{\mathcal F_0}=I_T-P_{\mathcal F_0}.
\]
The feasible operator $\widehat{\mathcal M}$ uses estimated projectors onto the same joint spans. This stacked construction permits variable-specific preliminary factor estimates while ensuring that the final transformation is common across every term in each equation.

\begin{remark}[A general factor-space route]\label{rem:factorprimitive}
The dimensions of the joint loading and factor spaces are fixed and consistently selected. Their nonzero singular values are separated from the idiosyncratic spectral norm. For deterministic sequences $c_{\Lambda,NT},c_{F,NT}\to0$,
\[
 \|\widehat P_\Lambda-P_{\mathcal L_0}\|=O_p(c_{\Lambda,NT}),\qquad
 \|\widehat P_F-P_{\mathcal F_0}\|=O_p(c_{F,NT}).
\]
The spectral norms of the idiosyncratic outcome, regressor, and instrument matrices are $O_p(\sqrt N+\sqrt T)$, and their aggregate products satisfy the moment and stochastic-equicontinuity bounds used below. Standard strong-factor principal-components conditions are one sufficient route to these projector rates. These features are not maintained by the general transfer proposition unless they are used to verify Assumption~\ref{ass:factororth}; the proposition relies on the normalized remainder bounds themselves.
\end{remark}

\subsubsection{A primitive sample-split strong-factor benchmark}

The Online Appendix verifies the high-level transfer conditions for a sample-split
spectral estimator of fixed-dimensional strong-factor spaces. Let
$\eta_{NT}=N^{-1/2}+T^{-1/2}$. Under separated signal eigenvalues, idiosyncratic spectral
norms of order $\sqrt N+\sqrt T$, conditional mean-square remainder bounds, and
\[
 \frac{\sqrt{NT}\eta_{NT}}{a_{NT}}\to0,\qquad
 \frac{\sqrt{NT}\eta_{NT}}{b_{NT}}\to0,\qquad
 \frac{\eta_{NT}}{\varrho_{NT}}\to0,
\]
the following result holds.

\begin{theorem}[Sample-split spectral factor transfer]\label{thm:pcbenchmark}
The estimated loading and factor projectors converge at rate $O_p(\eta_{NT})$.
Aggregate structural- and first-stage-score errors are
$O_p(\sqrt{NT}\eta_{NT})$, while instrument-second-moment and Jacobian errors are
$O_p(\eta_{NT})$. Hence the feasible experiment inherits the oracle limits and every
spectral classification whose signal and separating gap dominate $\eta_{NT}$.
\end{theorem}

\noindent\emph{Estimator, complete conditions, and proof.} See the Online Appendix.

\begin{assumption}[High-level factor-transfer remainders]\label{ass:factororth}
Let
\[
 R_{g,NT}(\beta)=\sum_{i,t}\{\widehat g_{it}(\beta)-g_{it}(\beta)\},
\]
\[
 R_{v,NT}=\sum_{i,t}\left[
 \operatorname{vec}(\widehat z_{it}^{o}\widehat v_{it}^{o\prime})
 -\operatorname{vec}(z_{it}^{o}v_{it}^{o\prime})\right],
\]
and
\[
 R_{zz,NT}=(NT)^{-1}\sum_{i,t}
 \{\widehat z_{it}^{o}\widehat z_{it}^{o\prime}
   -z_{it}^{o}z_{it}^{o\prime}\}.
\]
For the parameter set and local-relevance sequences used in the relevant theorem,
\[
 \sup_{\beta\in\mathcal B}a_{NT}^{-1}\|R_{g,NT}(\beta)\|=o_p(1),
 \qquad b_{NT}^{-1}\|R_{v,NT}\|=o_p(1),
 \qquad \|R_{zz,NT}\|=o_p(1).
\]
These are transfer conditions. Projector consistency by itself does not imply them.
\end{assumption}

\paragraph{Purpose.}
The primitive factor theorem covers one benchmark estimator. The next proposition states
the general normalized remainder conditions under which any feasible nuisance estimator
inherits the oracle limits and separated spectral classifications.

\begin{proposition}[Oracle-to-feasible transfer under normalized remainder bounds]\label{prop:oracle}
Under Assumptions~\ref{ass:moments} and \ref{ass:factororth},
the feasible and oracle structural scores, first-stage scores, and instrument second
moments are asymptotically equivalent at the normalizations appearing in
Assumption~\ref{ass:factororth}. Hence any finite-dimensional oracle weak limit based
continuously on those objects transfers to its feasible counterpart. For a claimed
spectral classification, additionally suppose that, for a deterministic relevant signal
or gap scale $\varrho_{NT}$,
\[
 \left\|(NT)^{-1}\sum_{i,t}
 \{\widehat z_{it}^{o}\widehat x_{it}^{o\prime}
   -z_{it}^{o}x_{it}^{o\prime}\}\right\|
 =o_p(\varrho_{NT}).
\]
Then Proposition~\ref{prop:stability} transfers every singular-value cluster and
identified subspace whose signal and separating gap dominate the combined sampling and
nuisance-estimation error. No classification claim is made at a boundary where the
signal or gap is of the same order as that error.
\end{proposition}

\noindent\emph{Proof.} See Online Appendix~\ref{app:proof:prop-oracle}.

The proposition is a transfer result, not an automatic consequence of consistent factor estimation. Under weak identification, a projection error that is negligible for prediction can still be first order relative to the excluded signal. To verify the proposition for a concrete estimator, a replication appendix should report: (i) the estimator of the joint loading and factor spaces; (ii) its projector rates; (iii) the bound for each linear and quadratic score remainder; (iv) the comparison of those bounds with $a_{NT}$, $b_{NT}$, and the smallest spectral gap used for classification; and (v) whether sample splitting or cross-fitting is used to control dependence between the estimated projection and the score. Orthogonality and the theorem-specific comparison between nuisance error, signal, and spectral gap are therefore essential.

\subsection{Multivariate weak-IV experiment}\label{subsec:multi}

To avoid confusion with the normalized Jacobian $H_{NT}$ used in the local shift
experiment, $\mathsf G_{NT}$ below denotes the \emph{unnormalized} sample Jacobian sum.
The symbol $\mathsf W_{NT}$ denotes a GMM weight matrix and is unrelated to the observed
control vector $w_{it}$ or the bootstrap interaction residual
$\widehat W_{it}(\beta)$.

Let $x_{it}\in\mathbb R^p$ and $z_{it}\in\mathbb R^q$, $q\ge p$. Define the unnormalized sample Jacobian and structural score by
\[
 \mathsf G_{NT}=\sum_{i,t}z_{it}^{o}x_{it}^{o\prime},\qquad
 U_{NT}=\sum_{i,t}z_{it}^{o}u_{it}^{o},\qquad
 \widehat Q_{zz,NT}=(NT)^{-1}\sum_{i,t}z_{it}^{o}z_{it}^{o\prime}.
\]
A scalar normalization for $\mathsf G_{NT}$ is generally incompatible with heterogeneous information accumulation across structural directions. Let $C_{NT}=\operatorname{diag}(c_{1,NT},\ldots,c_{p,NT})$ be a deterministic diagonal matrix with positive entries, defined after a fixed or consistently estimated rotation of the parameter coordinates. Assume
\[
 \left(a_{NT}^{-1}U_{NT},\ \mathsf G_{NT}C_{NT}^{-1}\right)
 \Rightarrow (Z_u,\mathcal H),
\]
where $Z_u\in\mathbb R^q$ and $\mathcal H\in\mathbb R^{q\times p}$. The columns of $C_{NT}$ encode direction-specific first-stage scales. Statements about individual directions are invariant only for simple separated singular values; repeated singular values are interpreted through their associated spectral projectors.

The symbol $h$ in the next definition denotes a fixed unit direction in parameter
space; it is not a sample index. This use agrees with the local-direction notation in the
preceding shift experiment.

\begin{definition}[Directional multivariate concentration]\label{def:multikappa}
For a unit structural direction $h\in\mathbb R^p$, define
\[
 m_{i,NT}(h)=T^{-1}\sum_t z_{it}^{o}v_{it}^{o\prime}h,\qquad
 \Omega_{v,NT}(h)=\operatorname{Var}\!\left(N^{-1/2}\sum_i m_{i,NT}(h)\right).
\]
The deterministic projected first-stage signal in direction $h$ is
$\mu_{NT}(h)=Q_{zz,NT}^{E}\Pi_{NT}h$. Define
\[
 \kappa_{NT}(h)=
 \begin{cases}
 N\,\mu_{NT}(h)'\Omega_{v,NT}(h)^{\dagger}\mu_{NT}(h),
 &\mu_{NT}(h)\in\mathcal R\{\Omega_{v,NT}(h)\},\\[3pt]
 +\infty,
 &\mu_{NT}(h)\notin\mathcal R\{\Omega_{v,NT}(h)\}.
 \end{cases}
\]
The second case records a signal component with zero asymptotic noise rather than
silently discarding it through the generalized inverse. The quantity is a
direction-specific signal-to-noise index. Minimizing it over a unit sphere is a useful diagnostic only after fixing the covariance-rank stratum and verifying continuity of $h\mapsto\Omega_{v,NT}(h)^{\dagger}$; it is not asserted to be a canonical concentration parameter without those additional conditions. This definition reduces to \eqref{eq:kappa} when $p=q=1$.
\end{definition}

\paragraph{Purpose.}
Weak identification is directional in multivariate models. The next theorem allows
different structural directions to accumulate information at different rates and gives
the resulting local GMM limit on a full-rank stratum.

\begin{theorem}[Multivariate weak-IV limit on a full-rank local stratum]\label{thm:multiratio}
Suppose
\[
 \left(a_{NT}^{-1}U_{NT},\ \mathsf G_{NT}C_{NT}^{-1}\right)
 \Rightarrow(Z_u,\mathcal H),
\]
$\mathsf W_{NT}\to_p\mathsf W>0$ and
$\Pr\{\operatorname{rank}(\mathcal H)=p\}=1$. Define the unconstrained linear-GMM estimator on the event
$\operatorname{rank}(\mathsf G_{NT})=p$ by
\[
 \widehat\beta-\beta_0=(\mathsf G_{NT}'\mathsf W_{NT}\mathsf G_{NT})^{-1}
 \mathsf G_{NT}'\mathsf W_{NT}U_{NT},
\]
and define it arbitrarily on the complementary event. Under the stated full-rank limit,
the complementary event has probability tending to zero. Then the oracle estimator
satisfies
\[
 a_{NT}^{-1}C_{NT}(\widehat\beta-\beta_0)
 \Rightarrow
 (\mathcal H'\mathsf W\mathcal H)^{-1}\mathcal H'\mathsf W Z_u.
\]
If the limiting first-stage matrix is rank deficient with positive probability, the theorem does not define a unique limit for the full parameter. In that case only functionals lying in the identified row space admit point-identification statements without an explicit regularization, anchoring, or parameter restriction; Anderson--Rubin inference remains well defined for the joint null. A feasible estimator built from factor-adjusted moments has the same limit only when Proposition~\ref{prop:oracle} applies to all entries of $(\mathsf G_{NT},U_{NT},\mathsf W_{NT})$ at the stated normalizations.
\end{theorem}

\noindent\emph{Proof.} See Online Appendix~\ref{app:proof:thm-multiratio}.

\paragraph{Interpretation.}
Different linear combinations of the structural parameter can converge at different
rates. The full-rank condition guarantees a unique local inverse, while rank loss would
leave unrestricted null-space components unidentified and require a different inferential
object.

\begin{corollary}[Multivariate Anderson--Rubin reduction]\label{cor:multiAR}
For testing $H_0:\beta=\beta_0$ jointly, the factor-adjusted score $U_{NT}(\beta_0)$ is free of $\Pi_{NT}$ in its null centering. Hence weak first-stage drift does not enter the null moment equation. If, in addition, the studentizing covariance estimator and critical-value
approximation satisfy the uniform conditions of Theorem~\ref{thm:uniform-ar} over the
maintained concentration-parameter set and fixed covariance-rank stratum, the resulting
quadratic AR test has asymptotic rejection probability uniformly no greater than its
nominal level on that set. Its inversion yields a joint confidence region for $\beta$; projection onto a subset of coordinates is valid but generally conservative.
\end{corollary}

\section{Covariance-optimal weighting under regular identification}\label{sec:efficiency}

Weak-identification methods are essential near rank boundaries, but they should reproduce
standard efficient GMM when projected information is regular. This section establishes
that benchmark and clarifies its scope: optimal weighting improves precision within a
well-identified model; it cannot replace information removed by projection.

Under a Gaussian regular score limit, the optimally weighted covariance matrix below is
the inverse of the projected information matrix in \eqref{eq:projected-information},
after matching the regular score and parameter normalizations. The result is therefore
the regular full-rank specialization of the local projected-moment experiment.

The identification analysis determines when regular estimation is possible. Efficiency is a separate question and is meaningful in the conventional local-asymptotic sense only on sequences for which the relevant structural directions remain regularly identified. This section therefore establishes an optimal-weighting result under strong identification and states explicitly why it does not extend mechanically to weak or rank-deficient sequences.

Let
\[
 \bar g_{NT}(\beta)=\frac1{NT}\sum_{i,t}z_{it}^{o}
 (y_{it}^{o}-x_{it}^{o\prime}\beta),
\]
and define the projected GMM estimator
\begin{equation}
 \widehat\beta_W\in\arg\min_{\beta\in\mathcal B}
 \bar g_{NT}(\beta)'\mathsf W_{NT}\bar g_{NT}(\beta),                  \label{eq:gmm}
\end{equation}
where $\mathsf W_{NT}$ is symmetric positive definite with probability approaching one. Let $a_{NT}$ be the structural-score normalization and suppose
\[
 a_{NT}^{-1}\sum_{i,t}g_{it}(\beta_0)\Rightarrow Z_g,
 \qquad \Var(Z_g)=\Omega_g.
\]
Because the score may contain unit, time, interaction, and cell projections, $a_{NT}$ need not equal $\sqrt{NT}$. The regular estimator rate is consequently $NT/a_{NT}$ rather than automatically $\sqrt{NT}$.
Define the regular sample Jacobian
\[
 \widehat J_{NT}^{\,s}
 =\frac1{NT}\sum_{i,t}z_{it}^{o}x_{it}^{o\prime}.
\]
The superscript $s$ distinguishes this sample-average Jacobian from the population
Jacobian $J_{NT}$ and from the unnormalized weak-IV matrix $\mathsf G_{NT}$.

\begin{assumption}[Regular identification and stable score covariance]\label{ass:regular-eff}
(i) $J_{NT}\to J_0$, where $J_0$ has full column rank.
(ii) $\mathsf W_{NT}\to_p \mathsf W$, where $\mathsf W$ is positive definite.
(iii) $\Omega_g$ is positive definite, and a consistent estimator
$\widehat\Omega_g$ is available under the maintained two-way dependence.
(iv) $\widehat J_{NT}^{\,s}\to_pJ_0$.
If the criterion is minimized over a restricted parameter set $\mathcal B$, additionally
assume that the unconstrained minimizer belongs to the interior of $\mathcal B$ with
probability approaching one; this implementation condition is not needed when
$\mathcal B=\mathbb R^p$.
\end{assumption}

\paragraph{Roadmap.}
When every projected direction remains well identified, the random-denominator problem
disappears. The next theorem recovers the familiar linear-GMM limit and shows what
optimal weighting can accomplish in that regular regime. It is a benchmark, not an
argument for using Wald inference near rank boundaries.

\begin{theorem}[Regular projected-GMM limit and covariance-optimal weighting]\label{thm:efficient-gmm}
Suppose
$a_{NT}^{-1}\sum_{i,t}g_{it}(\beta_0)\Rightarrow Z_g$
with $\Var(Z_g)=\Omega_g$, and Assumption~\ref{ass:regular-eff} holds. Then
\begin{equation}
 \frac{NT}{a_{NT}}(\widehat\beta_W-\beta_0)
 \Rightarrow
 (J_0'\mathsf WJ_0)^{-1}J_0'\mathsf WZ_g.                                  \label{eq:gmm-limit}
\end{equation}
Its asymptotic covariance is
\begin{equation}
 V(\mathsf W)=(J_0'\mathsf WJ_0)^{-1}J_0'\mathsf W\Omega_g \mathsf WJ_0(J_0'\mathsf WJ_0)^{-1}.       \label{eq:gmm-var}
\end{equation}
Within the class of regular linear-GMM estimators generated by \eqref{eq:gmm}, the choice $\mathsf W=\Omega_g^{-1}$ minimizes the covariance matrix of the linear limit in the Loewner order and yields
\[
 V_{\mathrm{opt}}=(J_0'\Omega_g^{-1}J_0)^{-1}.
\]
\end{theorem}

\noindent\emph{Proof.} See Online Appendix~\ref{app:proof:thm-efficient-gmm}.

\paragraph{Interpretation.}
Optimal weighting is a regular-model conclusion, not a remedy for weak identification.
When projected information is well conditioned, it minimizes covariance within the
stated linear-GMM class; when information is weak, identification-robust procedures
remain necessary.

\paragraph{Scope.}
This is a regular-model covariance comparison within the stated linear-GMM class, not a
semiparametric efficiency bound and not a remedy for weak identification. Feasible
two-step GMM inherits the oracle limit when Proposition~\ref{prop:oracle} and covariance
consistency hold. The complete corollary and qualifications are in the Online Appendix.

\section{Inference}\label{sec:inference}

The preceding results separate regular from weakly identified regimes. This section turns
that distinction into an inferential strategy. Wald procedures estimate through the
possibly unstable first stage, whereas null-restricted score procedures remain centered
without dividing by it. The geometry of the resulting confidence set is itself evidence
about how much the final specification can distinguish.

The local experiment clarifies the division between identification and dependence. In
the Gaussian fixed-rank case, local power is summarized by $h'\mathcal I_Ph$. Outside
that case, null-restricted inference remains valid only if critical values approximate
the actual score law rather than a Gaussian covariance surrogate.

\subsection{Why conventional Wald inference fails}

Wald inference summarizes uncertainty with an estimate and a standard error. After
information-reducing projection, however, the estimate can inherit a random first-stage
denominator. The next result identifies this failure rather than treating it as a
covariance-estimation problem.

A cluster-robust covariance estimator can consistently estimate the variance of a regular, asymptotically linear estimator. Under weak identification, 2SLS is not regular because its denominator remains random at first order. Studentizing only the numerator therefore does not recover a pivotal normal or chi-square limit.

Let $\widehat V_{2W}$ be any conventional two-way cluster variance estimator obtained by linearizing 2SLS around a nonrandom first-stage Jacobian, and define
\[
 W_{2SLS}(\beta_0)=\frac{(\widehat\beta_{IV}^{o}-\beta_0)^2}{\widehat V_{2W}}.
\]

\begin{proposition}[Failure of conventional Wald pivotality]\label{prop:wald}
Suppose Theorem~\ref{thm:ratio} holds and
\[
 \frac{a_{NT}}{b_{NT}}\to\rho\in(0,\infty),
 \qquad
 \frac{\widehat V_{2W}}{(a_{NT}/b_{NT})^2}\to_p v_0\in(0,\infty).
\]
Then
\[
 W_{2SLS}(\beta_0)\Rightarrow
 \frac{Z_u^2}{v_0(\mu_c+Z_v)^2}.
\]
This limit is $\chi_1^2$ only under additional restrictions making the displayed ratio
equal in law to a standard normal square. Such restrictions are not implied by the
maintained local experiment.
\end{proposition}

\noindent\emph{Proof.} See Online Appendix~\ref{app:proof:prop-wald}.

The proposition does not imply universal over-rejection. The direction and magnitude of distortion depend on the covariance and tail behavior of $(Z_u,Z_v)$ and on the variance estimator. Claims of ``severe over-rejection'' should therefore be stated as theoretical possibilities and demonstrated for explicit data-generating processes or simulations, rather than asserted as a universal theorem.

\subsection{Anderson--Rubin inference}

Null-restricted score inference avoids the unstable denominator. Its strength is not
that it always has a chi-square law, but that its centering does not depend on projected
first-stage strength. The reference distribution must still match the dependence and
covariance-rank regime.

For a candidate $\beta$, define
\[
 \bar g_{NT}(\beta)=(NT)^{-1}\sum_{i,t}g_{it}(\beta).
\]
Because the variance order of the aggregate score depends on its unit, time, interaction, and cell projections, no universal $\sqrt{NT}$ normalization is imposed. Define the quadratic form directly using a consistent estimator of the covariance matrix of the unnormalized score sum:
\begin{equation}
 AR_{NT}(\beta)
 =\left(\sum_{i,t}\widehat g_{it}(\beta)\right)'
 \widehat\Sigma_{g,NT}(\beta)^{\dagger}
 \left(\sum_{i,t}\widehat g_{it}(\beta)\right), \label{eq:ar}
\end{equation}
where $\dagger$ is the Moore--Penrose inverse and
$\widehat\Sigma_{g,NT}(\beta)$ estimates the covariance of the unnormalized score sum,
not the covariance of the sample average. If it is nonsingular with probability
approaching one, the ordinary inverse may be used.

\paragraph{Roadmap.}
The next proposition isolates the algebraic reason for weak-identification robustness:
under the structural null, the score is centered without estimating through the projected
first stage. Distributional validity still requires a law-appropriate covariance or
bootstrap approximation.

\begin{proposition}[Identification robustness of the null score]\label{prop:ar}
Under $H_0:\beta=\beta_0$, $g_{it}(\beta_0)=z_{it}^{o}u_{it}^{o}$ and its centering does not involve $\Pi_{NT}$. Thus identification strength affects the alternatives and power but not the null centering. Size robustness over a concentration-parameter class follows only when the covariance estimator and critical-value approximation are themselves uniformly valid over that class and over the maintained covariance-rank stratum.
\end{proposition}

\noindent\emph{Proof.} See Online Appendix~\ref{app:proof:prop-ar}.

\paragraph{Interpretation.}
The null score remains centered without dividing by the projected first stage. This is
the precise sense in which Anderson--Rubin inference is robust to weak projected
identification.

\begin{assumption}[Uniform null-score approximation over a maintained class]
\label{ass:uniform-ar}
Let $\mathfrak P_{NT}$ be a class of triangular-array laws satisfying the structural null,
and let $T_{NT}$ denote the feasible null-score statistic. For each
$P\in\mathfrak P_{NT}$, let $F_{NT,P}$ be a reference cdf and define its generalized
$(1-\alpha)$ quantile by
\[
 c_{NT,P}(1-\alpha)
 =\inf\{x\in\mathbb R:F_{NT,P}(x)\ge 1-\alpha\}.
\]
Assume the uniform distributional approximation
\[
 \sup_{P\in\mathfrak P_{NT}}\sup_{x\in\mathbb R}
 \left|P\{T_{NT}\le x\}-F_{NT,P}(x)\right|\to0.
\]
Assume also the uniform anti-concentration condition
\[
 \lim_{\delta\downarrow0}\;
 \limsup_{N,T\to\infty}\;
 \sup_{P\in\mathfrak P_{NT}}
 \left[
 F_{NT,P}\{c_{NT,P}(1-\alpha)+\delta\}
 -
 F_{NT,P}\{c_{NT,P}(1-\alpha)-\delta\}
 \right]
 =0.
 \label{eq:uniform-anticoncentration}
\]
The feasible critical value $\widehat c_{NT}(1-\alpha)$ is computed by the same rule for
every law in the class and satisfies, for every $\varepsilon>0$,
\[
 \sup_{P\in\mathfrak P_{NT}}
 P\left\{
 \left|\widehat c_{NT}(1-\alpha)-c_{NT,P}(1-\alpha)\right|>\varepsilon
 \right\}
 \to0.
\]
The class is restricted to a fixed covariance-rank stratum. Feasible null-score,
covariance, and projector errors are uniformly negligible at their theorem-specific
normalizations, while nuisance ranks and separating eigengaps remain fixed uniformly;
see Online Appendix Lemma~\ref{lem:uniform-transfer}.
\end{assumption}

\begin{theorem}[Uniform size of the projected null-score test]
\label{thm:uniform-ar}
Under Assumption~\ref{ass:uniform-ar}, the test that rejects when
$T_{NT}>\widehat c_{NT}(1-\alpha)$ satisfies
\[
 \limsup_{N,T\to\infty}\sup_{P\in\mathfrak P_{NT}}
 P\{T_{NT}>\widehat c_{NT}(1-\alpha)\}\le\alpha.
\]
Equivalently, inversion yields
\[
 \liminf_{N,T\to\infty}\inf_{P\in\mathfrak P_{NT}}
 P\{\beta_0\in\widehat{\mathcal C}_{1-\alpha}\}\ge1-\alpha.
\]
The conclusion covers strong, local-to-zero, and exactly zero projected first-stage
sequences whenever those sequences belong to the maintained class
$\mathfrak P_{NT}$. Uniformity is not asserted across covariance-rank changes or outside
the stated factor-transfer and critical-value class.
\end{theorem}

\noindent\emph{Proof.} See Online Appendix~\ref{app:proof:thm-uniform-ar}.

\paragraph{Interpretation.}
Uniform weak-identification robustness is a property of the complete null-score
approximation, not of centering alone. The theorem therefore states exactly what must be
uniform: the score law, rank stabilization, nuisance transfer, and critical values.

\paragraph{Concentration-index sequences.}
For the scalar benchmark, the maintained law class may be indexed by
$\kappa_{N\mid T}\to c$ with $c\in[0,\infty)$. Uniform statements are made over fixed
compact concentration sets $c\in[0,\bar c]$, including $c=0$, together with a fixed
covariance-rank stratum and uniform nuisance-transfer conditions. Strong-identification
limits are handled separately by the regular theory; no single approximation is claimed
to be uniformly sharp at the compactified point $c=\infty$.

\begin{remark}[Uniformity and parameter-space boundaries]
Theorem~\ref{thm:uniform-ar} is an implication theorem: pointwise weak convergence does
not establish its assumptions. The appropriate conclusion is the one-sided size bound
$\limsup\sup_{P\in\mathfrak P_{NT}}P\{\text{reject}\}\le\alpha$.
Exact equality need not hold when the reference cdf is continuous but locally flat at
its target quantile or when the limiting test is otherwise conservative. The stated
anti-concentration condition excludes a nonvanishing atom at the limiting critical
value. Primitive applications must verify score approximation, anti-concentration, rank
stability, critical values, and nuisance transfer uniformly, as in
\citet{andrewschengguggenberger2020}. Exactly zero first stages create no directional
$0/0$ because the null-score statistic does not divide by a first-stage estimate.
Intersecting the joint confidence set with a fixed closed economic restriction set
preserves coverage, although subvector projection may be conservative or disconnected.
\end{remark}

\paragraph{Purpose.}
The preceding result gives identification robustness but not a reference distribution.
The next proposition states the additional Gaussian and fixed-rank conditions under
which covariance studentization produces a chi-square limit.

\begin{proposition}[Gaussian fixed-rank limit of the quadratic AR statistic]\label{prop:ar-limit}
Under $H_0$, suppose there is a deterministic scalar normalization $c_{NT}>0$ such that
\[
 c_{NT}^{-1}\sum_{i,t}g_{it}(\beta_0)\Rightarrow Z_g,\qquad
 Z_g\sim N(0,\Omega_0),
\]
where $\Omega_0$ has fixed rank $r_g$. Suppose the feasible factor transfer holds at this scale and
\[
 \left\|c_{NT}^{-2}\widehat\Sigma_{g,NT}(\beta_0)-\Omega_0\right\|\to_p0,
\]
where $\widehat\Sigma_{g,NT}$ is rank-stabilized so that its estimated rank equals $r_g$ with probability approaching one. Then
\[
 AR_{NT}(\beta_0)\Rightarrow\chi^2_{r_g}.
\]
If the normalized null score has a non-Gaussian interaction limit, covariance consistency alone does not imply a chi-square law; critical values must instead approximate that non-Gaussian limit, for example through the exact PWB-H statistic under its own assumptions.
\end{proposition}

\noindent\emph{Proof.} See Online Appendix~\ref{app:proof:prop-ar-limit}.

\paragraph{Interpretation.}
Chi-square critical values are justified only when the normalized score is Gaussian and
the covariance rank is stable. If a bilinear interaction survives, the correct reference
law must approximate that non-Gaussian component rather than ignore it.

A confidence set is obtained by inversion:
\[
 \mathcal C_{1-\alpha}=\{\beta: AR_{NT}(\beta)\le c_{1-\alpha}(\beta)\}.
\]
As in classical weak-IV analysis, the set may be wide, disconnected, or unbounded.
Unboundedness is not a defect of inversion: near nonidentification, valid confidence sets
for parameters with unbounded range must be unbounded with positive probability
\citep{dufour1997}. Applied work should report the inversion grid, whether accepted
components touch either grid boundary, and the connected components of the accepted set.
A boundary-touching component should be described as unbounded relative to the explored
parameter range; it should not be summarized by a midpoint or replaced by a Wald interval.
These outcomes are informative manifestations of weak identification, not programming
failures.

\section{Factor-robust bootstrap inference}\label{sec:bootstrap}

A covariance matrix is not enough when the projected score retains a non-Gaussian
unit--time interaction. The purpose of this section is to approximate the full null law,
including persistent unit variation, common time shocks, their interaction, and cell
noise. We first give a self-contained bootstrap for the primitive benchmark and then
state when the broader PWB-H procedure transfers to feasible factor-adjusted scores.

The bootstrap is applied to the null-restricted score, not to the 2SLS estimator. A
standard nonparametric or residual bootstrap of 2SLS cannot repair the nonuniformity
created by a nearly singular first stage. The role of bootstrap inference is narrower:
approximate the potentially non-Gaussian null law generated by unit, time, interaction,
and cell components. The paper gives two routes. The first is a self-contained plug-in
bootstrap for the primitive finite-range benchmark. The second transfers the exact PWB-H
procedure to broader dependence regimes.

\subsection{A self-contained plug-in bootstrap for the primitive benchmark}
\label{subsec:primitive-bootstrap}

The primitive score theorem implies a limit of the form
\[
 L
 =G_AZ_a+G_DZ_d+G_0\{\mathfrak h(Z_\phi,Z_\psi)+Z_\varepsilon\}.
 \label{eq:primitive-limit-bootstrap}
\]
A covariance-only Gaussian approximation is generally invalid when the bilinear term
survives. The following bootstrap estimates each Gaussian building block and reconstructs
the nonlinear interaction explicitly.

Let
\[
 \widehat\Sigma_A,\quad
 \widehat\Sigma_D,\quad
 \widehat\Sigma_\varepsilon,\quad
 \widehat H_1,\ldots,\widehat H_d,\quad
 \widehat G_A,\widehat G_D,\widehat G_0
\]
be feasible estimators of the objects in Theorem~\ref{thm:primitive-benchmark}. Conditional
on the data, draw independently
\[
 \begin{pmatrix}Z_a^*\\ Z_\phi^*\end{pmatrix}
 \sim N(0,\widehat\Sigma_A),\qquad
 \begin{pmatrix}Z_d^*\\ Z_\psi^*\end{pmatrix}
 \sim N(0,\widehat\Sigma_D),\qquad
 Z_\varepsilon^*\sim N(0,\widehat\Sigma_\varepsilon).
\]
Define
\[
 \widehat{\mathfrak h}(z_\phi,z_\psi)
 =\bigl(z_\phi'\widehat H_1z_\psi,\ldots,
        z_\phi'\widehat H_dz_\psi\bigr)'
\]
and
\begin{equation}
 L_{NT}^*
 =\widehat G_AZ_a^*
 +\widehat G_DZ_d^*
 +\widehat G_0\{
   \widehat{\mathfrak h}(Z_\phi^*,Z_\psi^*)+Z_\varepsilon^*\}.
 \label{eq:primitive-bootstrap}
\end{equation}

\begin{assumption}[Feasible primitive-bootstrap inputs]\label{ass:primitiveboot}
The dimensions in Assumption~\ref{ass:primitivebenchmark} are fixed. The estimators
$\widehat\Sigma_A$, $\widehat\Sigma_D$, and
$\widehat\Sigma_\varepsilon$ are symmetric positive semidefinite with probability
approaching one, and:
\begin{enumerate}[label=(\roman*),leftmargin=2.6em]
\item
\[
 \|\widehat\Sigma_A-\Sigma_A\|
 +\|\widehat\Sigma_D-\Sigma_D\|
 +\|\widehat\Sigma_\varepsilon-\Sigma_\varepsilon\|
 \to_p0;
\]
\item $\max_{k\le d}\|\widehat H_k-H_k\|\to_p0$;
\item
\[
 \|\widehat G_A-G_A\|
 +\|\widehat G_D-G_D\|
 +\|\widehat G_0-G_0\|\to_p0.
\]
\end{enumerate}
\end{assumption}

The displayed consistency conditions are the complete inputs for the simulated limit
law. When these objects are constructed after factor estimation, their consistency must
be verified from the relevant factor-transfer bounds rather than inferred from projector
consistency alone.

\paragraph{Roadmap.}
When the interaction component survives, matching only the covariance discards part of
the first-order uncertainty. The next theorem reconstructs every component of the
primitive projected-score limit and therefore provides critical values for continuous
functionals of the full law, not merely for a Gaussian surrogate.

Let $\widehat\xi_{it,NT}$ denote the feasible centered stacked structural and
first-stage score obtained by replacing the oracle projection and population centering in
$\xi_{it,NT}$ with their feasible counterparts.

\begin{theorem}[Self-contained feasible bootstrap for the primitive score law]
\label{thm:primitive-bootstrap}
Under Assumption~\ref{ass:primitiveboot},
\[
 d_{BL}\{\mathcal L^*(L_{NT}^*),\mathcal L(L)\}\to_p0,
 \label{eq:primitive-bootstrap-bl}
\]
where $L$ is defined in
\eqref{eq:primitive-limit-bootstrap}. Consequently, if
\[
 K_{NT}^{-1}\sum_{i,t}\widehat\xi_{it,NT}\Rightarrow L
\]
under the null and $\varphi:\mathbb R^d\to\mathbb R$ is continuous, then
\[
 d_{BL}\!\left[
 \mathcal L^*\{\varphi(L_{NT}^*)\},
 \mathcal L\{\varphi(L)\}
 \right]\to_p0.
\]
If the distribution function of $\varphi(L)$ is continuous and strictly increasing at its
$(1-\alpha)$ quantile, the conditional bootstrap critical value yields asymptotic size
$\alpha$.
\end{theorem}

\noindent\emph{Proof.} See Online Appendix~\ref{app:proof:thm-primitive-bootstrap}.

\paragraph{Interpretation.}
The bootstrap reproduces the uncertainty generated by the empirical design after
nuisance projection: persistent unit variation, common shocks, their interaction, and
cell innovations. It therefore preserves the sources of uncertainty relevant for the
economic conclusion rather than replacing them with a Gaussian approximation having the
same covariance.

\begin{corollary}[Primitive-bootstrap critical values for a quadratic score statistic]
\label{cor:primitive-bootstrap-quadratic}
Suppose Assumption~\ref{ass:primitiveboot} holds, let
$\widehat\Omega_{NT}\to_p\Omega_L$ and
$\rank(\widehat\Omega_{NT})=\rank(\Omega_L)$ with probability approaching one, and
define
\[
 q_{NT}(x)=x'\widehat\Omega_{NT}^{\dagger}x,\qquad
 q(x)=x'\Omega_L^\dagger x.
\]
If the distribution of $q(L)$ is continuous at its $(1-\alpha)$ quantile, the conditional
$(1-\alpha)$ quantile of $q_{NT}(L_{NT}^*)$ yields asymptotic rejection probability
$\alpha$ for the corresponding null-restricted quadratic score statistic. No Gaussian or chi-square
approximation is required.
\end{corollary}

\noindent\emph{Proof.} See Online Appendix~\ref{app:proof:cor-primitive-bootstrap-quadratic}.

\paragraph{Implementation.}
For the fixed-range benchmark, estimate $\Sigma_A$ and $\Sigma_D$ from the finitely many
admissible unit and time covariance lags, estimate $\Sigma_\varepsilon$ from the
conditionally centered cell residual, and estimate the finite-rank interaction matrices
from the singular decomposition of the interaction projection. Positive-semidefinite
covariance estimates should be obtained by eigenvalue truncation if necessary. The
bootstrap then requires only independent Gaussian draws and the bilinear reconstruction
in \eqref{eq:primitive-bootstrap}. Every covariance block, interaction matrix, Gaussian
draw, and bootstrap critical value should be saved in the replication files.

\subsection{Wild bootstrap for broader dependence}

For each candidate $\beta$, form feasible scores $\widehat g_{it}(\beta)$ and
their feasible grand mean
\[
 \widehat{\bar g}_{NT}(\beta)=\frac1{NT}\sum_{i,t}\widehat g_{it}(\beta).
\]
Define the empirical projections
\begin{align}
 \widehat A_i(\beta)&=T^{-1}\sum_t\widehat g_{it}(\beta)-\widehat{\bar g}_{NT}(\beta),\\
 \widehat D_t(\beta)&=N^{-1}\sum_i\widehat g_{it}(\beta)-\widehat{\bar g}_{NT}(\beta),\\
 \widehat W_{it}(\beta)&=\widehat g_{it}(\beta)-T^{-1}\sum_s\widehat g_{is}(\beta)
 -N^{-1}\sum_j\widehat g_{jt}(\beta)+\widehat{\bar g}_{NT}(\beta). \label{eq:proj}
\end{align}
The last component contains both the second-order common-effect interaction and the conditionally centered cell innovation. Separating those pieces is generally infeasible and is not required by the projection-based construction.

Let $\eta_i^*$ and $\xi_t^*$ be mean-zero, variance-one multipliers that are independent across dimensions. Within dimensions, their covariance kernels mimic the maintained spatial and serial dependence. In the serial dimension, a Markov wild multiplier can be used; in the spatial dimension, correlated multipliers are generated from a positive-semidefinite kernel based on inter-unit distances. Following \citet{hounyolin2026}, define
\begin{equation}
 g_{it}^*(\beta)
 =\widehat m_{A,i}\widehat A_i(\beta)\eta_i^*
 +\widehat m_{D,t}\widehat D_t(\beta)\xi_t^*
 +\widehat W_{it}(\beta)\eta_i^*\xi_t^*, \label{eq:pwb}
\end{equation}
where $\widehat m_{A,i}$ and $\widehat m_{D,t}$ are the variance corrections prescribed by the selected PWB procedure.

\begin{assumption}[Bootstrap transfer conditions]\label{ass:boot}
(i) The factor-adjusted score belongs to the class of two-way score arrays covered by \citet{hounyolin2026}. (ii) The empirical projections in \eqref{eq:proj} are mean-square consistent at the rates required in the selected regime. (iii) Spatial and temporal bandwidths satisfy the paper's rate conditions. (iv) The PWB-H regime classifier and variance corrections are implemented as defined there. (v) Proposition~\ref{prop:oracle} has a conditional bootstrap analogue in bootstrap probability.
\end{assumption}

\begin{proposition}[Null-score reduction for PWB-H]\label{prop:boot}
Under $H_0$ and Assumption~\ref{ass:boot}, the oracle score equals $z_{it}^{o}u_{it}^{o}$ and contains no first-stage drift.
Consequently, whenever the oracle score array and the feasible factor-transfer remainder satisfy
the exact assumptions of the selected PWB-H theorem, bootstrap validity is unaffected by bounded
local first-stage coefficients except through any dependence of nuisance estimation on the reduced
form. Uniformity therefore requires the factor-transfer conditions to hold uniformly over that local
parameter set.
\end{proposition}

\noindent\emph{Proof.} See Online Appendix~\ref{app:proof:prop-boot}.

\begin{remark}[Division of labor]
Factor adjustment removes low-rank nuisance components, the Anderson--Rubin restriction handles weak identification, and PWB-H approximates the dependence-driven null law. None of the three substitutes for the others.
\end{remark}

\begin{remark}[Factor re-estimation]
Re-estimating factors in every bootstrap replication is not automatically ``safest.'' It changes the bootstrap map and requires a proof that the factor estimator reproduces its sampling error conditionally. Holding the factor space fixed is first-order valid under Proposition~\ref{prop:oracle}; re-estimation is valid only under a stronger bootstrap oracle-equivalence condition. The implementation should match the theorem actually proved.
\end{remark}

\subsection{Fixed-rank transfer under the exact PWB-H conditions}\label{subsec:exactpwb}

The next result deliberately incorporates the assumptions and the exact statistic of
\citet[Theorem~3.3]{hounyolin2026} by reference. This avoids replacing their
regime-specific representation, moment, spatial, serial, bandwidth, heterogeneity,
rescaling, and classifier conditions with a shorter but potentially nonequivalent list.

\paragraph{Purpose.}
The primitive plug-in bootstrap covers a transparent benchmark. The next theorem
provides the complementary transfer result for the broader PWB-H procedure under its
exact oracle statistic, scaling, classifier, and fixed-rank conditions.

\begin{theorem}[Fixed-rank feasible PWB-H transfer]\label{thm:exactpwb}
Suppose the oracle null-score array and the oracle PWB-H statistic satisfy all assumptions
of \citet[Theorem~3.3]{hounyolin2026}. Suppose further that, uniformly on the null set
$\mathcal B_0$, the feasible factor adjustment preserves (i) the normalized original score,
(ii) the normalized bootstrap score conditionally, (iii) every variance and discriminant
quantity entering PWB-H, and (iv) the rank stratum required by any generalized inverse used
in the present paper. Then, in bounded-Lipschitz distance conditional on the data, the feasible PWB-H
law and the oracle PWB-H law differ by $o_p(1)$. Consequently, the feasible statistic
inherits the oracle uniform approximation over the regimes covered by that theorem. No claim is made in the excluded
I\&N regime, and no claim of uniform conservativeness is made there.
\end{theorem}

\noindent\emph{Proof.} See Online Appendix~\ref{app:proof:thm-exactpwb}.

\paragraph{Interpretation.}
The transfer result is deliberately procedure-specific. Feasible validity is inherited
only because the estimated projection preserves every input used by the oracle PWB-H
statistic and regime classifier at the required scale.

\section{Implementation, diagnostics, and numerical evidence}\label{sec:implementation}

The theory changes empirical practice only if it alters what researchers conclude from
real designs. This section organizes implementation around three questions: how much
excluded variation survives the final projection, which structural directions remain
visible, and whether the projected-score law permits conventional studentization. The
simulations isolate the mechanism; the applications show how it changes economically
meaningful conclusions.

Implementation has six steps: apply the same control residualization to all variables;
estimate joint nuisance spaces and use one final equation-compatible projection; construct
the projected score, Jacobian, covariance inputs, and their normalizations; report
singular-value uncertainty, spectral gaps, subspace stability, rank sensitivity, and the
Projected Information Matrix; compute a null-restricted Anderson--Rubin statistic; and
invert it using chi-square critical values only on Gaussian fixed-rank strata and the
appropriate bootstrap otherwise.

\paragraph{Comparison of reported procedures.}
Conventional Wald/2SLS answers what regular linearization would conclude; the projected
Jacobian identifies surviving directions; the PIM combines sensitivity and precision on
a Gaussian fixed-rank stratum; and Gaussian or bootstrap AR procedures evaluate the
null-score law under their respective assumptions. The Online Appendix gives the full
comparison table and limitations of each procedure.

The paper does not propose a localized conditional projection method and therefore does
not claim shorter confidence sets than classical AR or Kleibergen-type procedures. Its
contribution is diagnostic and inferential reconstruction after nuisance projection.
The relevant comparison is among pre-projection conclusions, post-projection sensitivity
and information, and identification-robust confidence-set geometry.

The empirical outputs answer three questions: how much excluded variation survives
projection, which structural directions remain visible, and whether the projected-score
law is regular enough for covariance studentization. Unbounded, disconnected, or highly
elongated robust confidence regions are evidence of limited projected information, not
computational failures. The Online Appendix gives the reproducibility checklist and
detailed simulation design.

\subsection{Monte Carlo validation}\label{sec:mc}

The simulation reproduces a common empirical mistake: a large raw first stage generated
by variation that is also confounded. The comparison asks whether the conventional
diagnostic supports the correct economic sign and whether projected diagnostics reveal
how much valid information survives.

The simulations are designed to distinguish raw relevance from information that survives
the final specification. The data contain a heterogeneous common factor that creates a
very strong raw association between the instrument and endogenous regressor but also
enters the structural disturbance. The valid idiosyncratic first stage is varied across
strong, weak, and near-boundary designs. We compare the raw model, two-way demeaning, and
a common rank-one projection applied identically to the outcome, endogenous regressor,
and instrument. The true structural coefficient is $-0.5$.

Table~\ref{tab:mcvalidation} reports a 5,000-replication implementation with
$(N,T)=(50,30)$. In the weak-surviving-information design, the raw median first-stage
$F$ is 501.93, yet the raw IV estimator is centered near $0.437$ and both the Wald and
AR procedures reject the true value essentially always. This does not represent a size failure of
identification-robust inference: the raw moment is invalid because the common factor
violates exclusion. After the common projection, the median $F$ is $21.57$, the median
projected-information measure is $0.0144$, and the median IV estimate is $-0.494$. The Gaussian fixed-rank AR reference value is mildly oversized in some transformed
designs. This is expected outside its maintained Gaussian score conditions and reinforces
the need for the regime-appropriate bootstrap developed in Section~\ref{sec:bootstrap}. The main economic conclusion is nevertheless sharp: a
large pre-projection first stage can support the wrong sign while the final projected
specification correctly reveals weak information.

\begin{table}[t]
\centering
\caption{Monte Carlo validation of post-projection information diagnostics}
\label{tab:mcvalidation}
\scriptsize
\setlength{\tabcolsep}{3.5pt}
\begin{tabular}{@{}llrrrrr@{}}
\toprule
Design & Specification & \shortstack{Median\\$\widehat\beta$} & \shortstack{Median\\$F$} & \shortstack{Median\\$\widehat{\mathcal I}_P$} & \shortstack{Wald\\coverage} & \shortstack{AR\\rejection} \\
\midrule
Strong surviving & Raw & 0.231 & 953.02 & 0.6353 & 0.000 & 1.000 \\
 & Two-way FE & -0.497 & 177.85 & 0.1186 & 0.779 & 0.227 \\
 & Common projection & -0.503 & 725.44 & 0.4836 & 0.940 & 0.063 \\
Weak surviving & Raw & 0.437 & 501.93 & 0.3346 & 0.000 & 1.000 \\
 & Two-way FE & -0.457 & 5.81 & 0.0039 & 0.917 & 0.226 \\
 & Common projection & -0.494 & 21.57 & 0.0144 & 0.939 & 0.066 \\
Near boundary & Raw & 0.477 & 457.460 & 0.30500 & 0.000 & 1.000 \\
 & Two-way FE & -0.434 & 2.337 & 0.00156 & 0.943 & 0.224 \\
 & Common projection & -0.342 & 2.338 & 0.00156 & 0.943 & 0.058 \\
\bottomrule
\end{tabular}
\begin{minipage}{0.95\textwidth}\footnotesize
Notes: True coefficient $\beta_0=-0.5$; 5,000 replications. ``Common projection'' applies
the same rank-one estimated loading and factor spaces to all three variables. AR rejection is evaluated at the true value using the Gaussian fixed-rank reference value.
The package includes an exploratory two-way multiplier routine, but formal bootstrap
claims remain tied to the exact primitive or PWB-H conditions in Section~\ref{sec:bootstrap}.
The simulation is intended to
validate the diagnostic mechanism, not to replace the broader designs in the Online
Appendix.
At the near-boundary design, the two transformed specifications lie at essentially the
same information floor. Their unrounded median $F$ values are $2.337171$ and $2.337882$,
and their unrounded median $\widehat{\mathcal I}_P$ values are $0.00155811$ and
$0.00155859$; the estimates and AR rejection frequencies remain materially different.
\end{minipage}
\end{table}

The projected first-stage $F$, projected Jacobian, and PIM need not move together. In the
flagship application, the PIM can rise while the Jacobian falls because score variability
declines faster than moment sensitivity. The projected $F$ should therefore not be read
as a conventional sufficient measure of identification strength: relevance, sensitivity,
and covariance-adjusted information are distinct objects.

\subsection{Flagship application: external U.S. monetary shocks and international transmission}

The economic question is whether U.S. monetary tightening produces a persistent foreign
output contraction. A conventional panel can answer using variation that largely
reflects global comovement. The projected analysis asks the more demanding question:
does the conclusion survive after removing the estimated common component while
retaining differential country exposure to the external policy shock?

\label{sec:application}

The flagship application combines Release~6 of the Jord\`a--Schularick--Taylor
Macrohistory Database \citep{jordaschularicktaylor2017} with the public high-frequency
U.S. monetary-policy surprises maintained by the Federal Reserve Bank of San Francisco
and developed in \citet{bauerswanson2023}. We aggregate the event-level orthogonalized
surprise \texttt{MPS\_ORTH} within each calendar year.

A common U.S. shock is absorbed by unrestricted year effects. Identification therefore
comes from predetermined cross-country exposure. The baseline exposure is the lagged JST
exchange-rate-peg indicator. For country $i$ and year $t$, the endogenous policy
interaction is
\[
 x_{it}=\Delta i_t^{US}e_{i,t-1},
\]
and the excluded instrument is
\[
 z_{it}=MPShock_t e_{i,t-1}.
\]
The outcome at horizon $h$ is annualized cumulative real GDP-per-capita growth from $t$
to $t+h$. Controls include lagged domestic GDP growth, inflation, credit growth, the
lagged exposure, country effects, and year effects. The external shock strengthens the
timing argument, but the interaction design still requires the lagged exchange-rate
regime to be predetermined and excluded from future domestic outcomes except through
differential exposure to U.S. policy.

For each horizon $h=1,\ldots,5$, we compare country/year fixed-effects IV with the same
model after the outcome, endogenous policy interaction, and excluded shock interaction
are first residualized on the same control matrix and then passed through one common
rank-one two-sided projection. The controls are partialled out once rather than projected
as additional outcome matrices. The full-sample projection is used as a transparent
empirical diagnostic; formal feasible
inference additionally requires the oracle-transfer and rank-separation conditions stated
in the theory. Figure~\ref{fig:externalinformation} reports the first-stage statistic,
absolute projected Jacobian, Projected Information Matrix, and IV coefficient.

\begin{figure}[t]
\centering
\includegraphics[width=0.94\textwidth]{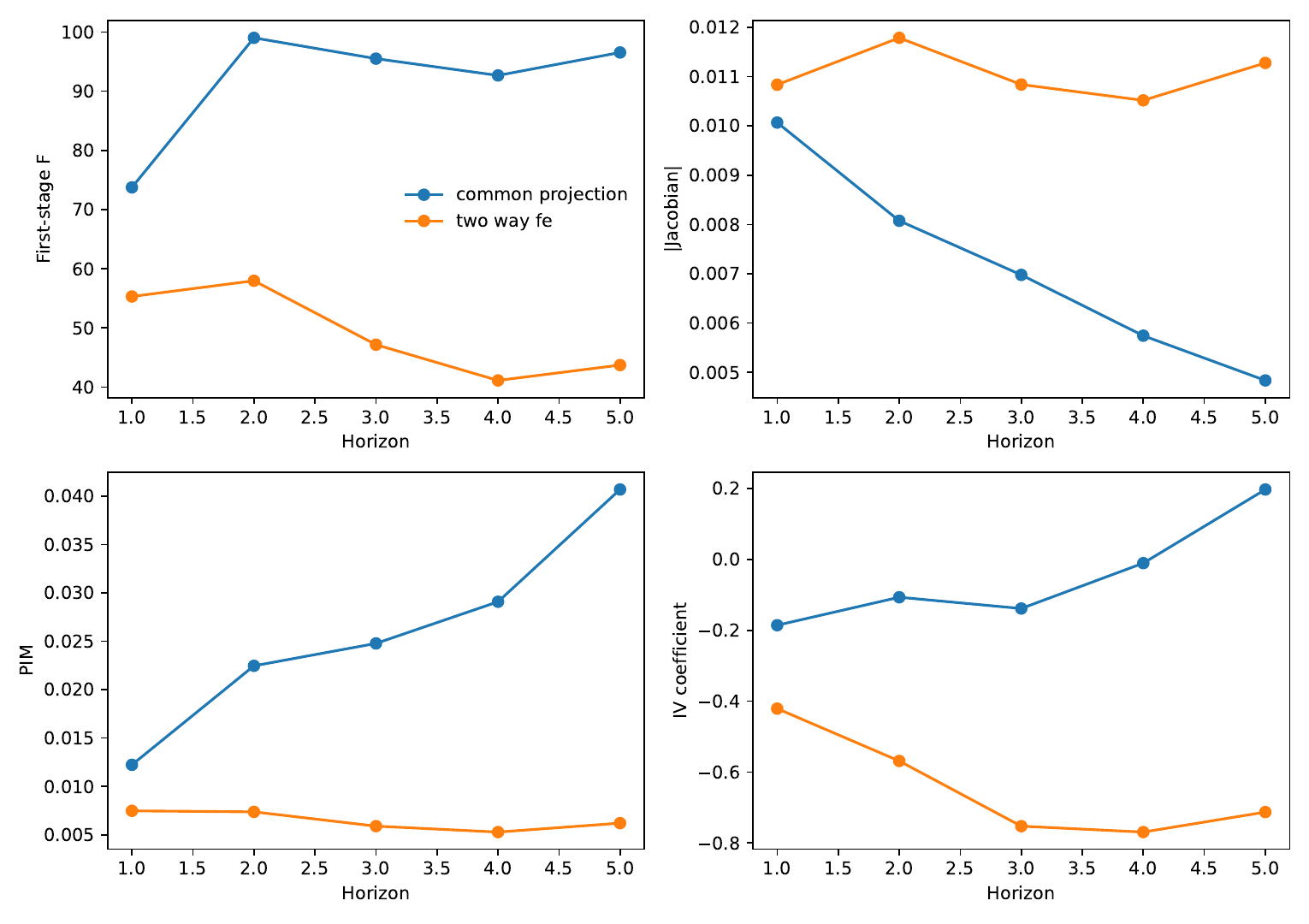}
\caption{External U.S. monetary shocks: information and coefficients by horizon}
\label{fig:externalinformation}
\begin{minipage}{0.94\textwidth}\footnotesize
Notes: The instrument is the annual sum of orthogonalized FOMC surprises interacted with
the lagged exchange-rate-peg indicator. The endogenous regressor is the annual change in
the U.S. short-term rate interacted with the same exposure. The outcome is annualized
cumulative real GDP-per-capita growth.
\end{minipage}
\end{figure}

The comparison materially changes the economic conclusion. Under two-way fixed effects,
the estimated coefficient is negative at every horizon, ranging from -0.77 to
-0.42. After common projection, the coefficient moves toward zero and becomes
positive at horizon five. At horizon three, the estimate changes from -0.75 to
-0.14. Thus, the apparent persistence is not robust to removing the estimated common
low-rank component. Under the maintained interaction-IV design, the country-specific
shock-exposure component retained by the final specification supports a much smaller and
less stable response.

The information diagnostics move in different directions. The projected Jacobian declines
from 0.0101 at horizon one to 0.0048 at horizon five. Yet the conventional partial first-stage $F$ and the PIM
rise after projection because score variability falls more sharply. At horizon three,
the first-stage statistic rises from 47.2 to 95.5, while the PIM rises from
0.0059 to 0.0248 despite a substantial decline in the Jacobian. This divergence is
the empirical point: relevance, moment sensitivity, and covariance-adjusted information
are distinct objects.

Identification-robust inference remains cautious. Over the reported $[-25,25]$ grid, the
two-way Gaussian-reference Anderson--Rubin confidence set touches both boundaries at every
horizon and under both specifications. Because these sets use the Gaussian fixed-rank
reference law, they are reported as transparent identification diagnostics rather than as
a universal two-way bootstrap result. Neither the conventional first-stage statistic nor
the projected PIM therefore supports a precise causal spillover coefficient under the
maintained interaction-IV design. The appropriate conclusion is qualitative: the persistent
negative response visible before common projection is not robust to removal of the global
low-rank component.

The baseline projected instrument retains about 55.3 percent of its residual variance
on average across horizons. Robustness exercises use raw rather than orthogonalized surprises,
strict-peg and trade-openness exposures, GDP and house-price outcomes, factor ranks zero
through two, and horizons one, three, and five. On the archived Linux environment, the complete flagship pipeline runs in 24.8 seconds
and the secondary credit pipeline in 5.3 seconds. These concurrent master-run timings include data construction, projection, diagnostics,
confidence-set inversion, and figure generation; they are recorded in the replication
report and will vary with hardware.

\section{Conclusion}\label{sec:conclusion}

Nuisance removal can alter the sign, persistence, and precision of structural conclusions
because it changes the variation left to distinguish economic mechanisms. Identification
must therefore be evaluated in the final projected moment problem. The projected
Jacobian records which directions remain visible, the projected-score law records their
precision, and, on Gaussian fixed-rank strata, the PIM combines the two. The theory
allows heterogeneous information rates, non-Gaussian two-way interaction limits,
feasible factor transfer, identification-robust tests, bootstrap procedures, and
uncertainty measures for the projected spectrum, rank, subspaces, and information matrix.

The simulations and applications show why the distinction matters. A raw first-stage
statistic above 500 can coexist with an invalid moment and an IV estimate of the wrong
sign; equation-compatible projection removes the confounding component and moves the
coefficient toward its structural value while revealing the surviving information. In the external U.S. monetary-shock application, conventional
estimates imply persistent foreign output contractions, whereas common projection moves
the response toward zero and reverses its sign at five years. The Jacobian falls while
the first-stage statistic and PIM rise, and Gaussian-reference Anderson--Rubin sets remain
unbounded. No
single diagnostic therefore summarizes the evidence.

The practical implication is concrete. An empirical paper that removes substantial
nuisance variation should report how much excluded variation survives, the projected
Jacobian spectrum, score uncertainty, the PIM, effective rank, and the geometry of an
identification-robust confidence set. These diagnostics should accompany rather than
replace the economic argument for instrument validity. Nuisance robustness and
informative identification are complementary requirements: a credible specification can
still leave the data unable to distinguish economically important mechanisms precisely.
The same diagnostic applies to classic IV designs, including compulsory-schooling
studies of returns to education \citep{angristkrueger1991} and historical-instrument
studies of institutions and growth \citep{acemoglujohnsonrobinson2001}: identifying
variation should be assessed after the exact final controls. The paper does not
re-estimate those studies.

The principle can extend to nonlinear GMM or simulated method of moments when nuisance
removal is represented by a common differentiable transformation of the model and data,
the simulated moment map admits a uniform local expansion, and simulation error is
asymptotically negligible at the score scale. The relevant objects would again be the
post-transformation Jacobian and score law. Non-smooth solution mappings, endogenous
simulation draws, or changing active constraints require separate theory and are not
covered here.

The broader lesson is the Projected Information Principle: projection changes the
experiment. Identification, information, and inference must therefore be reconstructed
from the moments and dependence structure that remain.

\bibliographystyle{apalike}

\newpage
%% Stage 50.1 rigorous audit
%\documentclass[12pt]{article}
%\usepackage{amsmath,amsthm,amssymb,amsfonts,bm,mathtools}
%\usepackage{geometry,setspace,natbib,graphicx,booktabs,enumitem,microtype}
%\usepackage{xcolor,hyperref}\usepackage{appendix}\usepackage{longtable}\usepackage{array}
%\usepackage{tikz}
%\usepackage{xr-hyper}
%\externaldocument{identification_information_nuisance_projection_main}
%\usetikzlibrary{arrows.meta,positioning,shapes.geometric}
%\geometry{margin=1in}
%\onehalfspacing
%\hypersetup{colorlinks=true,citecolor=blue,linkcolor=blue,urlcolor=blue,hypertexnames=false,pdfauthor={Ulrich Hounyo},pdftitle={Online Appendix to Identification and Information after Nuisance Projection}}
%
%\newtheorem{assumption}{Assumption}
%\newtheorem{theorem}{Theorem}
%\newtheorem{proposition}{Proposition}
%\newtheorem{corollary}{Corollary}
%\newtheorem{lemma}{Lemma}
%\newtheorem{remark}{Remark}
%\newtheorem{example}{Example}
%\newtheorem{definition}{Definition}
%\newcommand{\E}{\mathbb E}
%\newcommand{\V}{\mathbb V}
%\newcommand{\Cov}{\operatorname{Cov}}
%\newcommand{\R}{\mathbb R}
%\newcommand{\op}{o_p}
%\newcommand{\Op}{O_p}
%\newcommand{\toP}{\xrightarrow{p}}
%\newcommand{\toD}{\xrightarrow{d}}
%\newcommand{\cF}{\mathcal F}
%\newcommand{\cI}{\mathcal I}
%\newcommand{\cZ}{\mathcal Z}
%\newcommand{\vech}{\operatorname{vech}}
%\newcommand{\rank}{\operatorname{rank}}
%\newcommand{\diag}{\operatorname{diag}}
%\newcommand{\plim}{\operatorname{plim}}
%\newcommand{\Var}{\operatorname{Var}}

\paragraph{Online Appendix to ``Identification and Information after Nuisance Projection''}
%\author{
%Ulrich Hounyo\thanks{%
%Department of Economics, University at Albany -- State University of New
%York, Albany, NY 12222, USA. E-mail: \texttt{khounyo@albany.edu}. }\\
%Department of Economics\\
%University at Albany, SUNY
%}
%\date{August 2, 2026}
%
%\begin{document}
\maketitle
\noindent
This Online Appendix contains implementation and simulation details, all formal proofs,
primitive verification roadmaps, factor-transfer calculations, and bootstrap-transfer
arguments. Related literature and contribution positioning are consolidated in the main
paper to avoid duplication. Notation and theorem numbering follow the main paper.

\begin{figure}[ht]
\centering
\begin{tikzpicture}[node distance=0.9cm and 1.2cm,>=Latex,
  box/.style={draw,rounded corners,align=center,text width=4.6cm,
  minimum height=1.0cm,fill=gray!8,font=\small},
  arr/.style={->,thick}]
\node[box] (raw) {Raw outcome, endogenous regressors, and instruments};
\node[box,right=of raw] (proj) {One common equation-compatible nuisance projection};
\node[box,below=of raw] (objects) {Projected Jacobian and projected-score law};
\node[box,right=of objects] (infer) {PIM, effective rank, and identification-robust inference};
\draw[arr] (raw) -- (proj);
\draw[arr] (proj) -- (infer);
\draw[arr] (raw) -- (objects);
\draw[arr] (objects) -- (infer);
\end{tikzpicture}
\caption{The projected-information workflow}
\label{fig:conceptual-workflow}
\begin{minipage}{0.92\textwidth}\footnotesize
Notes: The same final transformation must be applied to every variable entering the
structural moment. Identification is determined by the projected Jacobian; precision by
the projected-score law. The PIM combines them only on Gaussian fixed-rank strata.
\end{minipage}
\end{figure}
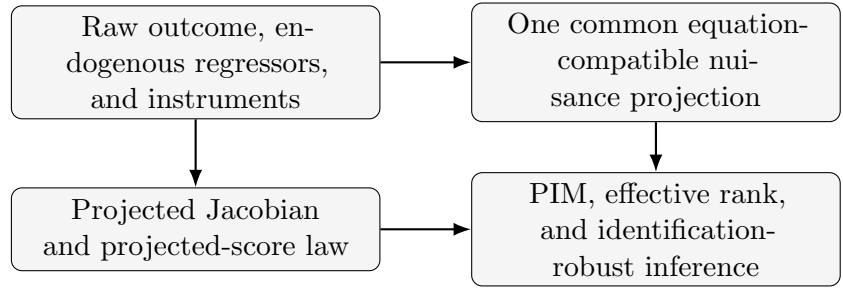

\begin{figure}[ht]
\centering
\begin{tikzpicture}[
 node distance=6mm,
 box/.style={draw,rounded corners,align=center,text width=5.1cm,minimum height=8mm,inner sep=4pt},
 arrow/.style={-{Latex[length=2mm]},thick}
]
\node[box] (raw) {Raw economic variation and potential confounding};
\node[box,below=of raw] (proj) {Remove fixed effects, controls, and latent common components};
\node[box,below=of proj] (cred) {Less confounding, but possibly less identifying variation};
\node[box,below=of cred] (remain) {Which economic directions remain informative, and with what precision?};
\node[box,below=of remain] (pim) {Projected moment experiment and Projected Information Matrix};
\draw[arrow] (raw) -- (proj);
\draw[arrow] (proj) -- (cred);
\draw[arrow] (cred) -- (remain);
\draw[arrow] (remain) -- (pim);
\end{tikzpicture}
\caption{Nuisance removal can improve specification while reducing identifying information.}
\label{fig:economic-motivation}
\end{figure}
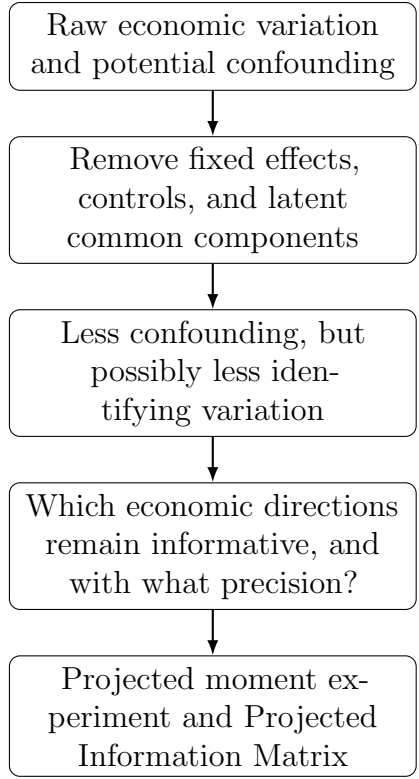

\section{Implementation, reproducibility, and special cases}

Projection compatibility, the identification map, and null-score centering are algebraic;
weak-IV limits require the stated score experiment; factor adjustment requires normalized
oracle-equivalence remainders; and bootstrap validity is restricted to the primitive
benchmark or the exact PWB-H conditions. An application should therefore report the
factor estimator and factor-number rule, score normalizations, covariance-rank stratum,
spectral separation, and the dependence regime used for critical values.

A reproducible implementation should: residualize every variable on the same observed
controls; estimate joint nuisance spaces and apply one common final projection; construct
the projected Jacobian, structural null score, first-stage score, and covariance inputs;
report singular values, confidence intervals, spectral gaps, subspace-stability ratios,
effective rank, and projected-information uncertainty; invert a null-restricted robust
test on an economically meaningful grid; and retain unbounded or disconnected confidence
sets. The replication package should save intermediate projected variables, factor
choices, normalizations, covariance estimates, bootstrap draws, seeds, and all routines
needed to reproduce tables and figures.

\subsection{Programmatic algorithm for the reported diagnostics}\label{app:algorithm}

The implementation uses residualization, singular-value decompositions,
Moore--Penrose inverses, covariance assembly, and score inversion. It does not profile
over a high-dimensional nuisance parameter and uses neither BFGS, Newton--Raphson, MCMC,
simulated annealing, penalty paths, nor barrier methods. Fixed effects and factor spaces
are removed by linear algebra before structural inference. For a dense $N\times T$
matrix, a full SVD costs $O\{\min(NT^2,N^2T)\}$; truncated methods can be used when the
maintained factor rank is small.

Where a numerical threshold is required, the archived code uses rank tolerance
$10^{-10}$. Structural score inversion uses a deterministic reported grid and therefore
has no optimizer starting value or convergence criterion. A dense tensor grid is suitable
only for one- or two-dimensional structural parameters. For larger $p$, the same
null-score statistic should be evaluated by adaptive set-search, directional slices, or
subvector projection of the joint confidence region; the paper does not claim that a
dense $p$-dimensional grid is computationally attractive. The Monte Carlo generator uses
NumPy's \texttt{default\_rng} with seed \texttt{20260802}. The master driver records
software versions, checksums, and elapsed times, and missing benchmark files cause a
failure unless initialization is requested explicitly.

\begin{table}[t]
\centering
\caption{What the reported procedures answer}
\label{tab:procedure-comparison}
\small
\begin{tabular}{@{}p{0.22\textwidth}p{0.30\textwidth}p{0.38\textwidth}@{}}
\toprule
Procedure or diagnostic & Question answered & Limitation \\
\midrule
Conventional Wald/2SLS & What would regular linearization conclude? & Can be nonpivotal when the projected first stage is random at first order. \\
Projected Jacobian spectrum & Which structural directions remain visible after nuisance removal? & Sensitivity alone does not measure score precision. \\
Projected Information Matrix & How much covariance-adjusted information remains on a Gaussian fixed-rank stratum? & Not full-data Fisher information and not smooth across rank changes. \\
Gaussian-reference AR & Does a null-restricted score reject under a Gaussian fixed-rank approximation? & Descriptive if Gaussianity or rank stability is doubtful. \\
Primitive/PWB-H bootstrap AR & What is the null-score reference law under the maintained multiway dependence regime? & Valid only under the exact primitive or transfer assumptions. \\
\bottomrule
\end{tabular}
\end{table}

The procedure below implements the objects reported in the simulations and applications.
It uses direct linear-algebra operations and null-score inversion for the maintained
projected-moment model.

\begin{enumerate}[leftmargin=2.5em,label=\textbf{Step \arabic*.}]
\item \textbf{Common observed-control residualization.}
Form one control matrix $W$ containing included controls and fixed effects. For each
analysis matrix $A$, compute $A^{\perp}=A-W(W'W)^{\dagger}W'A$ using the same observation
set and the same $W$.

\item \textbf{Estimate one common nuisance space.}
Stack the residual matrices used by the equation, estimate the selected left and right
low-rank spaces, and construct
$\widehat M_L=I_N-\widehat P_L$ and $\widehat M_F=I_T-\widehat P_F$.

\item \textbf{Apply the equation-compatible projection.}
For every residual matrix, set
$\widehat A^{o}=\widehat M_L A^{\perp}\widehat M_F$. The same projector pair is used for
the outcome, endogenous regressors, and instruments, as required by
Lemma~\ref{lem:projectioncompat}.

\item \textbf{Construct the Jacobian and spectral diagnostics.}
Compute
\[
 \widehat J_{NT}^{s}=(NT)^{-1}\sum_{i,t}\widehat z_{it}^{o}
 \widehat x_{it}^{o\prime}.
\]
Report singular values, effective rank, spectral gaps, and uncertainty-to-gap ratios.

\item \textbf{Estimate score uncertainty and projected information.}
At candidate $\beta$, form
$\widehat g_{it}(\beta)=\widehat z_{it}^{o}
(\widehat y_{it}^{o}-\widehat x_{it}^{o\prime}\beta)$ and estimate the covariance of
the correctly normalized aggregate score. On a stable Gaussian fixed-rank stratum,
calculate
$\widehat{\mathcal I}_{P,NT}=\widehat H_{NT}'
\widehat\Omega_{NT}^{\dagger}\widehat H_{NT}$.

\item \textbf{Invert the null-restricted score statistic.}
On a prespecified grid $\mathcal G$, compute
\[
 AR_{NT}(\beta)=\widehat G_{NT}(\beta)'
 \widehat\Sigma_{g,NT}(\beta)^{\dagger}\widehat G_{NT}(\beta),
 \qquad
 \widehat G_{NT}(\beta)=\sum_{i,t}\widehat g_{it}(\beta).
\]
Use Gaussian critical values only under Proposition~\ref{prop:ar-limit}; otherwise use
the primitive bootstrap or PWB-H under its exact assumptions. Preserve disconnected
components and report a set as unbounded on the grid when it touches both endpoints.

\item \textbf{Save verification objects.}
Save projected matrices, factor choices, singular values, covariance ranks, PIM inputs,
accepted grid points, seeds, software versions, checksums, and elapsed times. The master
driver compares reproduced CSV files with archived benchmarks.
\end{enumerate}

\subsection{Projection stability and rank boundaries}
\label{app:projection-stability}

The oracle residual operator
$\mathcal M_0(A)=M_{\mathcal L_0}AM_{\mathcal F_0}$ depends only on the maintained
nuisance spaces. Loss of structural rank in $J_{NT}$ does not make this oracle operator
singular. Feasible projection is a different issue: consistency of
$\widehat P_{\mathcal L}$ and $\widehat P_{\mathcal F}$ requires fixed nuisance ranks and
separating eigenvalue gaps. Davis--Kahan perturbation bounds then control projector error
by spectral-estimation error divided by the relevant gap. If a nuisance eigengap closes,
the projector may be discontinuous and no uniform feasible-transfer claim is made across
that boundary.

Tikhonov residualization is not a substitute for this condition. A ridge matrix is
generally neither idempotent nor an exact annihilator of the nuisance space and therefore
changes the transformed equation. It can be reported as a numerical sensitivity check
only if its induced bias is separately controlled; the paper's inferential results use
orthogonal projectors and fixed-rank Moore--Penrose inverses.

\subsection{Uniformity, boundaries, and constrained parameter spaces}

The maintained model is a linear projected-moment model. The paper makes three distinct claims. Structural-equation compatibility and null-score
centering are exact algebraic statements. Spectral and PIM delta methods are pointwise on
separated fixed-rank strata. Theorem~\ref{thm:uniform-ar} gives uniform size only over
classes satisfying the stated uniform score and critical-value approximation. Exactly
zero first stages create no $0/0$ in the null-score statistic because no directional
first-stage normalization is used. Arbitrary nonlinear parameter manifolds are not
covered; in the maintained linear model, intersection with a closed restriction set
preserves joint coverage, while subvector projection can be conservative and disconnected.

Useful descriptive checks include the share of instrument variation removed by projection
and the information curve
\[
 \widehat I(T_0)=N\,\widehat\Pi(T_0)'\widehat Q_{zz}(T_0)
 \widehat\Omega_v(T_0)^{\dagger}\widehat Q_{zz}(T_0)\widehat\Pi(T_0),
 \qquad T_0\le T.
\]
A flat curve indicates that additional dates add little first-stage information, but it
is not a pretest for switching to conventional inference.

The framework nests familiar cases. With no latent factors, the oracle-transfer conditions
disappear. Under stable full-rank identification and a regular score CLT, projected GMM
has conventional asymptotic normality. In scalar IV, the spectrum reduces to the projected
first-stage moment. If one score projection and the interaction component vanish, the
problem reduces to one-way dependence. When
$\Omega_{v,NT}\asymp T^{-1}$, the usual $\sqrt{NT}$ localization is recovered. The full
machinery is most useful when excluded variation, endogeneity, and score dependence load
on different panel dimensions.

\section{Monte Carlo and empirical validation design}

\subsection{Theorem-aligned Monte Carlo design}

The baseline validation creates a heterogeneous common factor that is simultaneously
relevant for the instrument, endogenous regressor, and structural disturbance. Raw IV is
therefore strongly relevant but invalid. A small idiosyncratic coefficient controls the
information that remains after the common projection. This design separates three
objects that conventional simulations often conflate: raw correlation, valid
post-projection relevance, and sampling precision.

For each replication, generate
\begin{align*}
 z_{it}&=\lambda_i^z f_t+e_{it}^z,\\
 x_{it}&=\pi e_{it}^z+\lambda_i^x f_t+v_{it},\\
 y_{it}&=\beta_0x_{it}+\lambda_i^y f_t+u_{it},
\end{align*}
where $(u_{it},v_{it})$ are correlated and all idiosyncratic innovations are independent
of $e_{it}^z$. Heterogeneous loadings prevent two-way additive effects from removing the
common component. The three values $\pi\in\{0.35,0.06,0.02\}$ produce strong, weak, and
near-boundary surviving information. The same rank-one joint projection is estimated
from the stacked outcome, regressor, and instrument matrices and applied to each matrix.

The reported core design uses 5,000 replications and reports: score-law approximation;
raw, two-way-FE, oracle, and feasible results; Wald coverage; AR size and power;
Gaussian-reference AR rejection; singular-value confidence-interval coverage; rank
recovery; subspace error; and projected-information coverage. Additional designs vary
$(N,T)$, factor strength, factor-number misspecification, unit and time score projections,
and finite-rank interaction components. A separate primitive design retains the
information-rate construction already stated in the theory, with
$\pi_{NT}=c/\sqrt{NT^\delta}$ and $\Omega_{v,NT}\asymp T^{-\delta}+T^{-1}$.

\subsection{External monetary-shock flagship application}

The flagship application combines JST Release~6 with the Federal Reserve Bank of San
Francisco workbook \texttt{monetary-policy-surprises-data.xlsx}. The analysis uses
\texttt{FOMC (update 2023)} and constructs annual shocks by summing event-level
\texttt{MPS} and \texttt{MPS\_ORTH}. The baseline uses \texttt{MPS\_ORTH}.

For horizon $h$,
\begin{align*}
 y_{it}(h)
 &=\frac{100}{h}
 \{\log RGDPpc_{i,t+h}-\log RGDPpc_{it}\},\\
 x_{it}
 &=\Delta i_t^{US}e_{i,t-1},\\
 z_{it}
 &=MPShock_t e_{i,t-1}.
\end{align*}
The baseline exposure is the lagged peg indicator. Controls are lagged domestic GDP
growth, inflation, real credit growth, and exposure, together with country and year
effects. House-price growth, strict-peg exposure, and lagged trade openness are
alternatives.

The code residualizes the outcome, endogenous interaction, and instrument on one common
control matrix; constructs the largest balanced common-support rectangle; estimates one
common left and one common right nuisance space from the three residualized matrices; and
applies the same projectors to those three matrices. The controls are not projected a
second time after Frisch--Waugh residualization. For horizons one through five it reports the coefficient, first-stage $F$,
projected Jacobian, PIM, effective rank, surviving instrument variance, and AR geometry.

Under two-way fixed effects the GDP coefficient is negative at every horizon. After
rank-one common projection it moves toward zero and is positive at horizon five. The
projected Jacobian declines with the horizon, while the PIM rises because projected score
variation falls. The Gaussian fixed-rank reference AR confidence sets touch both ends of the reported
$[-25,25]$ grid in every baseline specification and are recorded as unbounded on the
grid. They are descriptive identification diagnostics unless the maintained Gaussian
score and covariance-rank conditions hold.

The high-frequency shock improves the timing of the common monetary innovation, but the
interaction design also relies on the predeterminedness and exclusion of lagged
exchange-rate exposure. The raw and projected specifications are therefore compared
conditional on the same maintained exposure restriction; projection is not represented
as creating exogeneity. The application therefore demonstrates how common projection
changes the apparent persistence and information content of international spillovers; it
does not claim a precise causal elasticity.

\subsection{Secondary domestic-credit validation}

The secondary JST exercise uses current real bank-credit growth instrumented by its
second and third lags to explain future real house-price growth. Its detailed construction,
robustness analysis, and machine-readable results remain in the replication archive.

\begin{appendices}

\section{Proofs and technical lemmas}\label{app:proofs}

\begin{lemma}[Uniform feasible-to-oracle transfer]\label{lem:uniform-transfer}
Let $\mathfrak P_{NT}$ be a maintained law class. Assume, uniformly over
$P\in\mathfrak P_{NT}$:
\begin{enumerate}[label=\textup{(U\arabic*)},leftmargin=2.7em]
\item\label{cond:uniform-score}
For every $\varepsilon>0$,
\[
 \sup_{P\in\mathfrak P_{NT}}
 P\!\left\{
 c_{NT,P}^{-1}\|R_{g,NT,P}(\beta_0)\|>\varepsilon
 \right\}\to0,
 \tag{U1}\label{eq:uniform-factor-transfer}
\]
where $R_{g,NT,P}(\beta_0)$ is the feasible-minus-oracle null-score remainder and
$c_{NT,P}$ is the null-score normalization.
\item\label{cond:uniform-covariance}
For every $\varepsilon>0$,
\[
 \sup_{P\in\mathfrak P_{NT}}
 P\!\left\{
 \|\widehat\Omega_{g,NT,P}-\Omega_{g,NT,P}\|>\varepsilon
 \right\}\to0
\]
on a fixed covariance-rank stratum.
\item\label{cond:uniform-projector}
For every $\varepsilon>0$,
\[
 \sup_{P\in\mathfrak P_{NT}}
 P\!\left\{
 \max_{\bullet\in\{\mathcal L,\mathcal F\}}
 \|\widehat P_{\bullet,NT,P}-P_{\bullet,NT,P}\|>\varepsilon
 \right\}\to0,
\]
with fixed nuisance ranks and separating eigengaps bounded below uniformly.
\item\label{cond:uniform-equivalence}
For every $\varepsilon>0$,
\[
 \sup_{P\in\mathfrak P_{NT}}
 P\{|T_{NT}-T_{NT}^{o}|>\varepsilon\}\to0.
\]
\item\label{cond:uniform-anticoncentration}
The oracle reference cdfs satisfy the uniform anti-concentration condition
\eqref{eq:uniform-anticoncentration}.
\end{enumerate}
Then
\[
 \sup_{P\in\mathfrak P_{NT}}\sup_x
 |P\{T_{NT}\le x\}-P\{T_{NT}^{o}\le x\}|\to0.
\]
Conditions \ref{cond:uniform-score}--\ref{cond:uniform-projector} are primitive sufficient
conditions for \ref{cond:uniform-equivalence} whenever the statistic is locally Lipschitz
on the maintained fixed-rank stratum.
\end{lemma}

\subsection{Proof of Lemma~\ref{lem:uniform-transfer}}\label{app:proof:uniform-transfer}
\begin{proof}
For any $\delta>0$,
\[
 P\{T_{NT}\le x\}
 \le P\{T_{NT}^{o}\le x+\delta\}
 +P\{|T_{NT}-T_{NT}^{o}|>\delta\},
\]
and the reverse inequality follows with $x-\delta$. Taking the supremum over $x$ and
$P$, condition~\ref{cond:uniform-equivalence} makes the stochastic-equivalence term
vanish for fixed $\delta$. Sending $\delta\downarrow0$ through
condition~\ref{cond:uniform-anticoncentration} proves the result.
\end{proof}

\begin{proposition}[Strong-identification root and Hausdorff recovery]
\label{prop:strong-set-recovery}
Let
\[
 \widehat{\mathcal Q}_{NT}(\beta)
 =\widehat q_{NT}-\widehat J_{NT}\beta
\]
be an affine feasible sample moment map on a compact parameter set containing $\beta_0$
in its interior. Suppose
\[
 \sigma_p(J_{NT})\ge c>0,\qquad
 \|\widehat J_{NT}-J_{NT}\|=o_p(1),\qquad
 \|\widehat{\mathcal Q}_{NT}(\beta_0)\|=O_p(r_{NT}),
\]
for deterministic $r_{NT}\downarrow0$. Then, with probability approaching one,
$\widehat J_{NT}$ has full column rank and the least-squares root
\[
 \widehat\beta_{NT}
 =\arg\min_{\beta}\|\widehat{\mathcal Q}_{NT}(\beta)\|^2
\]
is unique and satisfies
\[
 \|\widehat\beta_{NT}-\beta_0\|=O_p(r_{NT}),\qquad
 d_H(\{\widehat\beta_{NT}\},\{\beta_0\})=O_p(r_{NT}).
\]
In particular, if $r_{NT}=(NT)^{-1/2}$, the singleton identified set is recovered at the
standard root-$NT$ rate.
\end{proposition}

\subsection{Proof of Proposition~\ref{prop:strong-set-recovery}}
\label{app:proof:strong-set-recovery}

\begin{proof}
Weyl's inequality gives
$\sigma_p(\widehat J_{NT})\ge c-o_p(1)$, so $\widehat J_{NT}$ has full column rank with
probability approaching one. The least-squares first-order condition yields
\[
 \widehat\beta_{NT}-\beta_0
 =(\widehat J_{NT}'\widehat J_{NT})^{-1}
 \widehat J_{NT}'\widehat{\mathcal Q}_{NT}(\beta_0).
\]
The operator norm of the left inverse is
$\sigma_p(\widehat J_{NT})^{-1}=O_p(1)$; hence the estimation error is
$O_p(r_{NT})$. Because $r_{NT}\downarrow0$, $\widehat\beta_{NT}-\beta_0=o_p(1)$; since
$\beta_0$ is interior to the compact parameter set, the unconstrained least-squares
minimizer is interior with probability approaching one and therefore coincides with the
constrained minimizer. Uniqueness follows from full column rank of $\widehat J_{NT}$.
The Hausdorff distance between singleton sets equals the Euclidean distance between their
elements.
\end{proof}

\paragraph{Index convention.}
Throughout the proofs, $i$ and $t$ index panel units and periods. Symbols such as $j$ and
$\ell$ are used only for instrument, parameter, singular-value, or block coordinates.
A change from $i$ to $j$ therefore never denotes an unannounced change in the panel
summation index. All double panel sums are written explicitly as
$\sum_{i=1}^{N}\sum_{t=1}^{T}$.

\begin{proposition}[Power collapse along nearly projected-out directions]
\label{prop:near-orthogonal-power}
Consider a Gaussian fixed-rank projected experiment and unit local directions $h_{NT}$
such that
\[
 h_{NT}'\mathcal I_{P,NT}h_{NT}\to0.
\]
For every level-$\alpha$ quadratic score test whose local power is governed by the
corresponding noncentral chi-square law, rejection probability under
$\beta_{NT}=\beta_0+D_{NT}^{-1}h_{NT}$ converges to $\alpha$.
\end{proposition}

\subsection{Proof of Proposition~\ref{prop:near-orthogonal-power}}
\label{app:proof:near-orthogonal-power}

\begin{proof}
Under the maintained Gaussian fixed-rank experiment, the quadratic score statistic in
direction $h_{NT}$ converges under the local alternative to a noncentral chi-square law
with noncentrality
$\lambda_{NT}=h_{NT}'\mathcal I_{P,NT}h_{NT}$. By assumption,
$\lambda_{NT}\to0$. Continuity of the noncentral chi-square distribution in its
noncentrality parameter implies convergence of rejection probability to the central
chi-square rejection probability, which equals $\alpha$ at a continuous critical value.
\end{proof}

\subsection{Guide to the proofs}\label{app:proof-concordance}

The table identifies every formal result by its printed theorem, proposition, lemma, or
corollary number and title. The second column gives the numbered proof section and page
in this Online Appendix. No source-code labels are displayed.

\begin{longtable}{>{\raggedright\arraybackslash}p{0.64\textwidth}
                  >{\raggedright\arraybackslash}p{0.27\textwidth}}
\toprule
Formal result in the main paper & Proof in this Online Appendix\\
\midrule
\endfirsthead
\toprule
Formal result in the main paper & Proof in this Online Appendix\\
\midrule
\endhead
Lemma~\ref{lem:projectioncompat}: \emph{Projection compatibility} & Section~\ref{app:proof:lem-projectioncompat}, p.~\pageref{app:proof:lem-projectioncompat} \\
Theorem~\ref{thm:primitive-benchmark}: \emph{Primitive joint projected-score limit} & Section~\ref{app:proof:thm-primitive-benchmark}, p.~\pageref{app:proof:thm-primitive-benchmark} \\
Theorem~\ref{thm:identification-map}: \emph{Projected identification map} & Section~\ref{app:proof:thm-identification-map}, p.~\pageref{app:proof:thm-identification-map} \\
Proposition~\ref{prop:stability}: \emph{Singular-value and subspace stability} & Section~\ref{app:proof:prop-stability}, p.~\pageref{app:proof:prop-stability} \\
Theorem~\ref{thm:singular-inference}: \emph{First-order inference for separated singular values} & Section~\ref{app:proof:thm-singular-inference}, p.~\pageref{app:proof:thm-singular-inference} \\
Corollary~\ref{cor:singular-ci}: \emph{Confidence intervals for a simple singular value} & Section~\ref{app:proof:cor-singular-ci}, p.~\pageref{app:proof:cor-singular-ci} \\
Theorem~\ref{thm:rank-consistency}: \emph{Rank consistency on separated strata} & Section~\ref{app:proof:thm-rank-consistency}, p.~\pageref{app:proof:thm-rank-consistency} \\
Theorem~\ref{thm:subspace-reliability}: \emph{Subspace stability and inferential reliability} & Section~\ref{app:proof:thm-subspace-reliability}, p.~\pageref{app:proof:thm-subspace-reliability} \\
Proposition~\ref{prop:balance}: \emph{Weak-first-stage balance} & Section~\ref{app:proof:prop-balance}, p.~\pageref{app:proof:prop-balance} \\
Theorem~\ref{thm:localmoment}: \emph{Local projected-moment shift representation} & Section~\ref{app:proof:thm-localmoment}, p.~\pageref{app:proof:thm-localmoment} \\
Corollary~\ref{cor:ar-local-power}: \emph{Gaussian local power of the quadratic AR statistic} & Section~\ref{app:proof:cor-ar-local-power}, p.~\pageref{app:proof:cor-ar-local-power} \\
Corollary~\ref{cor:gaussian-blue}: \emph{Optimal linear unbiased rule in the Gaussian moment shift family} & Section~\ref{app:proof:cor-gaussian-blue}, p.~\pageref{app:proof:cor-gaussian-blue} \\
Theorem~\ref{thm:pim-delta}: \emph{Delta method for the projected information matrix} & Section~\ref{app:proof:thm-pim-delta}, p.~\pageref{app:proof:thm-pim-delta} \\
Corollary~\ref{cor:pim-directional}: \emph{Directional projected-information inference} & Section~\ref{app:proof:cor-pim-directional}, p.~\pageref{app:proof:cor-pim-directional} \\
Corollary~\ref{cor:pim-eigen}: \emph{Inference for a simple information eigenvalue} & Section~\ref{app:proof:cor-pim-eigen}, p.~\pageref{app:proof:cor-pim-eigen} \\
Corollary~\ref{cor:pim-condition}: \emph{Condition-number inference on positive-definite strata} & Section~\ref{app:proof:cor-pim-condition}, p.~\pageref{app:proof:cor-pim-condition} \\
Theorem~\ref{thm:pim-bootstrap-region}: \emph{Bootstrap confidence regions for projected information} & Section~\ref{app:proof:thm-pim-bootstrap-region}, p.~\pageref{app:proof:thm-pim-bootstrap-region} \\
Theorem~\ref{thm:ratio}: \emph{Weak-IV ratio limit} & Section~\ref{app:proof:thm-ratio}, p.~\pageref{app:proof:thm-ratio} \\
Theorem~\ref{thm:pcbenchmark}: \emph{Sample-split spectral factor transfer} & Section~\ref{app:proof:thm-pcbenchmark}, p.~\pageref{app:proof:thm-pcbenchmark} \\
Proposition~\ref{prop:oracle}: \emph{Oracle-to-feasible transfer under normalized remainder bounds} & Section~\ref{app:proof:prop-oracle}, p.~\pageref{app:proof:prop-oracle} \\
Theorem~\ref{thm:multiratio}: \emph{Multivariate weak-IV limit on a full-rank local stratum} & Section~\ref{app:proof:thm-multiratio}, p.~\pageref{app:proof:thm-multiratio} \\
Corollary~\ref{cor:multiAR}: \emph{Multivariate Anderson--Rubin reduction} & Section~\ref{app:proof:cor-multiAR}, p.~\pageref{app:proof:cor-multiAR} \\
Theorem~\ref{thm:efficient-gmm}: \emph{Regular projected-GMM limit and covariance-optimal weighting} & Section~\ref{app:proof:thm-efficient-gmm}, p.~\pageref{app:proof:thm-efficient-gmm} \\
Proposition~\ref{prop:wald}: \emph{Failure of conventional Wald pivotality} & Section~\ref{app:proof:prop-wald}, p.~\pageref{app:proof:prop-wald} \\
Proposition~\ref{prop:ar}: \emph{Identification robustness of the null score} & Section~\ref{app:proof:prop-ar}, p.~\pageref{app:proof:prop-ar} \\
Theorem~\ref{thm:uniform-ar}: \emph{Uniform size of the projected null-score test} & Section~\ref{app:proof:thm-uniform-ar}, p.~\pageref{app:proof:thm-uniform-ar} \\
Proposition~\ref{prop:ar-limit}: \emph{Gaussian fixed-rank limit of the quadratic AR statistic} & Section~\ref{app:proof:prop-ar-limit}, p.~\pageref{app:proof:prop-ar-limit} \\
Theorem~\ref{thm:primitive-bootstrap}: \emph{Self-contained feasible bootstrap for the primitive score law} & Section~\ref{app:proof:thm-primitive-bootstrap}, p.~\pageref{app:proof:thm-primitive-bootstrap} \\
Corollary~\ref{cor:primitive-bootstrap-quadratic}: \emph{Primitive-bootstrap critical values for a quadratic score statistic} & Section~\ref{app:proof:cor-primitive-bootstrap-quadratic}, p.~\pageref{app:proof:cor-primitive-bootstrap-quadratic} \\
Proposition~\ref{prop:boot}: \emph{Null-score reduction for PWB-H} & Section~\ref{app:proof:prop-boot}, p.~\pageref{app:proof:prop-boot} \\
Theorem~\ref{thm:exactpwb}: \emph{Fixed-rank feasible PWB-H transfer} & Section~\ref{app:proof:thm-exactpwb}, p.~\pageref{app:proof:thm-exactpwb} \\
\bottomrule
\end{longtable}

The proofs below follow the order of the main paper. Each proof identifies the exact
assumptions used; assumptions that are not required are not invoked.

\section{Primitive finite-range benchmark: complete conditions}\label{app:primitive-benchmark}

The next result replaces Assumption~\ref{ass:primitiveclt} for one transparent class of
arrays. It is deliberately narrower than the full framework, but it shows exactly how
unit, time, interaction, and cell components combine and why the limit need not be
Gaussian.

Let $d=q+pq$ and define the centered stacked score
\[
 \xi_{it,NT}
 =
 \begin{pmatrix}
 z_{it}^{o}u_{it}^{o}-\E[z_{it}^{o}u_{it}^{o}]\\
 \operatorname{vec}(z_{it}^{o}v_{it}^{o\prime})
 -\E[\operatorname{vec}(z_{it}^{o}v_{it}^{o\prime})]
 \end{pmatrix}\in\mathbb R^d.
\]

\begin{assumption}[Primitive score benchmark]\label{ass:primitivebenchmark-details}
For fixed integers $m_A,m_D,r_A,r_D$, the centered score admits
\[
 \xi_{it,NT}
 =a_{i,NT}+d_{t,NT}
 +h_{NT}(\phi_{i,NT},\psi_{t,NT})+\varepsilon_{it,NT},
 \label{eq:primitive-score-decomp}
\]
where the $k$th coordinate of the interaction is
\[
 [h_{NT}(\phi,\psi)]_k=\phi'H_{k,NT}\psi,
 \qquad k=1,\ldots,d,
\]
and the following conditions hold.

\begin{enumerate}[label=(\roman*),leftmargin=2.6em]
\item Let
$R_{i,NT}^{A}=(a_{i,NT}',\phi_{i,NT}')'$.
The triangular sequence $\{R_{i,NT}^{A}\}_{i=1}^N$ is centered, strictly stationary
within each row, and $m_A$-dependent. Its $(2+\eta)$ moments are uniformly bounded for
some $\eta>0$. For $0\le h\le m_A$, define the row-stationary lag covariance
$\Gamma_{A,NT}(h)=\Cov(R_{i,NT}^{A},R_{i+h,NT}^{A})$ for any admissible interior index
$i$. Then
\[
 \Gamma_{A,NT}(0)+
 \sum_{h=1}^{m_A}\{\Gamma_{A,NT}(h)+\Gamma_{A,NT}(h)'\}
 \to\Sigma_A.
\]
\item With $R_{t,NT}^{D}=(d_{t,NT}',\psi_{t,NT}')'$, the triangular sequence
$\{R_{t,NT}^{D}\}_{t=1}^T$ satisfies the analogous conditions with dependence range
$m_D$ and limiting long-run covariance $\Sigma_D$.
\item The unit-state sequence, time-state sequence, and cell-innovation array are mutually
independent. The vectors $\{\varepsilon_{it,NT}\}_{i,t}$ are independent across cells,
centered, and satisfy
\[
 \frac1{NT}\sum_{i,t}\E[\varepsilon_{it,NT}\varepsilon_{it,NT}']\to\Sigma_\varepsilon,
\]
together with the Lindeberg condition
\[
 \frac1{NT}\sum_{i,t}
 \E\!\left[\|\varepsilon_{it,NT}\|^2
 \mathbf 1\{\|\varepsilon_{it,NT}\|>\epsilon\sqrt{NT}\}\right]\to0
\]
for every $\epsilon>0$.
A convenient sufficient condition for uniform applications is that, for some
$\eta>0$,
\[
 \sup_{P\in\mathfrak P_{NT}}
 (NT)^{-1-\eta/2}\sum_{i,t}
 \E_P\|\varepsilon_{it,NT}\|^{2+\eta}\to0.
 \label{eq:uniform-lyapunov-cell}
\]
Together with the uniformly bounded $(2+\eta)$ moments of the unit and time blocks, this
is the $L_{2+\eta}$ qualification used when the primitive CLTs are required uniformly.
\item $H_{k,NT}\to H_k$ for every $k$, and the dimensions $r_A,r_D$ do not grow.
\item \emph{(Block-separated normalization.)} The unit, time, and
interaction/cell components accumulate at the distinct orders $T\sqrt N$,
$N\sqrt T$, and $\sqrt{NT}$. Since
$T\sqrt N/\sqrt{NT}=\sqrt T\to\infty$ and
$N\sqrt T/\sqrt{NT}=\sqrt N\to\infty$, one scalar-equivalent normalizer cannot keep all
three components nondegenerate in the same score coordinate. We therefore impose the
coordinate separation required by the stated joint limit.

There is a fixed orthogonal decomposition
\[
 \mathbb R^d=V_A\oplus V_D\oplus V_0
\]
with orthogonal projectors $P_A,P_D,P_0$ such that, for every admissible argument,
\[
\begin{gathered}
 P_Da_{i,NT}=P_0a_{i,NT}=0,\qquad
 P_Ad_{t,NT}=P_0d_{t,NT}=0,\\
 P_Ah_{NT}(\phi,\psi)=P_Dh_{NT}(\phi,\psi)=0,\qquad
 P_A\varepsilon_{it,NT}=P_D\varepsilon_{it,NT}=0.
\end{gathered}
\]
Thus the unit main effect lies in $V_A$, the time main effect in $V_D$, and the
interaction and cell components in $V_0$. Define
\[
 K_{NT}
 =T\sqrt N\,P_A+N\sqrt T\,P_D+\sqrt{NT}\,P_0.
 \label{eq:block-normalizer}
\]
Then
\[
 T\sqrt N\,K_{NT}^{-1}P_A=P_A,\qquad
 N\sqrt T\,K_{NT}^{-1}P_D=P_D,\qquad
 \sqrt{NT}\,K_{NT}^{-1}P_0=P_0.
\]
Fixed nonsingular changes of coordinates within the three blocks are permitted, provided
their eigenvalues are uniformly bounded away from zero and infinity.
\end{enumerate}
\end{assumption}

\begin{remark}[Why block separation is necessary]\label{rem:block-necessity}
If $\sqrt{NT}\,K_{NT}^{-1}\to G_0\ne0$ under a single unrestricted normalizer, then
\[
 T\sqrt N\,K_{NT}^{-1}
 =\sqrt T\{\sqrt{NT}\,K_{NT}^{-1}\}
\]
cannot converge to a finite nonzero matrix on the range of $G_0$. The analogous
contradiction holds for the time component. Hence a nondegenerate main effect and a
nondegenerate interaction/cell component cannot coexist in the same score coordinate.
The direct-sum condition states the geometry needed for their joint coexistence.
\end{remark}

\begin{example}[A non-vacuous block design]\label{ex:blockdgp}
Let $d=2$. In the first score coordinate take a unit main effect only,
$\xi_{it,1}=a_{i,NT}$, where $\{a_{i,NT}\}_i$ is centered, $m_A$-dependent, and
\[
 \Var\!\left(N^{-1/2}\sum_i a_{i,NT}\right)\to\sigma_a^2>0.
\]
In the second coordinate take a pure rank-one interaction
$\xi_{it,2}=\phi_{i,NT}\psi_{t,NT}$, where the two collections are independent,
centered, i.i.d., and have unit variance. With
$V_A=\operatorname{span}(e_1)$, $V_0=\operatorname{span}(e_2)$, $V_D=\{0\}$, and
\[
 K_{NT}=\operatorname{diag}(T\sqrt N,\sqrt{NT}),
\]
we obtain
\[
 K_{NT}^{-1}\sum_{i,t}\xi_{it,NT}
 \Rightarrow
 (\sigma_a Z_1,\;Z_\phi Z_\psi)',
\]
where $Z_1,Z_\phi,Z_\psi$ are mutually independent standard normals. The first
coordinate is Gaussian; the second has excess kurtosis $6$ and is non-Gaussian.
The assumption is nonempty, and covariance information alone does not determine the
joint law of a statistic combining both coordinates.
\end{example}

For $z_\phi\in\mathbb R^{r_A}$ and $z_\psi\in\mathbb R^{r_D}$, write
\[
 \mathfrak h(z_\phi,z_\psi)
 =\bigl(z_\phi'H_1z_\psi,\ldots,z_\phi'H_dz_\psi\bigr)'.
\]

\paragraph{Purpose.}
The next theorem derives the stochastic building block of the paper for a transparent
finite-range benchmark. It shows how unit, time, interaction, and cell components combine
in the normalized projected score and explains why the limiting law can be non-Gaussian.

\paragraph{Restatement of the primitive limit.}
Under \eqref{eq:block-normalizer}, the normalizer limits are the block projectors:
$G_A:=P_A$, $G_D:=P_D$, and $G_0:=P_0$. We use $P_\bullet$ throughout below.

Under Assumption~\ref{ass:primitivebenchmark-details},
\[
 K_{NT}^{-1}\sum_{i=1}^N\sum_{t=1}^T\xi_{it,NT}
 \Rightarrow
 P_A Z_a+P_D Z_d+P_0\{\mathfrak h(Z_\phi,Z_\psi)+Z_\varepsilon\},
 \label{eq:primitive-limit}
\]
where
\[
 \begin{pmatrix}Z_a\\ Z_\phi\end{pmatrix}\sim N(0,\Sigma_A),\qquad
 \begin{pmatrix}Z_d\\ Z_\psi\end{pmatrix}\sim N(0,\Sigma_D),\qquad
 Z_\varepsilon\sim N(0,\Sigma_\varepsilon),
\]
and the three displayed Gaussian blocks are mutually independent. The limit is Gaussian
on $V_A\oplus V_D$ and contains the interaction and cell components only on $V_0$.
If the Gaussian inputs are nondegenerate on their maintained supports, a nonzero
bilinear interaction contributes a nonvanishing second Gaussian-chaos component on
$V_0$; hence the displayed limit is not Gaussian unless that component is degenerate.
No score coordinate simultaneously carries a nondegenerate main effect and a
nondegenerate interaction/cell component.

If $K_{NT}=\operatorname{diag}(A_{NT},B_{NT})$, partitioned conformably with the
structural and first-stage scores, and the corresponding partitioned limit is
nondegenerate, then \eqref{eq:primitive-limit} implies
Assumption~\ref{ass:primitiveclt} for this benchmark class.

\paragraph{Interpretation.}
The theorem shows that uncertainty after nuisance removal can originate from persistent
unit heterogeneity, common time shocks, their interaction, and cell-level innovations.
The main effects accumulate at $T\sqrt N$ and $N\sqrt T$, whereas the interaction and
cell components accumulate at $\sqrt{NT}$. These rates can coexist in one vector limit
only on separated score blocks. When a nonzero interaction survives on $V_0$, covariance
alone does not describe the relevant uncertainty, and Gaussian critical values may be
misleading even in a large panel.

\section{Sample-split strong-factor benchmark: complete conditions}\label{app:sample-split-factor}

The following benchmark derives the projector and score-transfer rates for a concrete
spectral estimator. It uses an auxiliary sample, or an independent block fold, to estimate
the nuisance spaces. This independence is stronger than necessary but makes the role of
orthogonality transparent.

Let $\{A_{\ell,NT}^{\mathrm{tr}}:\ell=1,\ldots,L\}$ be a fixed collection of auxiliary
$N\times T$ matrices with
\[
 A_{\ell,NT}^{\mathrm{tr}}=N_{\ell,NT}+E_{\ell,NT},
 \qquad
 N_{\ell,NT}=L_{\ell,NT}F_{\ell,NT}'.
\]
Define
\[
 S_{L,NT}^{\mathrm{tr}}
 =\frac1T\sum_{\ell=1}^L
 A_{\ell,NT}^{\mathrm{tr}}A_{\ell,NT}^{\mathrm{tr}\prime},
 \qquad
 S_{F,NT}^{\mathrm{tr}}
 =\frac1N\sum_{\ell=1}^L
 A_{\ell,NT}^{\mathrm{tr}\prime}A_{\ell,NT}^{\mathrm{tr}}.
\]
The estimators $\widehat P_\Lambda$ and $\widehat P_F$ are the spectral projectors
associated with the $r_L$ and $r_F$ largest eigenvalues of these matrices.

\begin{assumption}[Sample-split spectral factor benchmark]\label{ass:pcbenchmark}
Let $\eta_{NT}=N^{-1/2}+T^{-1/2}$.
\begin{enumerate}[label=(\roman*),leftmargin=2.6em]
\item The dimensions $L,r_L,r_F$ are fixed, and
\[
 \mathcal R\!\left(\frac1T\sum_\ell
 N_{\ell,NT}N_{\ell,NT}'\right)=\mathcal L_0,
 \qquad
 \mathcal R\!\left(\frac1N\sum_\ell
 N_{\ell,NT}'N_{\ell,NT}\right)=\mathcal F_0.
\]
\item The smallest nonzero eigenvalues of these two signal covariance matrices are at
least $cN$ and $cT$, respectively, for some $c>0$, while the next eigenvalues are zero.
\item Uniformly over $\ell$,
\[
 \|N_{\ell,NT}\|=O(\sqrt{NT}),
 \qquad
 \|E_{\ell,NT}\|=O_p(\sqrt N+\sqrt T).
\]
\item The auxiliary matrices are independent of the score-evaluation fold. Let
$R_{g,NT}^{\mathrm{tr}}(\beta)$ and $R_{v,NT}^{\mathrm{tr}}$ denote the aggregate
structural-score and first-stage-score remainders obtained by substituting the estimated
projectors. Conditional on the training sigma-field $\mathcal T_{NT}$,
\[
 \sup_{\beta\in\mathcal B}
 \left\|\E[R_{g,NT}^{\mathrm{tr}}(\beta)\mid\mathcal T_{NT}]\right\|
 =o_p(\sqrt{NT}\eta_{NT}),
\]
and
\[
 \E\!\left[
 \sup_{\beta\in\mathcal B}
 \left\|
 R_{g,NT}^{\mathrm{tr}}(\beta)
 -\E[R_{g,NT}^{\mathrm{tr}}(\beta)\mid\mathcal T_{NT}]
 \right\|^2
 \,\middle|\,\mathcal T_{NT}\right]
 \le C NT\eta_{NT}^2
\]
with probability approaching one. The analogous conditional mean and second-moment
bounds hold for $R_{v,NT}^{\mathrm{tr}}$ without the supremum over $\beta$. The
corresponding sample-average instrument-second-moment and Jacobian remainders have norm
$O_p(\eta_{NT})$.
\item The normalizations and relevant spectral scale satisfy
\[
 \frac{\sqrt{NT}\eta_{NT}}{a_{NT}}\to0,\qquad
 \frac{\sqrt{NT}\eta_{NT}}{b_{NT}}\to0,\qquad
 \frac{\eta_{NT}}{\varrho_{NT}}\to0.
\]
\end{enumerate}
\end{assumption}

\paragraph{Purpose.}
The oracle theory is useful empirically only if estimated nuisance spaces preserve the
score experiment and the projected Jacobian. The next theorem verifies this transfer for
a sample-split strong-factor spectral estimator under explicit conditional remainder bounds.

\paragraph{Restatement of the transfer result.}
Under Assumptions~\ref{ass:moments} and \ref{ass:pcbenchmark},
\[
 \|\widehat P_\Lambda-P_{\mathcal L_0}\|=O_p(\eta_{NT}),
 \qquad
 \|\widehat P_F-P_{\mathcal F_0}\|=O_p(\eta_{NT}).
\]
Moreover,
\[
 \sup_{\beta\in\mathcal B}
 \left\|\sum_{i,t}\{\widehat g_{it}(\beta)-g_{it}(\beta)\}\right\|
 =O_p(\sqrt{NT}\eta_{NT}),
\]
\[
 \left\|\sum_{i,t}\left[
 \operatorname{vec}(\widehat z_{it}^{o}\widehat v_{it}^{o\prime})
 -\operatorname{vec}(z_{it}^{o}v_{it}^{o\prime})\right]\right\|
 =O_p(\sqrt{NT}\eta_{NT}),
\]
and the feasible instrument second moment and projected Jacobian differ from their
oracle counterparts by $O_p(\eta_{NT})$. Hence the remainder conditions in
Assumption~\ref{ass:factororth} hold for this benchmark. In addition,
$\widehat Q_{zz,NT}$ remains positive definite with probability approaching one by
Assumption~\ref{ass:moments}, and the $O_p(\eta_{NT})$ Jacobian error transfers every
spectral classification whose signal and separating gap dominate $\eta_{NT}$.
Proposition~\ref{prop:oracle} therefore transfers the oracle score experiment and all
such separated classifications.

The same conclusion holds under block cross-fitting if each evaluation block is
asymptotically independent of its training blocks and the displayed conditional moment
bounds hold fold by fold. The theorem derives the spectral projector rates; the
score- and Jacobian-transfer bounds rely on the explicit conditional moment requirements
in Assumption~\ref{ass:pcbenchmark}(iv).

\paragraph{Interpretation.}
Strong factors deliver projector error of order $N^{-1/2}+T^{-1/2}$ in this benchmark.
That rate is adequate for weak-identification inference only when the normalized score
and the relevant singular-value gap dominate the induced error. Thus ``consistent
factors'' is not the operative condition; the relevant comparison is between projector
error, score scale, and identification scale.

\section{Finite-sample stress test with increasing nuisance rank}
\label{app:nuisance-rank-stress}

The maintained asymptotic theory fixes the nuisance ranks. To assess finite-sample
sensitivity outside that class, this experiment increases the common loading and factor
rank with panel size while estimating one common projector from the stacked outcome,
regressor, and instrument matrices. Every replication is retained; a numerical failure
is counted as rejection and reported explicitly.

\begin{table}[ht]
\centering
\caption{Finite-sample sensitivity to nuisance-space dimension}
\label{tab:nuisance-rank-stress}
\begin{tabular}{rrrrrrrr}
\toprule
$N$ & $T$ & $r$ & $r/\min(N,T)$ & Proj. error & AR size & Wald size & Failures\\
\midrule
40 & 20 & 1 & 0.050 & 0.228 & 0.093 & 0.080 & 0 \\
40 & 20 & 3 & 0.150 & 0.411 & 0.120 & 0.147 & 0 \\
40 & 20 & 6 & 0.300 & 0.579 & 0.193 & 0.227 & 0 \\
80 & 40 & 2 & 0.050 & 0.254 & 0.027 & 0.033 & 0 \\
80 & 40 & 5 & 0.125 & 0.369 & 0.100 & 0.120 & 0 \\
80 & 40 & 10 & 0.250 & 0.553 & 0.180 & 0.253 & 0 \\
\bottomrule
\end{tabular}
\begin{minipage}{0.95\textwidth}\footnotesize
Notes: 150 replications per design, seed 20260802. ``Proj. error'' is the median spectral
norm error of the estimated left/right nuisance projectors. The experiment is a stress
test, not an extension of the fixed-rank theorem. Code:
\texttt{Replication/code/run\_nuisance\_rank\_stress.py}.
\end{minipage}
\end{table}

The table documents the finite-sample deterioration, if any, in projector estimation and
test size as nuisance complexity rises. It does not support extrapolation to ranks
proportional to panel dimensions; such sequences require separate high-rank factor theory.

\section{Small-sample power with weak projected first stages}
\label{app:power-grid}

The following diagnostic grid studies first-stage statistics in the weak range.
Because $F$ does not uniquely index the projected experiment, the simulation calibrates
ten designs and reports the \emph{achieved} median projected $F$. The panel has
$(N,T)=(40,20)$, 750 replications per design, and the same common rank-one projection as
the main Monte Carlo. The null is $\beta=-0.5$, the fixed alternative is $\beta=-0.1$,
and the AR statistic is inverted over $[-3,2]$.

\begin{figure}[ht]
\centering
\includegraphics[width=0.72\textwidth]{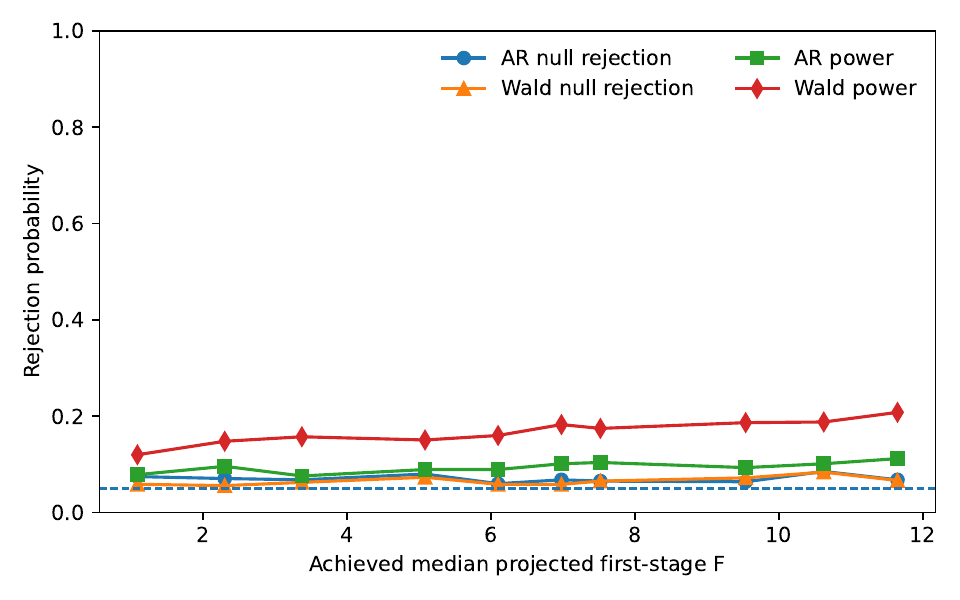}
\caption{Small-sample Wald and AR rejection across projected first-stage strength}
\label{fig:power-grid}
\begin{minipage}{0.90\textwidth}\footnotesize
Notes: The horizontal axis reports achieved median projected $F$, which is not a
sufficient statistic for the projected experiment. The dashed line is the five-percent
nominal level. Code: \texttt{Replication/code/run\_power\_grid.py}.
\end{minipage}
\end{figure}

The achieved median projected $F$ ranges from 1.14 to 13.02.
AR null rejection ranges from 0.057 to 0.090, whereas Wald null rejection ranges from 0.047 to 0.077. Against $\beta=-0.1$, AR power ranges from 0.070 to 0.150, and Wald power ranges from 0.120 to 0.213. The comparison is between maintained null-score inference and regular Wald inference; there is no separate analytical projection shortcut in this paper. Boundary contact remains common in the weakest designs. The experiment shows the size-power trade-off without treating $F$ as a sufficient statistic for the projected experiment.

\section{Secondary evidence, applied guide, and extensions}
\label{app:secondary-guide}

This section collects secondary evidence, reporting guidance, and extensions moved from
the submission manuscript for concision.

\subsection{Secondary validation: domestic credit and house prices}
\label{sec:secondarycredit}

The earlier JST application remains as a secondary validation exercise. It relates
three-year real house-price growth to current real bank-credit growth, instrumented by
its second and third lags. In the common-support sample, the two-way fixed-effects
coefficient is $-0.517$ with first-stage
$F=20.53$. After a common rank-one projection, the coefficient is
$-0.919$ and $F=10.25$. Roughly 68 percent of residual instrument
variance survives projection, but the Gaussian-reference Anderson--Rubin confidence set is unbounded on the
reported $[-5,5]$ grid.

\subsection*{What the evidence changes}

The simulations and applications do not simply produce different standard errors. They
change the economic statement supported by the design. A very large raw first stage can
support the wrong sign when its variation is confounded; a persistent international
output response can disappear after common projection; and a large projected $F$ can
coexist with an unbounded robust confidence set. Nuisance robustness, moment sensitivity,
and inferential precision must therefore be evaluated jointly.

The largest first-stage statistic need not belong to the most informative valid design
once nuisance variation has been removed.

\section{A guide for applied researchers}\label{sec:guide}

The framework suggests a short reporting protocol whenever identification is assessed
after fixed effects, latent factors, machine-learning residualization, or another
substantial nuisance transformation.

\begin{enumerate}[leftmargin=2.2em,label=\textbf{\arabic*.}]
\item \textbf{Describe the transformation.}
State which variation is removed, why it is regarded as nuisance, and whether the same
final transformation is applied to the outcome, endogenous regressors, and instruments.

\item \textbf{Document surviving excluded variation.}
Report the source of instrument variation before and after projection and the fraction
that survives. A large nominal sample does not imply many independent identifying
directions.

\item \textbf{Report the projected Jacobian spectrum.}
The singular values summarize moment sensitivity after nuisance removal. Report effective
rank, the weakest retained singular value, and uncertainty about the associated subspace
when economically relevant.

\item \textbf{Separate sensitivity from precision.}
Report the projected-score covariance or the relevant full score law together with the
PIM when the Gaussian fixed-rank representation is appropriate. A first-stage statistic,
Jacobian, and PIM need not move together.

\item \textbf{Use identification-robust inference.}
Report null-restricted confidence sets or tests whose calibration matches the dependence
structure. Unbounded or disconnected sets are substantive findings, not computational
failures; the former can be required for valid coverage near nonidentification
\citep{dufour1997}.

\item \textbf{Check nuisance-estimation transfer.}
Show that factor or machine-learning remainders are negligible at the signal, score, and
spectral-gap scales used by the claimed result.

\item \textbf{Keep validity and information distinct.}
Projected-information diagnostics do not establish the economic exclusion restriction.
They assess what the maintained design can distinguish after the final transformation.
\end{enumerate}

This checklist is not a pretest for selecting a preferred estimator. It is a transparent
description of the statistical experiment on which the economic conclusion rests.

\section{Broader implications and extensions}\label{sec:broader}

The formal results use linear panel IV because it makes the information consequences of
projection transparent. The organizing principle is broader: whenever structural
inference follows nuisance removal, identification and information should be defined for
the final transformed moments. Three ingredients are needed in a new setting: a
transformation that preserves the target moment restriction, a local experiment for the
transformed score, and feasible remainder bounds showing that estimated nuisance objects
reproduce the oracle geometry. The present results do not make those extensions
automatic, but they identify the questions each extension must answer.

\paragraph{Nonlinear and semiparametric moments.}
For nonlinear GMM, quantile restrictions, distributional IV, and other semiparametric
models, the projected Jacobian is replaced by the derivative of the transformed moment
map at the target parameter. Its range and null space continue to describe locally
visible and invisible directions. The corresponding information object combines that
derivative with the covariance or full law of the transformed score. Curvature and
nonsmoothness create model-specific complications, but the distinction between nuisance
robustness and surviving identification remains unchanged.

\paragraph{Orthogonal scores and machine learning.}
Double/debiased machine learning constructs scores that are locally insensitive to
first-stage nuisance errors. Orthogonality protects the target moment from small nuisance
estimation mistakes; it does not guarantee that the residualized treatment or instrument
contains substantial identifying variation. Projected-information diagnostics can
therefore complement orthogonal-score inference by measuring the sensitivity and rank of
the final moment problem after machine-learning residualization.

\paragraph{Synthetic controls and matrix completion.}
Synthetic-control and matrix-completion methods remove low-rank counterfactual structure
before estimating treatment effects. Their credibility often improves as the latent
component becomes richer, but the remaining treated-versus-counterfactual contrast may
become weak. A projected-information analysis would ask which treatment-effect contrasts
remain visible after fitting the latent structure and whether factor-estimation error is
small relative to those contrasts and their spectral gaps.

\paragraph{Local projections and dynamic treatment effects.}
Horizon-specific local projections routinely absorb unit effects, time effects, trends,
and state controls. The information remaining for a shock coefficient can change sharply
with the horizon, even when the same nominal sample is used. The flagship application
illustrates this point: the projected Jacobian, first-stage statistic, PIM, and coefficient
path move differently across horizons. The same logic applies to dynamic treatment
effects and event-study designs, where increasingly rich controls may remove the
variation that separates long-run responses.

\paragraph{Network and multiway moment models.}
Dyadic, network, matched-employer--employee, and other multiway data combine nuisance
projections with dependence generated by shared indices. In those settings, the relevant
score law may remain non-Gaussian after projection, so covariance-based information is
only part of the experiment. Extending the framework requires joint analysis of the
projected derivative and the full multiway score limit, together with a bootstrap that
reproduces both.

These examples point to a common research program rather than a mechanical recipe.
Projection should be treated as part of the statistical experiment, not as preprocessing
whose information consequences can be ignored. The central empirical question remains:
which economically meaningful directions survive the final nuisance transformation, and
with what precision?

\section{Scalar concentration and asymmetric normalization: details}\label{app:scalar-concentration}

For one endogenous regressor and one instrument, write
\[
 \zeta_{i,NT}=T^{-1}\sum_t z_{it}^{o}v_{it}^{o},\qquad
 \Omega_{v,NT}=\operatorname{Var}\!\left(N^{-1/2}\sum_i\zeta_{i,NT}\right).
\]
Suppose $T^{\delta}\Omega_{v,NT}\to\Omega_v\in(0,\infty)$. A scalar signal-to-noise concentration sequence is
\begin{equation}
 \kappa_{N\mid T}
 =N\pi_{NT}^{2}\frac{(Q_{zz,NT}^{E})^{2}}{\Omega_{v,NT}},
 \label{eq:kappa-details}
\end{equation}
which converges to the expression using $Q_{zz}$ under Assumption~\ref{ass:moments}.
A nondegenerate local experiment requires
\begin{equation}
 \pi_{NT}=\frac{c}{r_{NT}},\qquad r_{NT}=\sqrt{NT^{\delta}},       \label{eq:local}
\end{equation}
whereas the unnormalized first-stage score is normalized by
\begin{equation}
 b_{NT}=T\sqrt{N\Omega_{v,NT}}\asymp\sqrt N\,T^{1-\delta/2}.     \label{eq:bnt}
\end{equation}
These rates coincide only when $\delta=1$.

\paragraph{Purpose.}
The next proposition derives the local first-stage sequence from an explicit
signal-to-noise balance. It clarifies why coefficient localization and score
normalization are distinct and may depend differently on the two panel dimensions.

\paragraph{Restatement of the balance result.}
Suppose $T^{\delta}\Omega_{v,NT}\to\Omega_v\in(0,\infty)$,
$\widehat Q_{zz,NT}\to_pQ_{zz}>0$, and
\[
 b_{NT}^{-1}\sum_{i,t}z_{it}^{o}v_{it}^{o}\Rightarrow Z_v,
 \qquad b_{NT}=T\sqrt{N\Omega_{v,NT}}.
\]
If $\pi_{NT}=c/\sqrt{NT^\delta}$, then
\[
 b_{NT}^{-1}\pi_{NT}\sum_{i,t}(z_{it}^{o})^2
 \to_p\mu_c:=cQ_{zz}/\sqrt{\Omega_v},
\]
and consequently
\[
 b_{NT}^{-1}\sum_{i,t}z_{it}^{o}x_{it}^{o}\Rightarrow\mu_c+Z_v.
\]

\paragraph{Interpretation.}
The local coefficient sequence is determined by the amount of independent first-stage
variation, not mechanically by the number of panel cells. Additional time periods can
improve nuisance estimation without increasing first-order identifying information at
the same rate.

\begin{remark}[Identification and dependence accumulate separately]
The local coefficient rate $r_{NT}$ is obtained by balancing the deterministic projected first-stage drift against the stochastic scale of the first-stage score; it is therefore not a property of the Jacobian alone. The score normalization $b_{NT}$ is determined by the dependence of the first-stage score. The structural-score normalization used for inference may be different from both. This separation is the defining feature of asymmetric identification.
\end{remark}

\section{Regular GMM: feasible optimality and scope}\label{app:regular-gmm-scope}

\paragraph{Feasible oracle covariance optimality.}
Suppose Proposition~\ref{prop:oracle} holds at the structural-score scale, Assumption~\ref{ass:regular-eff} holds for the oracle moments, and the feasible covariance estimator satisfies $\widehat\Omega_g-\Omega_g=o_p(1)$. Then feasible two-step GMM with $\mathsf W_{NT}=\widehat\Omega_g^{-1}$ has the same first-order distribution and the same within-class covariance minimum $V_{\mathrm{opt}}$ as the oracle projected estimator.

\begin{remark}[Scope of the covariance-optimality statement]
Theorem~\ref{thm:efficient-gmm} is deliberately restricted to strong identification and does not establish a semiparametric efficiency bound. When singular values of $J_{NT}$ shrink at the stochastic score rate, the estimator is nonregular and the classical GMM variance comparison is not an efficiency bound. When the limiting Jacobian is rank deficient, no regular point estimator exists for its null-space directions. An identification-dependent random weighting rule can improve finite-sample behavior in particular designs, but it cannot create information in unidentified directions and may change the weak-limit experiment. Such a rule requires a separate decision-theoretic analysis and is not justified by the regular GMM theorem.
\end{remark}

\begin{remark}[Practical implication]
Optimal weighting and identification-robust inference solve different problems. Efficient two-step GMM is appropriate for well-identified directions. Uniformly valid confidence statements across strong, weak, and rank-deficient regimes should continue to be based on null-restricted statistics such as Anderson--Rubin, with the dependence law approximated by the proposed bootstrap.
\end{remark}

\subsection{Proof of Theorem~\ref{thm:localmoment}}\label{app:proof:thm-localmoment}

\begin{proof}
The proof separates algebra from probability. First, the linear structural equation
gives an exact decomposition of the null and candidate scores under a local alternative.
Second, normalized score convergence and Jacobian convergence deliver the shift family.
Third, on the Gaussian fixed-rank stratum, the singular Gaussian likelihood ratio yields
the PIM and its rank characterization.

Under the local law indexed by $h$,
\[
 y_{it}^{o}-x_{it}^{o\prime}\beta_0
 =u_{it}^{o}+x_{it}^{o\prime}D_{NT}^{-1}h.
\]
Therefore
\[
 \mathcal S_{NT}^{0}(h)
 =A_{NT}^{-1}\sum_{i,t}z_{it}^{o}u_{it}^{o}
 +\widehat H_{NT}h.
\]
Assumption~\ref{ass:localmoment} and Slutsky's theorem give
$\mathcal S_{NT}^{0}(h)\Rightarrow Z+\mathcal Hh$. Similarly,
\[
 \mathcal S_{NT}(k;h)
 =A_{NT}^{-1}\sum_{i,t}z_{it}^{o}u_{it}^{o}
 +\widehat H_{NT}(h-k),
\]
which yields the second limit. Joint convergence over every finite collection follows
because all coordinates are continuous functions of the common pair
$\{A_{NT}^{-1}\sum z_{it}^{o}u_{it}^{o},\widehat H_{NT}\}$.

Now suppose $Z\sim N(0,\Omega)$ and
$\mathcal R(\mathcal H)\subseteq\mathcal R(\Omega)$. Let
$\Omega=P_r\Lambda_rP_r'$ be the positive-eigenvalue decomposition. Write
$\widetilde{\mathcal H}=\Lambda_r^{-1/2}P_r'\mathcal H$. Then
\[
 \mathcal I_P
 =\mathcal H'P_r\Lambda_r^{-1}P_r'\mathcal H
 =\widetilde{\mathcal H}'\widetilde{\mathcal H}.
\]
Hence $\mathcal I_P$ is positive semidefinite,
$\rank(\mathcal I_P)=\rank(\mathcal H)$, and it is positive definite exactly when
$\mathcal H$ has full column rank.

The Gaussian shift $Y\sim N(\mathcal Hh,\Omega)$ is supported on
$\mathcal R(\Omega)$. In whitened coordinates
$\widetilde Y=\Lambda_r^{-1/2}P_r'Y$,
\[
 \widetilde Y\sim N(\widetilde{\mathcal H}h,I_r).
\]
The ordinary Gaussian density ratio relative to $h=0$ is
\[
 h'\widetilde{\mathcal H}'\widetilde Y
 -\frac12h'\widetilde{\mathcal H}'\widetilde{\mathcal H}h,
\]
which equals \eqref{eq:gaussian-shift-lr}.
\end{proof}

\subsection{Proof of Corollary~\ref{cor:ar-local-power}}\label{app:proof:cor-ar-local-power}

\begin{proof}
Theorem~\ref{thm:localmoment} and covariance consistency imply
\[
 \mathcal S_{NT}^{0}(h)'\widehat\Omega_{NT}^{\dagger}\mathcal S_{NT}^{0}(h)
 \Rightarrow
 (Z+\mathcal Hh)'\Omega^\dagger(Z+\mathcal Hh).
\]
In the whitened coordinates from the proof of Theorem~\ref{thm:localmoment}, the
limiting quadratic form is
\[
 \|\xi+\widetilde{\mathcal H}h\|^2,
 \qquad \xi\sim N(0,I_{r_\Omega}),
\]
which is noncentral chi-square with degrees of freedom $r_\Omega$ and noncentrality
\[
 \|\widetilde{\mathcal H}h\|^2
 =h'\mathcal H'\Omega^\dagger\mathcal Hh
 =h'\mathcal I_Ph.
\]
\end{proof}

\subsection{Proof of Corollary~\ref{cor:gaussian-blue}}\label{app:proof:cor-gaussian-blue}

\begin{proof}
Because $\mathcal I_P$ is positive definite,
\[
 \E[\widehat h_P]
 =\mathcal I_P^{-1}\mathcal H'\Omega^\dagger\mathcal Hh=h.
\]
The range condition implies $\Omega\Omega^\dagger\mathcal H=\mathcal H$, and therefore
\[
 \Var(\widehat h_P)
 =\mathcal I_P^{-1}\mathcal H'\Omega^\dagger
 \Omega\Omega^\dagger\mathcal H\mathcal I_P^{-1}
 =\mathcal I_P^{-1}.
\]

Let $KY$ be another linear unbiased estimator, so $K\mathcal H=I_p$, and define
$K_*=\mathcal I_P^{-1}\mathcal H'\Omega^\dagger$. Then
$(K-K_*)\mathcal H=0$ and
\[
 (K-K_*)\Omega K_*'
 =(K-K_*)\mathcal H\mathcal I_P^{-1}=0.
\]
Consequently,
\[
 K\Omega K'-K_*\Omega K_*'
 =(K-K_*)\Omega(K-K_*)'\succeq0.
\]
\end{proof}

\subsection{Proof of Theorem~\ref{thm:singular-inference}}\label{app:proof:thm-singular-inference}

\begin{proof}
Let $J=J_{NT}$, $\widehat J=J+\Delta$, and let $\sigma_j$ be simple. Consider the
symmetric dilation
\[
 \mathcal D(J)=
 \begin{pmatrix}0&J\\J'&0\end{pmatrix}.
\]
Its nonzero eigenvalues are $\{\pm\sigma_k(J)\}$. An eigenvector associated with
$\sigma_j(J)$ is
\[
 2^{-1/2}\begin{pmatrix}u_{j,NT}\\v_{j,NT}\end{pmatrix}.
\]
The perturbation is
\[
 \mathcal D(\Delta)=
 \begin{pmatrix}0&\Delta\\\Delta'&0\end{pmatrix},
 \qquad
 \|\mathcal D(\Delta)\|=\|\Delta\|.
\]
The first-order expansion for a simple eigenvalue of a symmetric matrix, with the
dilation gap $\delta_{j,NT}$ defined in the main paper, gives, on an event whose
probability approaches one,
\[
 \sigma_j(\widehat J)-\sigma_j(J)
 =u_{j,NT}'\Delta v_{j,NT}
 +O\!\left(\frac{\|\Delta\|^2}{\delta_{j,NT}}\right).
\]
Assumption~\ref{ass:jacobian-clt} gives $\|\Delta\|=O_p(s_{NT})$, so the remainder is
$O_p(s_{NT}^2/\delta_{j,NT})=o_p(s_{NT})$. The first limit follows from the continuous
mapping theorem.

Conditionally on the data, Weyl's inequality and
$s_{NT}/\delta_{j,NT}\to0$ imply that the sample singular value remains simple and
separated with probability approaching one. Apply the same expansion around
$\widehat J_{NT}$. Wedin's theorem gives convergence of the empirical singular vectors,
and the bootstrap perturbation is $O_{p^*}(s_{NT})$ in probability. The bootstrap
remainder is therefore $o_{p^*}(s_{NT})$. Assumption~\ref{ass:jacobian-clt}, convergence of the empirical singular vectors, and
conditional Slutsky yield \eqref{eq:singular-bootstrap}. No continuity assumption on the
limit law is needed for this bounded-Lipschitz conclusion.
\end{proof}

\subsection{Proof of Theorem~\ref{thm:rank-consistency}}\label{app:proof:thm-rank-consistency}

\begin{proof}
Assumption~\ref{ass:jacobian-clt} implies
$\|\Delta_{NT}\|=O_p(s_{NT})$. Hence Weyl's inequality gives
\[
 \max_{j\le p}
 |\sigma_j(\widehat J_{NT})-\sigma_j(J_{NT})|
 \le \|\Delta_{NT}\|=O_p(s_{NT}).
\]
Since $s_{NT}/c_{NT}\to0$, every zero population singular value is estimated below
$c_{NT}$ with probability approaching one. If $r\ge1$, the condition
$\sigma_r(J_{NT})/c_{NT}\to\infty$ implies that every nonzero singular value exceeds $c_{NT}$ by
more than the stochastic error with probability approaching one. If $r=0$, there are no
nonzero singular values to retain. Hence
$\Pr(\widehat r_{NT}=r)\to1$.
\end{proof}

\subsection{Proof of Theorem~\ref{thm:subspace-reliability}}\label{app:proof:thm-subspace-reliability}

\begin{proof}
Let $E_{NT}=\widehat J_{NT}-J_{NT}$. A global sin--theta bound for singular
subspaces gives, for a universal constant $C$,
\[
 \|\widehat P_R-P_R\|
 \le
 C\min\left\{1,\frac{\|E_{NT}\|}{\gamma_{NT}}\right\}.
\]
(The truncation by one reflects the fact that orthogonal-projector distances are
uniformly bounded.) Since $\|E_{NT}\|=\|\Delta_{NT}\|=O_p(s_{NT})$, monotonicity of
$x\mapsto\min\{1,x\}$ yields
\[
 \|\widehat P_R-P_R\|
 =
 O_p\!\left(\min\left\{1,\frac{s_{NT}}{\gamma_{NT}}\right\}\right).
\]
If $s_{NT}/\gamma_{NT}\to0$, the right-hand side is $o_p(1)$. Without that rate
separation, the bound does not imply consistency, which proves the final qualification.
\end{proof}

\subsection{Proof of Corollary~\ref{cor:pim-directional}}\label{app:proof:cor-pim-directional}

\begin{proof}
Apply the continuous linear functional $M\mapsto h'Mh$ to
\eqref{eq:pim-delta}. Bootstrap validity follows from the same continuous mapping.
\end{proof}

\subsection{Proof of Corollary~\ref{cor:pim-eigen}}\label{app:proof:cor-pim-eigen}

\begin{proof}
For a simple separated eigenvalue of a symmetric matrix, first-order perturbation gives
\[
 \lambda_j(\widehat{\mathcal I}_{P,NT})
 -\lambda_j(\mathcal I_P)
 =
 e_j'
 \{\widehat{\mathcal I}_{P,NT}-\mathcal I_P\}
 e_j
 +o_p(\ell_{NT}).
\]
Combine this expansion with Theorem~\ref{thm:pim-delta}. The conditional bootstrap
argument is identical.
\end{proof}

\subsection{Proof of Corollary~\ref{cor:pim-condition}}\label{app:proof:cor-pim-condition}

\begin{proof}
Apply Corollary~\ref{cor:pim-eigen} to the largest and smallest simple eigenvalues and
then apply the ordinary delta method to the smooth map $(a,b)\mapsto a/b$ at
$b=\lambda_{\min}(\mathcal I_P)>0$.
\end{proof}

\subsection{Proof of Theorem~\ref{thm:pim-delta}}\label{app:proof:thm-pim-delta}

\begin{proof}
The argument is a fixed-rank delta method. We first linearize the Moore--Penrose
inverse along the tangent space of the covariance-rank manifold. We then expand
$H'\Omega^\dagger H$ in its Jacobian and covariance arguments and collect the first-order
terms. The bootstrap follows from the same expansion on the conditional same-rank event.

Write
$\Delta_H=\widehat H_{NT}-\mathcal H$ and
$\Delta_\Omega=\widehat\Omega_{NT}-\Omega$. On the symmetric fixed-rank manifold, and along perturbations satisfying the tangent
condition in Assumption~\ref{ass:pim-clt}, the Moore--Penrose inverse is Fr\'echet
differentiable and
\[
 \widehat\Omega_{NT}^{\dagger}
 =
 \Omega^\dagger
 +\dot\Omega^\dagger[\Delta_\Omega]
 +o_p(\|\Delta_\Omega\|).
\]
Substitute this expansion and
$\widehat H_{NT}=\mathcal H+\Delta_H$ into
$\widehat{\mathcal I}_{P,NT}
=\widehat H_{NT}'
 \widehat\Omega_{NT}^{\dagger}
 \widehat H_{NT}$.
The three terms linear in $(\Delta_H,\Delta_\Omega)$ are
\[
 \Delta_H'\Omega^\dagger\mathcal H,\qquad
 \mathcal H'\Omega^\dagger\Delta_H,\qquad
 \mathcal H'\dot\Omega^\dagger[\Delta_\Omega]\mathcal H.
\]
Assumption~\ref{ass:pim-clt} implies
$\|\Delta_H\|+\|\Delta_\Omega\|=O_p(\ell_{NT})$.
Every remaining product contains at least two first-order perturbations, or one
$O_p(\ell_{NT})$ perturbation multiplied by the
$o_p(\ell_{NT})$ inverse-expansion remainder. It is therefore
$o_p(\ell_{NT})$. This proves \eqref{eq:pim-expansion}. Applying the joint
continuous mapping theorem to \eqref{eq:pim-joint-clt} gives
\eqref{eq:pim-delta}--\eqref{eq:pim-linear-map}.

For the bootstrap, Assumption~\ref{ass:pim-clt} implies that both
$\widehat\Omega_{NT}$ and $\widehat\Omega_{NT}^*$ have rank $r_\Omega$ with the required
unconditional and conditional probabilities. The bootstrap tangency condition controls
the scaled normal-normal secant block. On that event, apply the same fixed-rank expansion
around $(\widehat H_{NT},\widehat\Omega_{NT})$. Consistency of the derivative map and
the assumed joint bootstrap law imply conditional convergence of the linearized
bootstrap statistic to $\mathbb Z_{\mathcal I}$. The bootstrap remainder is
$o_{p^*}(1)$ after division by $\ell_{NT}$, proving
\eqref{eq:pim-bootstrap-law}.
\end{proof}

\subsection{Proof of Theorem~\ref{thm:pim-bootstrap-region}}\label{app:proof:thm-pim-bootstrap-region}

\begin{proof}
Work on the event that both $\widehat\Omega_{NT}$ and
$\widehat\Omega_{NT}^*$ have rank $r_\Omega$. Its unconditional and conditional
probabilities converge to one by Assumption~\ref{ass:pim-clt}; on this event the
Moore--Penrose map is continuous on the maintained fixed-rank manifold.

Theorem~\ref{thm:pim-delta} and the continuous mapping theorem imply
\[
 \ell_{NT}^{-1}
 \|\widehat{\mathcal I}_{P,NT}-\mathcal I_P\|_{\mathsf F}
 \Rightarrow
 \|\mathbb Z_{\mathcal I}\|_{\mathsf F}.
\]
The bootstrap version converges conditionally to the same law. Continuity at the target
quantile yields conditional quantile consistency. Therefore,
\[
 \Pr\!\left(
 \|\widehat{\mathcal I}_{P,NT}-\mathcal I_P\|_{\mathsf F}
 \le \ell_{NT}\widehat c_{1-\alpha}^{*,\mathcal I}
 \right)\to1-\alpha,
\]
which is equivalent to coverage of \eqref{eq:pim-region}.
\end{proof}

\subsection{Proof of Corollary~\ref{cor:singular-ci}}\label{app:proof:cor-singular-ci}

\begin{proof}
Theorem~\ref{thm:singular-inference} gives conditional weak convergence of the centered
bootstrap singular value to the same continuous limit as the centered sample singular
value. Continuity at the two target quantiles implies conditional quantile consistency.
The stated interval is obtained by inverting the centered bootstrap inequalities, and
its coverage therefore converges to $1-\alpha$.

If $\widehat s_{NT}/s_{NT}\to_p1$, replacing $s_{NT}$ by
$\widehat s_{NT}$ follows from Slutsky's theorem, both in the statistic and in the
interval endpoints.
\end{proof}

\subsection{Proof of Proposition~\ref{prop:oracle}}\label{app:proof:prop-oracle}

\begin{proof}
Assumption~\ref{ass:factororth} states exactly that the feasible structural score,
first-stage score, and instrument second moment differ from their oracle counterparts by
$o_p(1)$ after the normalizations used by the relevant oracle result. Joint Slutsky and
the continuous mapping theorem therefore transfer every finite-dimensional oracle limit
that is continuous in these inputs. Under the additional theorem-specific condition in
Proposition~\ref{prop:oracle}, the feasible Jacobian error is
$o_p(\varrho_{NT})$. Combining this error with the oracle sampling error and applying
Proposition~\ref{prop:stability} transfers every spectral cluster whose signal and
separating gap dominate the total perturbation. The same
argument fails at a boundary where the signal or gap is of the perturbation order, which
explains the stated qualification.
\end{proof}

\subsection{Proof of Corollary~\ref{cor:multiAR}}\label{app:proof:cor-multiAR}

\begin{proof}
Under $H_0:\beta=\beta_0$, the projected structural equation gives
$y_{it}^{o}-x_{it}^{o\prime}\beta_0=u_{it}^{o}$. Hence the multivariate null score is
$z_{it}^{o}u_{it}^{o}$ and contains no local first-stage matrix. The asserted reduction
and identification robustness follow directly. Any reference distribution or bootstrap
validity additionally requires the covariance-rank and score-law conditions stated in
the corresponding inference theorem.
\end{proof}

\subsection{Proof of feasible oracle covariance optimality}\label{app:proof:cor-oracle-eff}

\begin{proof}
Proposition~\ref{prop:oracle} makes the feasible score and Jacobian first-order equivalent
to their oracle counterparts at the regular score scale. Consistency of
$\widehat\Omega_g$ implies $\widehat\Omega_g^{-1}\to_p\Omega_g^{-1}$. Substitution into
the exact linear-GMM representation in the proof of Theorem~\ref{thm:efficient-gmm},
followed by Slutsky's theorem, yields the same limiting distribution and covariance
minimum as the oracle estimator.
\end{proof}

\subsection{Proof of Proposition~\ref{prop:boot}}\label{app:proof:prop-boot}

\begin{proof}
Under the null, Proposition~\ref{prop:ar} gives the oracle score
$g_{it}(\beta_0)=z_{it}^{o}u_{it}^{o}$, so the local reduced-form coefficient does not
enter its centering. The PWB-H theorem can therefore be applied to this null-score array
whenever its exact dependence, moment, bandwidth, scaling, and regime-classifier
conditions hold. Replacing the oracle score by the feasible factor-adjusted score is
valid when the conditional bootstrap analogue of Proposition~\ref{prop:oracle} makes all
inputs asymptotically equivalent at the PWB-H normalization. Uniformity over a local
first-stage class follows only if those transfer bounds hold uniformly over that class.
\end{proof}

\subsection{Proof of Theorem~\ref{thm:primitive-benchmark}}\label{app:proof:thm-primitive-benchmark}

\begin{proof}
The proof has three steps. We first obtain separate central limit theorems for the
unit, time, and cell components. We then combine the unit and time limits through the
bilinear mapping that generates the interaction term. Finally, independence of the
primitive blocks yields the stated joint limit.

Summing \eqref{eq:primitive-score-decomp} over the panel gives the exact identity
\[
 \sum_{i,t}\xi_{it,NT}
 =T\sum_i a_{i,NT}
 +N\sum_t d_{t,NT}
 +\mathfrak h_{NT}\!\left(\sum_i\phi_{i,NT},\sum_t\psi_{t,NT}\right)
 +\sum_{i,t}\varepsilon_{it,NT},
\]
where the $k$th coordinate of $\mathfrak h_{NT}(x,y)$ is $x'H_{k,NT}y$.
For the $k$th coordinate,
\[
\begin{aligned}
 \sum_{i,t}\phi_{i,NT}'H_{k,NT}\psi_{t,NT}
 &=
 \left(\sum_i\phi_{i,NT}\right)'H_{k,NT}
 \left(\sum_t\psi_{t,NT}\right)\\
 &=
 \sqrt{NT}
 \left(N^{-1/2}\sum_i\phi_{i,NT}\right)'H_{k,NT}
 \left(T^{-1/2}\sum_t\psi_{t,NT}\right).
\end{aligned}
\]
Hence the interaction accumulates at $\sqrt{NT}$, the same order as the cell term and
strictly slower than either main effect.

By the multivariate central limit theorem for fixed-range dependent triangular arrays,
the uniform $(2+\eta)$ moment bound, $m_A$-dependence, and convergence of
\[
 \Gamma_{A,NT}(0)+\sum_{h=1}^{m_A}
 \{\Gamma_{A,NT}(h)+\Gamma_{A,NT}(h)'\}
\]
imply
\[
 \frac1{\sqrt N}\sum_i
 \begin{pmatrix}a_{i,NT}\\ \phi_{i,NT}\end{pmatrix}
 \Rightarrow
 \begin{pmatrix}Z_a\\ Z_\phi\end{pmatrix},
 \qquad
 \frac1{\sqrt T}\sum_t
 \begin{pmatrix}d_{t,NT}\\ \psi_{t,NT}\end{pmatrix}
 \Rightarrow
 \begin{pmatrix}Z_d\\ Z_\psi\end{pmatrix}.
\]
For completeness, the result follows by applying the scalar $m$-dependent CLT to every
Cramér--Wold linear combination. The uniform $(2+\eta)$ moment bound gives the
Lindeberg condition, and the finite dependence ranges reduce the asymptotic variance to
the finite covariance sums in Assumption~\ref{ass:primitivebenchmark-details}. Independence of
the unit and time collections gives joint independence of the two Gaussian blocks.

Next,
\[
 \frac1{\sqrt{NT}}
 \mathfrak h_{NT}\!\left(\sum_i\phi_{i,NT},\sum_t\psi_{t,NT}\right)
 =
 \mathfrak h_{NT}\!\left(
 N^{-1/2}\sum_i\phi_{i,NT},
 T^{-1/2}\sum_t\psi_{t,NT}\right).
\]
Because $H_{k,NT}\to H_k$ for every $k$ and the dimensions are fixed, the bilinear
maps $\mathfrak h_{NT}$ converge to $\mathfrak h$ uniformly on compact sets. Indeed, for
any $(x_{NT},y_{NT})\to(x,y)$,
\[
\begin{aligned}
 &\|\mathfrak h_{NT}(x_{NT},y_{NT})-\mathfrak h(x,y)\|\\
 &\quad\le
 C\max_k\|H_{k,NT}-H_k\|\,\|x_{NT}\|\,\|y_{NT}\|\\
 &\qquad\quad+
 \|\mathfrak h(x_{NT},y_{NT})-\mathfrak h(x,y)\|
 \to0.
\end{aligned}
\]
where $C$ depends only on the fixed output dimension. The partial-sum arguments are tight,
so the extended continuous mapping theorem gives convergence of this term to
$\mathfrak h(Z_\phi,Z_\psi)$.

The cell innovations form an independent triangular array that is independent of the
unit and time collections. The covariance and Lindeberg conditions therefore imply, by
the multivariate Lindeberg--Feller theorem,
\[
 (NT)^{-1/2}\sum_{i,t}\varepsilon_{it,NT}
 \Rightarrow Z_\varepsilon\sim N(0,\Sigma_\varepsilon).
\]
The cell array is independent of the unit and time collections, so this convergence is
joint with the preceding partial-sum limits and $Z_\varepsilon$ is independent of the
other Gaussian blocks.

Finally, use
\[
 K_{NT}^{-1}
 =(T\sqrt N)^{-1}P_A+(N\sqrt T)^{-1}P_D+(\sqrt{NT})^{-1}P_0.
\]
The coordinate restrictions in Assumption~\ref{ass:primitivebenchmark-details}(v) place the unit,
time, and interaction/cell aggregates in $V_A$, $V_D$, and $V_0$, respectively. Each
aggregate is therefore multiplied only by its own rate. Orthogonality of the projectors
eliminates every cross-block term, so no divergent factor such as $\sqrt T$ or
$\sqrt N$ remains. The joint convergence established above and Slutsky's theorem yield
\[
 K_{NT}^{-1}\sum_{i,t}\xi_{it,NT}
 \Rightarrow
 G_A Z_a+G_D Z_d+G_0\{\mathfrak h(Z_\phi,Z_\psi)+Z_\varepsilon\}.
\]
On the maintained nondegenerate Gaussian supports, a nonzero bilinear term is a genuinely quadratic Gaussian polynomial and is therefore non-Gaussian; degeneracy includes the case in which the bilinear map vanishes on those supports. The final block-diagonal
statement follows by partitioning the left-hand side and $K_{NT}$ conformably.
\end{proof}

\subsection{Proof of Theorem~\ref{thm:pcbenchmark}}\label{app:proof:thm-pcbenchmark}

\begin{proof}
Write the auxiliary left covariance as
\[
 S_{L,NT}^{\mathrm{tr}}
 =S_{L,NT}^{0}+\Delta_{L,NT},
 \qquad
 S_{L,NT}^{0}=\frac1T\sum_{\ell=1}^L
 N_{\ell,NT}N_{\ell,NT}'.
\]
For each $\ell$,
\[
 A_\ell A_\ell'-N_\ell N_\ell'
 =N_\ell E_\ell'+E_\ell N_\ell'+E_\ell E_\ell'.
\]
Because $L$ is fixed,
\[
 \|\Delta_{L,NT}\|
 \le \frac1T\sum_{\ell=1}^L
 \{2\|N_{\ell,NT}\|\|E_{\ell,NT}\|+\|E_{\ell,NT}\|^2\}.
\]
Using the rates in Assumption~\ref{ass:pcbenchmark},
\[
 \frac{\|\Delta_{L,NT}\|}{N}
 =O_p(N^{-1/2}+T^{-1/2})=O_p(\eta_{NT}).
\]
The signal eigengap is at least $cN$. Davis--Kahan, equivalently Wedin applied to the
associated singular spaces, therefore gives
\[
 \|\widehat P_\Lambda-P_{\mathcal L_0}\|
 \le \frac{2\|\Delta_{L,NT}\|}{cN}
 =O_p(\eta_{NT}).
\]
The same argument applied to
$S_{F,NT}^{\mathrm{tr}}$ yields
$\|\widehat P_F-P_{\mathcal F_0}\|=O_p(\eta_{NT})$.

Let
$\Delta_\Lambda=\widehat P_\Lambda-P_{\mathcal L_0}$ and
$\Delta_F=\widehat P_F-P_{\mathcal F_0}$. For every evaluation-fold matrix $A$,
\[
 \widehat{\mathcal M}A-\mathcal M_0A
 =-\Delta_\Lambda AM_{\mathcal F_0}
  -M_{\mathcal L_0}A\Delta_F
  +\Delta_\Lambda A\Delta_F.
\]
Substitute this identity into each feasible structural-score product. Conditional on the
training sigma-field, the score difference is the aggregate remainder
$R_{g,NT}^{\mathrm{tr}}(\beta)$. Assumption~\ref{ass:pcbenchmark}(iv) controls both the
supremum of its conditional mean and the conditional second moment of the centered
supremum. Conditional Markov's inequality, followed by iterated expectation, therefore
gives
\[
 \sup_{\beta\in\mathcal B}
 \left\|\sum_{i,t}\{\widehat g_{it}(\beta)-g_{it}(\beta)\}\right\|
 =O_p(\sqrt{NT}\eta_{NT}).
\]
The same expansion and conditional bound give the displayed first-stage-score rate.
The sample-average second-moment and Jacobian rates are imposed in part (iv) and equal
$O_p(\eta_{NT})$.

Finally, part (v) implies that the two aggregate score remainders are $o_p(a_{NT})$ and
$o_p(b_{NT})$, while the Jacobian error is $o_p(\varrho_{NT})$. Thus the high-level
transfer and nonabsorption conditions hold. The block-cross-fitted version follows by
applying the same conditional argument on each fixed number of folds and summing the
fold-specific remainders.
\end{proof}

\subsection{Proof of Lemma~\ref{lem:projectioncompat} (Projection compatibility)}\label{app:proof:lem-projectioncompat}
\begin{proof}
Fix $r$. If $\operatorname{col}(L_r)\subseteq\mathcal L_0$, then $M_{\mathcal L_0}L_r=0$; if $\operatorname{col}(F_r)\subseteq\mathcal F_0$, then $F_r'M_{\mathcal F_0}=0$. In either case,
\[
 M_{\mathcal L_0}L_rF_r'M_{\mathcal F_0}=0.
\]
Summing over $r$ and using linearity proves $\mathcal M_0(N)=0$. Applying the same operator to every term of each maintained equation leaves the common coefficient matrices unchanged and removes the nuisance terms. The condition is sufficient, not necessary: cancellation may also occur for matrices outside these spans, but such cancellation is not stable enough to serve as a maintained model restriction.
\end{proof}

\subsection{Proof of Theorem~\ref{thm:identification-map} (Projected identification map)}\label{app:proof:thm-identification-map}
\begin{proof}
For any $\beta$, substitute $y_{it}^{o}=x_{it}^{o\prime}\beta_0+u_{it}^{o}$ into \eqref{eq:idmap}. Population-average exclusion gives
\[
 \mathcal Q_{NT}(\beta)
 =\frac1{NT}\sum_{i,t}\mathbb E[z_{it}^{o}u_{it}^{o}]
   -\left\{\frac1{NT}\sum_{i,t}\mathbb E[z_{it}^{o}x_{it}^{o\prime}]\right\}(\beta-\beta_0)
 =-J_{NT}(\beta-\beta_0).
\]
Thus $\beta$ and $\beta_0$ generate the same population moments exactly when $\beta-\beta_0\in\mathcal N(J_{NT})$. On an unrestricted affine parameter space the root is unique if and only if $\mathcal N(J_{NT})=\{0\}$, equivalently $\rank(J_{NT})=p$. On a restricted set $\mathcal B$, the root $\beta_0$ is unique exactly when
\[
 \mathcal N(J_{NT})\cap(\mathcal B-\beta_0)=\{0\}.
\]
The stronger condition
$\mathcal N(J_{NT})\cap(\mathcal B-\mathcal B)=\{0\}$
is equivalent to injectivity of the moment map over every pair of points in $\mathcal B$.
For the singular-space statement, write the thin SVD as $J_{NT}=L_{NT}S_{NT}R_{NT}'$. A perturbation $h$ changes the moment only through $S_{NT}R_{NT}'h$; hence the right singular subspaces are the structural directions seen by the moments, and their singular values are the corresponding deterministic gains. Finally, if $J_{NT}\to J_0$, the limiting zero set is $\{\beta:J_0(\beta-\beta_0)=0\}=\beta_0+\mathcal N(J_0)$. Since $\mathbb R^p=\mathcal R(J_0')\oplus\mathcal N(J_0)$, only the row-space projection of $\beta_0$ is pinned down by the limiting moments.
\end{proof}

\subsection{Proof of Proposition~\ref{prop:stability} (Singular-value and subspace stability)}\label{app:proof:prop-stability}
\begin{proof}
Let $E_{NT}=\widehat J_{NT}-J_{NT}$. By the triangle inequality,
$\|E_{NT}\|\le \tau_{NT}+\rho_{NT}=e_{NT}$. Weyl's singular-value inequality gives
\[
 |\sigma_j(\widehat J_{NT})-\sigma_j(J_{NT})|
 \le \|E_{NT}\|\le e_{NT}.
\]
For a cluster separated by $\gamma_{NT}$, Wedin's sin--theta theorem implies, on
$\{\|E_{NT}\|<\gamma_{NT}/2\}$,
\[
 \|\widehat P_R-P_R\|
 \le \frac{2\|E_{NT}\|}{\gamma_{NT}}
 \le \frac{2e_{NT}}{\gamma_{NT}}.
\]
This proves the eventwise bound. If $e_{NT}=O_p(r_{NT})$ and
$r_{NT}=o(\gamma_{NT})$, the event has probability approaching one and the stated
$O_p(r_{NT}/\gamma_{NT})$ rate follows.

For part (i), Weyl's inequality implies
$\sigma_p(\widehat J_{NT})\ge \sigma_p(J_{NT})-e_{NT}$.
The two rate restrictions make the right-hand side exceed $c_{NT}$ with probability
approaching one. Part (ii) is a statement about what the perturbation bound can certify:
if both the population singular value and estimation error are of order $d_{NT}$, then
the Weyl interval
$[\{\sigma_j(\widehat J_{NT})-e_{NT}\}_+,\sigma_j(\widehat J_{NT})+e_{NT}]$
has width of the same order as its distance from zero. No uniform impossibility claim is
deduced without specifying a statistical experiment. Part (iii) follows by applying the
same singular-value and projector bounds to every retained cluster.
\end{proof}

\subsection{Proof of Proposition~\ref{prop:balance} (Weak-first-stage balance)}\label{app:proof:prop-balance}
\begin{proof}
The reduced form gives
$\sum_{i,t}z_{it}^{o}x_{it}^{o}
=\pi_{NT}\sum_{i,t}(z_{it}^{o})^2+\sum_{i,t}z_{it}^{o}v_{it}^{o}$.
Moreover,
\[
 \frac{NT\pi_{NT}\widehat Q_{zz,NT}}{b_{NT}}
 =
 \frac{c\,\widehat Q_{zz,NT}}{\sqrt{T^\delta\Omega_{v,NT}}}
 \to_p\frac{cQ_{zz}}{\sqrt{\Omega_v}}.
\]
The conclusion follows from the assumed score convergence and Slutsky's theorem.
\end{proof}

\subsection{Proof of Theorem~\ref{thm:ratio} (Weak-IV ratio limit)}\label{app:proof:thm-ratio}
\begin{proof}
The proof makes the random denominator explicit. We first write the exact IV
estimation error as the structural score divided by the projected first stage. We then
insert the two normalizations and apply the joint weak limit. The final cases depend only
on the relative rate $a_{NT}/b_{NT}$.

Start from the exact identity in the preceding display and multiply by $b_{NT}/a_{NT}$:
\[
 \frac{b_{NT}}{a_{NT}}(\widehat\beta_{IV}^{o}-\beta_0)
 =\frac{a_{NT}^{-1}S_{u,NT}}
 {b_{NT}^{-1}NT\pi_{NT}\widehat Q_{zz,NT}+b_{NT}^{-1}S_{v,NT}}.
\]
Proposition~\ref{prop:balance} gives
$b_{NT}^{-1}NT\pi_{NT}\widehat Q_{zz,NT}\to_p\mu_c$, with
$\mu_c=cQ_{zz}/\sqrt{\Omega_v}$, while Assumption~\ref{ass:joint} gives joint convergence of the normalized numerator and stochastic denominator component. Hence the denominator converges jointly to $\mu_c+Z_v$. The map $(x,y)\mapsto x/y$ is continuous on $\mathbb R\times(\mathbb R\setminus\{0\})$; the no-atom condition $\Pr(\mu_c+Z_v=0)=0$ therefore permits the extended continuous mapping theorem and yields \eqref{eq:ratio}. The three rate statements follow by writing
$\widehat\beta_{IV}^{o}-\beta_0=(a_{NT}/b_{NT})R_{NT}$ with $R_{NT}\Rightarrow Z_u/(\mu_c+Z_v)$.
\end{proof}

\subsection{Proof of Theorem~\ref{thm:multiratio} (Multivariate weak-IV limit on a full-rank local stratum)}\label{app:proof:thm-multiratio}
\begin{proof}
Set $\mathsf A_{NT}=\mathsf G_{NT}C_{NT}^{-1}$ and
$\widetilde U_{NT}=a_{NT}^{-1}U_{NT}$. The assumed joint convergence gives
$(\widetilde U_{NT},\mathsf A_{NT},\mathsf W_{NT})
\Rightarrow(Z_u,\mathcal H,\mathsf W)$.
Since $\mathsf G_{NT}=\mathsf A_{NT}C_{NT}$, the linear-GMM first-order condition yields, whenever
$\mathsf G_{NT}$ has full column rank,
\[
 a_{NT}^{-1}C_{NT}(\widehat\beta-\beta_0)
 =(\mathsf A_{NT}'\mathsf W_{NT}\mathsf A_{NT})^{-1}\mathsf A_{NT}'\mathsf W_{NT}\widetilde U_{NT}.
\]
Because $\mathsf W$ is positive definite and $\mathcal H$ has full column rank almost surely,
\[
 \lambda_{\min}(\mathcal H'\mathsf W\mathcal H)
 \ge \lambda_{\min}(\mathsf W)\,\sigma_p(\mathcal H)^2>0
 \quad\text{almost surely}.
\]
Thus the limiting Gram matrix is nonsingular, and the probability that the sample Gram
matrix is nonsingular tends to one. The map
$(u,A,W)\mapsto(A'WA)^{-1}A'Wu$ is continuous on this full-rank domain. The continuous
mapping theorem gives the displayed limit.

This is an oracle statement. A feasible factor-adjusted estimator has the same limit only
if its score, first-stage matrix, and weight matrix differ from the oracle objects by
remainders negligible at the $a_{NT}$ and column-specific $C_{NT}$ scales. If
$\mathcal H$ loses rank, the inverse map is discontinuous. A Moore--Penrose or regularized
selection then depends on an explicit anchor and does not recover an unrestricted
null-space component of $\beta_0$.
\end{proof}

\subsection{Proof of Theorem~\ref{thm:efficient-gmm} (Regular projected-GMM limit and covariance-optimal weighting)}\label{app:proof:thm-efficient-gmm}
\begin{proof}
The proof has two parts. Exact linearity gives the estimator representation and hence
the regular limit. A projection argument then compares the covariance of an arbitrary
linear-GMM rule with the covariance under inverse-score weighting.

Use the sample-average Jacobian $\widehat J_{NT}^{\,s}$ defined in the main paper. In the linear model,
\[
 \bar g_{NT}(\beta)
 =\bar g_{NT}(\beta_0)-\widehat J_{NT}^{\,s}(\beta-\beta_0)
\]
exactly. On the event that $\widehat J_{NT}^{\,s\prime}\mathsf W_{NT}\widehat J_{NT}^{\,s}$ is nonsingular, the
unconstrained minimizer therefore satisfies
\[
 \widehat\beta_W-\beta_0
 =(\widehat J_{NT}^{\,s\prime}\mathsf W_{NT}\widehat J_{NT}^{\,s})^{-1}
 \widehat J_{NT}^{\,s\prime}\mathsf W_{NT}\bar g_{NT}(\beta_0).
\]
Assumption~\ref{ass:regular-eff} gives $\widehat J_{NT}^{\,s}\to_pJ_0$ and $\mathsf W_{NT}\to_p\mathsf W$.
Multiplying by $NT/a_{NT}$ and using
$a_{NT}^{-1}\sum_{i,t}g_{it}(\beta_0)\Rightarrow Z_g$ yields
\[
 \frac{NT}{a_{NT}}(\widehat\beta_W-\beta_0)
 \Rightarrow (J_0'\mathsf WJ_0)^{-1}J_0'\mathsf WZ_g.
\]
The covariance formula follows by direct calculation.

For covariance optimality, define
$K_{\mathsf W}=(J_0'\mathsf WJ_0)^{-1}J_0'\mathsf W$ and
$K_*=(J_0'\Omega_g^{-1}J_0)^{-1}J_0'\Omega_g^{-1}$.
Both satisfy $K_{\mathsf W}J_0=K_*J_0=I_p$. Hence
\[
 K_{\mathsf W}\Omega_gK_{\mathsf W}'-K_*\Omega_gK_*'
 =(K_{\mathsf W}-K_*)\Omega_g(K_{\mathsf W}-K_*)'
\]
because the cross terms vanish:
$(K_{\mathsf W}-K_*)\Omega_gK_*'
=(K_{\mathsf W}-K_*)J_0(J_0'\Omega_g^{-1}J_0)^{-1}=0$.
The right-hand side is positive semidefinite, proving Loewner-order covariance
optimality within the stated regular linear-GMM class.
\end{proof}

\subsection{Proof of Proposition~\ref{prop:wald} (Failure of conventional Wald pivotality)}\label{app:proof:prop-wald}
\begin{proof}
Theorem~\ref{thm:ratio} and $a_{NT}/b_{NT}\to\rho$ imply
$\widehat\beta_{IV}^{o}-\beta_0\Rightarrow\rho Z_u/(\mu_c+Z_v)$.
The stated variance convergence and Slutsky's theorem give the displayed limit; the common factor $(a_{NT}/b_{NT})^2$ cancels.
\end{proof}

\subsection{Proof of Proposition~\ref{prop:ar} (Identification robustness of the null score)}\label{app:proof:prop-ar}
\begin{proof}
Under the null, the maintained structural equation gives
$y_{it}^{o}-x_{it}^{o\prime}\beta_0=u_{it}^{o}$, and therefore
$g_{it}(\beta_0)=z_{it}^{o}u_{it}^{o}$. The reduced-form coefficient $\Pi_{NT}$ drops out of this identity. Uniform size is a separate probabilistic statement: it follows only if the covariance estimator, generalized-inverse operation on a fixed-rank stratum, and critical-value approximation are uniformly valid over the stated parameter class.
\end{proof}

\subsection{Proof of Theorem~\ref{thm:uniform-ar}}\label{app:proof:thm-uniform-ar}

The proof uses the same basic decomposition as generic uniform-size arguments: control
the distributional approximation, the estimated critical value, and the probability mass
near the reference quantile separately.
Every bound below is taken before the supremum over
$P\in\mathfrak P_{NT}$ is released; no pointwise limit is substituted into a uniform
statement. The statistic is evaluated directly, so the proof uses neither a subgradient
condition nor local convexity of an optimization criterion.

\begin{proof}
Write $c_{NT,P}=c_{NT,P}(1-\alpha)$ and
$\widehat c_{NT}=\widehat c_{NT}(1-\alpha)$. Fix $\varepsilon>0$.
By \eqref{eq:uniform-anticoncentration}, choose $\delta>0$ so that, for all sufficiently
large $(N,T)$,
\[
 \sup_{P\in\mathfrak P_{NT}}
 \{F_{NT,P}(c_{NT,P}+\delta)-F_{NT,P}(c_{NT,P}-\delta)\}
 \le\varepsilon.
\]
On the event $|\widehat c_{NT}-c_{NT,P}|\le\delta$,
\[
 \{T_{NT}>\widehat c_{NT}\}
 \subseteq
 \{T_{NT}>c_{NT,P}-\delta\}.
\]
Let
\[
 r_{NT}
 =
 \sup_{P\in\mathfrak P_{NT}}\sup_{x\in\mathbb R}
 \left|P\{T_{NT}\le x\}-F_{NT,P}(x)\right|.
\]
Because $F_{NT,P}(c_{NT,P})\ge1-\alpha$ by the generalized-quantile definition
(the reference test may be conservative when the cdf has an atom),
\[
\begin{aligned}
 P\{T_{NT}>c_{NT,P}-\delta\}
 &\le 1-F_{NT,P}(c_{NT,P}-\delta)+r_{NT}\\
 &\le \alpha+
 \{F_{NT,P}(c_{NT,P})-F_{NT,P}(c_{NT,P}-\delta)\}
 +r_{NT}\\
 &\le \alpha+\varepsilon+r_{NT}.
\end{aligned}
\]
Adding
$P\{|\widehat c_{NT}-c_{NT,P}|>\delta\}$, taking the supremum over
$P\in\mathfrak P_{NT}$ and the limsup over $(N,T)$, and then letting
$\varepsilon\downarrow0$ proves the size bound. The coverage statement follows from the
definition
$\widehat{\mathcal C}_{1-\alpha}
=\{\beta:T_{NT}(\beta)\le\widehat c_{NT,\beta}(1-\alpha)\}$,
with the same uniform approximation imposed at the true null value.
\end{proof}

\subsection{Proof of Proposition~\ref{prop:ar-limit} (Gaussian fixed-rank limit of the quadratic AR statistic)}\label{app:proof:prop-ar-limit}
\begin{proof}
Let
\[
 S_{NT}=c_{NT}^{-1}\sum_{i,t}g_{it}(\beta_0),
 \qquad
 \widehat\Omega_{NT}=c_{NT}^{-2}\widehat\Sigma_{g,NT}(\beta_0).
\]
The assumptions imply $(S_{NT},\widehat\Omega_{NT})\Rightarrow(Z_g,\Omega_0)$.
Because the rank is eventually $r_g$ and the positive eigenvalues of $\Omega_0$ are
bounded away from zero on the maintained stratum, the Moore--Penrose map is continuous
at $\Omega_0$. Hence
\[
 AR_{NT}(\beta_0)
 =S_{NT}'\widehat\Omega_{NT}^{\dagger}S_{NT}
 \Rightarrow Z_g'\Omega_0^{\dagger}Z_g.
\]
Write the spectral decomposition
$\Omega_0=P\operatorname{diag}(\lambda_1,\ldots,\lambda_{r_g},0,\ldots,0)P'$,
where each $\lambda_j>0$. Since $Z_g\sim N(0,\Omega_0)$, there is
$\xi\sim N(0,I_{r_g})$ such that
\[
 Z_g=P_{r_g}\operatorname{diag}(\lambda_1^{1/2},\ldots,\lambda_{r_g}^{1/2})\xi,
\]
with $P_{r_g}$ collecting the first $r_g$ eigenvectors. Therefore
\[
 Z_g'\Omega_0^{\dagger}Z_g
 =\xi'\xi\sim\chi^2_{r_g}.
\]
If $Z_g$ is non-Gaussian, the same continuous-mapping argument yields
$Z_g'\Omega_0^{\dagger}Z_g$, but that quadratic form is not generally chi-square.
\end{proof}

\subsection{Proof of Theorem~\ref{thm:primitive-bootstrap}}\label{app:proof:thm-primitive-bootstrap}

\begin{proof}
Because the dimensions are fixed, convergence of the covariance estimators implies
convergence of their symmetric positive-semidefinite square roots:
\[
 \|\widehat\Sigma_A^{1/2}-\Sigma_A^{1/2}\|
 +\|\widehat\Sigma_D^{1/2}-\Sigma_D^{1/2}\|
 +\|\widehat\Sigma_\varepsilon^{1/2}
      -\Sigma_\varepsilon^{1/2}\|\to_p0.
\]
On an enlarged probability space, let $U_A,U_D,U_\varepsilon$ be independent standard
Gaussian vectors of the appropriate dimensions, independent of the data. Construct the
bootstrap Gaussian blocks using the estimated square roots and construct target blocks
with the corresponding population square roots and the same standard Gaussian vectors.
The covariance-square-root, $H_k$, and $G$ consistency conditions imply
\[
 L_{NT}^*-L^*=o_{p^*}(1)
\]
in outer probability, where $L^*$ has the same law as $L$. For the bilinear term,
insert and subtract the versions that change one of
$Z_\phi^*$, $\widehat H_k$, and $Z_\psi^*$ at a time. Cauchy--Schwarz and fixed-dimensional
Gaussian moment bounds make every resulting term $o_{p^*}(1)$ in outer probability.

For every $f$ with bounded-Lipschitz norm at most one,
\[
 \left|\E^*f(L_{NT}^*)-\E f(L)\right|
 \le \E^*\{\|L_{NT}^*-L^*\|\wedge2\}+o_p(1)=o_p(1).
\]
Taking the supremum over the bounded-Lipschitz class proves
\eqref{eq:primitive-bootstrap-bl}. The continuous-mapping statement follows from the
same coupling. Conditional quantile consistency and the size conclusion follow from weak
convergence and continuity with strict increase at the target quantile.
\end{proof}

\subsection{Proof of Corollary~\ref{cor:primitive-bootstrap-quadratic}}\label{app:proof:cor-primitive-bootstrap-quadratic}

\begin{proof}
Because
$\widehat\Omega_{NT}\to_p\Omega_L$ and the estimated rank equals
$\rank(\Omega_L)$ with probability approaching one,
$\widehat\Omega_{NT}^{\dagger}\to_p\Omega_L^\dagger$. Theorem~\ref{thm:primitive-bootstrap},
the continuous mapping theorem, and Slutsky's theorem give conditional convergence of
\[
 L_{NT}^{*\prime}\widehat\Omega_{NT}^{\dagger}L_{NT}^*
\]
to $L'\Omega_L^\dagger L$. Continuity at the target quantile yields conditional quantile
consistency and asymptotic size.
\end{proof}

\subsection{Proof of Theorem~\ref{thm:exactpwb} (Fixed-rank feasible PWB-H transfer)}\label{app:proof:thm-exactpwb}
\begin{proof}
Let $\mathcal T_{NT}^{o}$ and $\mathcal T_{NT}^{*o}$ denote the oracle original-sample
and bootstrap PWB-H statistics, and let $\widehat{\mathcal T}_{NT}$ and
$\widehat{\mathcal T}_{NT}^{*}$ denote their feasible counterparts. By assumption, all
score, variance, discriminant, classifier, and fixed-rank inputs entering the statistic
differ by $o_p(1)$ in the original sample and by $o_{p^*}(1)$ conditionally in the bootstrap.
Continuity of the PWB-H map on the maintained regime and rank stratum therefore gives
\[
 d_{BL}\{\mathcal L^*(\widehat{\mathcal T}_{NT}^{*}),
          \mathcal L^*(\mathcal T_{NT}^{*o})\}=o_p(1),
 \qquad
 \widehat{\mathcal T}_{NT}-\mathcal T_{NT}^{o}=o_p(1),
\]
where $d_{BL}$ is bounded-Lipschitz distance. The oracle approximation in
\citet[Theorem~3.3]{hounyolin2026}, combined with these two displays and the triangle
inequality, yields the feasible approximation. The feasible regimes and the excluded
I\&N regime are inherited unchanged from the oracle theorem.
\end{proof}

The proofs are organized by logical status. Exact algebraic statements are proved directly. Distributional results are derived conditionally on the stated joint convergence assumptions. Oracle-to-feasible and bootstrap results are transfer theorems: their proofs verify continuity and negligible-remainder steps, while primitive rates for a particular factor estimator or bootstrap design must be established separately. This organization is deliberate and should be preserved in any revision or replication supplement.

\subsection{Sufficient-condition roadmap for the joint score limit}\label{app:primitive}
Theorem~\ref{thm:primitive-benchmark} provides a complete proof for a finite-range, finite-rank benchmark. Assumption~\ref{ass:primitiveclt} remains high level for broader spatially and serially dependent arrays because a generic mixing statement does not determine the normalization or limit law. A more general primitive verification proceeds component by component.

\begin{enumerate}[label=\textbf{P\arabic*:},leftmargin=3.5em]
\item Apply the exact conditional-projection decomposition to the stacked structural and first-stage score. Verify explicitly that the panel average of the centering terms is zero.
\item Compute the variance order of the unit projection, choose its normalization, and establish a spatial central limit theorem under stated mixing coefficients, moment exponents, and increasing-domain conditions.
\item Repeat the calculation for the time projection using a serial central limit theorem and the corresponding long-run variance.
\item Expand the interaction projection in a finite-rank orthonormal representation. Establish joint convergence of the finite collection of spatial and temporal partial sums and then map them into the finite-rank interaction limit.
\item Bound the infinite-rank remainder uniformly in $(N,T)$ using the stated square-summability and moment conditions.
\item Establish a conditional Lindeberg--Feller theorem for the cell innovation after conditioning on the unit and time states.
\item Derive all cross-covariances among the unit, time, interaction, and cell components. Independence of the primitive collections does not automatically make every projected component independent.
\item Combine the components by a converging-together theorem and verify that the resulting covariance rank is constant on the parameter stratum used for inference.
\end{enumerate}

This roadmap is not itself a proof for an unspecified array. For each application or Monte Carlo design, the supplement must state the exact mixing coefficients, variance orders, normalizing matrices, covariance-rank stratum, and tail bounds that imply Assumption~\ref{ass:primitiveclt}.

\subsection{Additional details for Theorems~\ref{thm:ratio} and \ref{thm:multiratio}}
The formal proofs appear above; this subsection records two additional details useful for replication. For the scalar statistic define
\[
 D_{NT}^{s}=b_{NT}^{-1}\{NT\pi_{NT}\widehat Q_{zz,NT}+S_{v,NT}\}.
\]
Then $(a_{NT}^{-1}S_{u,NT},D_{NT}^{s})\Rightarrow(Z_u,\mu_c+Z_v)$. The no-atom condition implies that, for every $\varepsilon>0$, one can choose $\eta>0$ such that
$\Pr(|\mu_c+Z_v|\le2\eta)<\varepsilon$. Portmanteau then gives
$\limsup\Pr(|D_{NT}^{s}|\le\eta)\le\varepsilon$. Thus the sample denominator is bounded away from zero with arbitrarily high limiting probability, which justifies taking the ratio without silently conditioning on a high-probability event.

For the multivariate estimator set $\mathsf A_{NT}=\mathsf G_{NT}C_{NT}^{-1}$ and
$\widetilde U_{NT}=a_{NT}^{-1}U_{NT}$. On the event that
$\mathsf A_{NT}'\mathsf W_{NT}\mathsf A_{NT}$ is nonsingular,
\[
 a_{NT}^{-1}C_{NT}(\widehat\beta-\beta_0)
 =(\mathsf A_{NT}'\mathsf W_{NT}\mathsf A_{NT})^{-1}\mathsf A_{NT}'\mathsf W_{NT}\widetilde U_{NT}.
\]
The smallest eigenvalue of the limiting Gram matrix is strictly positive almost surely. Hence, for any sequence $\eta_{NT}\downarrow0$ sufficiently slowly,
$\Pr\{\lambda_{\min}(\mathsf A_{NT}'\mathsf W_{NT}\mathsf A_{NT})>\eta_{NT}\}\to1$.
This makes explicit the high-probability domain on which inversion is valid. It also shows why the theorem is pointwise on a separated full-rank stratum and is not uniform as the limiting rank changes.

\subsection{Verification of factor-oracle transfer}\label{app:factorproof}

The algebraic identity
\[
 \widehat{\mathcal M}A-\mathcal M_0A
 =-\Delta_\Lambda AM_{F_0}-M_{\Lambda_0}A\Delta_F
   +\Delta_\Lambda A\Delta_F,
 \qquad
 \Delta_\Lambda=\widehat P_\Lambda-P_{\mathcal L_0},\quad
 \Delta_F=\widehat P_F-P_{\mathcal F_0},
\]
is exact. Expanding every feasible score and projected second moment with this identity
reduces the transfer claim to bounds on linear and quadratic projector-error terms.
For example, with Frobenius norms,
\[
 |\operatorname{tr}(Z'\Delta_\Lambda U\Delta_F)|
 \le \|Z\|_F\|U\|_F\|\Delta_\Lambda\|\|\Delta_F\|.
\]
For a generic feasible score term, insert
$\widehat{\mathcal M}A=\mathcal M_0A+R_A$, where
\[
 R_A=-\Delta_\Lambda AM_{\mathcal F_0}
     -M_{\mathcal L_0}A\Delta_F
     +\Delta_\Lambda A\Delta_F.
\]
At the matrix level, products entering a score difference split into two linear
remainders and one quadratic remainder:
$R_Z'U^o+Z^{o\prime}R_U+R_Z'R_U$ (with the appropriate vectorization and summation). Summing these terms and applying Cauchy--Schwarz,
the trace inequality above, and the projector-rate assumptions gives the normalized
bounds required in Assumption~\ref{ass:factororth}. The same expansion applies to
the first-stage score and the projected second moments. This verifies the three remainder bounds stated in Assumption~\ref{ass:factororth}.

This argument verifies how the general high-level remainder conditions transfer oracle
limits. Theorem~\ref{thm:pcbenchmark} supplies one primitive sample-split strong-factor
benchmark. Extending it to same-sample principal components under general two-way
dependence requires an explicit orthogonality expansion and sharper control of the
dependence between projector estimation and the evaluation score.

\subsection{Bootstrap transfer from oracle to feasible scores}\label{app:boottransfer}

Let $\mathcal T_{NT}^{*o}$ denote the oracle PWB-H statistic exactly as defined by
\citet{hounyolin2026}, and let $\widehat{\mathcal T}_{NT}^{*}$ be its feasible
factor-adjusted counterpart. Under the transfer conditions in
Theorem~\ref{thm:exactpwb}, the original and bootstrap score inputs, variance inputs,
and regime discriminants differ by $o_p(1)$ and $o_{p^*}(1)$ at the normalizations
used by PWB-H. On the maintained fixed-rank and fixed-regime strata, the map from
these inputs to the PWB-H statistic is continuous with probability approaching one.
Conditional bounded-Lipschitz equivalence and a triangle inequality then transfer
the oracle approximation to the feasible statistic. This argument does not establish
bootstrap validity for an arbitrary alternative studentization.

\end{appendices}

\section*{References}
References are listed in the main paper.
%\end{document}

\end{document}